\documentclass[pdflatex,sn-mathphys-num]{sn-jnl}

\usepackage[T1]{fontenc}
\usepackage[utf8]{inputenc}
\usepackage{verbatim}
\usepackage{tikz}
\usepackage{subcaption}
\usetikzlibrary{shapes.geometric}
\usetikzlibrary{quotes}
\usetikzlibrary{arrows.meta}

\usepackage{hyperref}
\usepackage[title]{appendix}
\theoremstyle{thmstyleone}
\usepackage{natbib}
\usepackage{bbm}
\usepackage{amsmath,amssymb,enumerate,ifthen,tikz,floatrow,amsmath,algorithmic}
\usepackage{mathtools}
\usepackage{graphicx}

\graphicspath{ {./images/} }
\newtheorem{theorem}{Theorem} 
\newtheorem{lemma}{Lemma} 
\newtheorem{claim}{Claim} % remove [lemma] to separate the numbers from lemmas

\newtheorem{corollary}{Corollary}

\newtheorem{definition}{Definition} % remove [lemma] to separate the numbers from lemmas
 
\newtheorem{observation}{Observation} % remove [lemma] to separate the numbers from lemmas
\newtheorem{insight}{Insight}

\usepackage{float}

\usepackage{hyperref}

\usepackage{hyperref}

\newcommand{\RR}{\mathbb{R}}

\newcommand{\eps}{\varepsilon}

\newcommand{\opt}{{\text{opt}}}
\newcommand{\OPT}{{\text{OPT}}}

\newcommand{\norm}[1]{|| #1 ||}
\newcommand{\abs}[1]{\left.|#1|\right.}
\newcommand{\tw}{\textwidth}

\newcommand{\boundellipse}[3]% center, xdim, ydim
{(#1) ellipse (#2 and #3)
}

\newcommand{\calI}{{\mathcal I}}

\newcommand{\calO}{{\mathcal O}}

\newcommand{\E}{\mathbb{E}}

\begin{document}

\author{\fnm{Benyamin} \sur{Ghaseminia}}\email{ghasemin@ualberta.ca}
% {Department of Computing Science, University of Alberta, Edmonton, Canada}

\author[1]{\fnm{Mohammad R.} \sur{Salavatipour}}\email{mrs@ualberta.ca}
% {https://orcid.org/0000-0002-7650-2045}
% {Department of Computing Science, University of Alberta, Edmonton, Canada}

\affil[]{\orgdiv{Department of Computing Science}, \orgname{University of Alberta}, \city{Edmonton}, \country{Canada}}
\affil[1]{Funded by NSERC}

% \authorrunning{B. Ghaseminia and M. Salavatipour}

% \Copyright{B. Ghaseminia and M. Salavatipour}

\title[A PTAS for TSP with Neighbourhoods Over Parallel Line Segments]{A PTAS for Travelling Salesman Problem with Neighbourhoods Over
Parallel Line Segments of Similar Length\footnote{An extended abstract of this paper is to appear in proceedings of SoCG 2025. Most of this work appeared in the M.Sc. thesis of the first author \cite{Ben}.}}
\abstract{
    We consider the Travelling Salesman Problem with Neighbourhoods (TSPN) on the Euclidean plane ($\RR^2$) and present a Polynomial-Time Approximation Scheme (PTAS) when the neighbourhoods are parallel line segments with lengths between $ [1, \lambda] $ for any constant value $ \lambda \ge 1 $.
    
    In TSPN (which generalizes classic TSP), each client represents a set (or neighbourhood) of points in a metric 
    and the goal is to find a minimum cost TSP tour that visits
    at least one point from each client set. In the Euclidean setting, each neighbourhood is a region on the plane.
    TSPN is significantly more difficult than classic TSP even in the Euclidean setting, as it captures group TSP.
    
    A notable case of TSPN is when each neighbourhood is a line segment. Although there are PTASs for when
    neighbourhoods are fat objects (with limited overlap), TSPN over line segments is \textbf{APX}-hard even if all
    the line segments have unit length. For parallel (unit) line segments, the best approximation factor is $3\sqrt2$ from more than two decades ago \cite{DM03}.
    
    The PTAS we present in this paper settles the approximability of this case of the problem. Our algorithm finds a $ (1 + \eps) $-factor approximation for an instance of the problem for $n$ segments with lengths in $ [1,\lambda] $ in time $ n^{O(\lambda/\eps^3)} $.   
}

\maketitle

\newpage
\tableofcontents
\newpage

\section{Introduction}
The Travelling Salesman Problem (TSP) is one of the most fundamental and well-studied problems in combinatorial optimization due to its wide range of applications. In TSP, one is given a set of points in a metric space and the goal is to find a (closed) tour (or walk) of minimum length visiting all the points. 
For several decades, the classic algorithm by Christofides \cite{Christofides} and independently by Serdyukov \cite{Serdyukov78} which implies a $\frac{3}{2}$-approximation was the best-known approximation for TSP until a recent result in \cite{KKG22} which shows a slight improvement.
	Several generalizations (or special cases) of TSP have been studied as well, the most notable is when the points are given in fixed dimension Euclidean space. Arora and Mitchell \cite{Arora98,mitchell1999guillotine} presented different PTASs for (fixed dimension) Euclidean TSP. There have been many papers that have extended these results. Arkin and Hassin \cite{ArkinH94} introduced the notion of TSP with neighbourhoods (TSPN). 

	An instance of TSPN
	is a set of neighbourhoods (or regions) given in a metric space, and the goal is to find a minimum length (or cost) tour that visits all these regions. Each region can be a single point or could be defined by a subset of points. They gave several $O(1)$-approximations for the geometric settings
	where each region is some well-defined shape on the plane, e.g. disks, and parallel unit length segments. Several papers have studied TSPN for various classes of neighbourhoods and under different metrics.
	
	TSPN is much more difficult than TSP in general and in special cases, just as group Steiner tree is much more difficult than Steiner tree (one can consider each neighbourhood as a group/set from which at least one point needs to be visited). In group Steiner tree or group TSP, one is given a metric along with groups of terminals (each group is a finite set). The goal is to find a minimum cost Steiner tree (or a tour) that contains (or visits) at least one terminal from each group. TSPN generalizes group TSP by allowing infinite size groups. Using the hardness result for group Steiner tree \cite{HK03}, it follows that general TSPN is hard to approximate within a factor better than $\Omega(\log^{2-\eps} n)$ for any $\eps>0$ even on tree metrics. The algorithms for group Steiner tree on trees in \cite{Garg98}, and embedding of metrics onto tree metrics in \cite{FRT04}, imply an $O(\log^3 n)$-approximation for TSPN in general metrics.  
	Unlike Euclidean TSP (which has a PTAS), TSPN is \textbf{APX}-hard on the Euclidean plane (i.e. $\RR^2$) \cite{Berg05}. 
	The special case when each region is an arbitrary finite set of points in the Euclidean plane (group TSP) has no constant
	approximation \cite{SafraS03} and the problem remains \textbf{APX}-hard even when each region consists of exactly two points \cite{DO08}.
	
	Focusing on Euclidean metrics, most of the earlier work have studied the cases where the regions (or objects) are {\em fat}. Roughly speaking, it usually means the ratio of the smallest enclosing circle to the largest circle fitting inside the object is bounded. There are some work on when regions are {\em not} fat, most notably when the regions are (infinite) lines or line segments or in higher dimensions when they  are hyperplanes. For the case of infinite line segments on $\RR^2$, the problem for $ n $ lines can be solved exactly in $O(n^4\log n)$ time by a reduction
	to the Shortest Watchman Route Problem (see \cite{DELM03,Jonsson02}). For the same setting,
	Dumitrescu and Mitchell \cite{DM03} presented a linear time $\frac{\pi}{2}$-approximation, which was improved to $\sqrt{2}$ by Jonsson \cite{Jonsson02} (again in linear time).
	For infinite lines in higher dimensions (i.e. $\geq 3$), the problem is proved to be \textbf{APX}-hard (see \cite{AKL022} and references there). For neighbourhoods being hyperplanes and dimension being $d\geq 3$, Dumitrescu and T\'{o}th \cite{DM16} present a constant factor approximation (which grows exponentially with $d$). For arbitrary $d$, they  present an $O(\log^3 n)$-approximation. For any fixed $d\geq 3$, authors of \cite{AF20} present a PTAS.
	
	For parallel (unit) line segments on the plane ($\RR^2$) Arkin and Hassin \cite{ArkinH94} presented a $(3\sqrt{2}+1)$-approximation which was improved to $3\sqrt{2}$ by \cite{DM03}, and this remains the best-known approximation for this case as far as we know for over two decades. Authors of \cite{EFS09} proved that TSPN for unit line segments (in arbitrary orientation) is \textbf{APX}-hard.
    In this paper, we settle the approximability of TSPN when regions are parallel line segments of similar length (unit length is a special case) and present a PTAS for it. 
    % As mentioned above, the best-known approximation for unit length parallel segments has ratio $3\sqrt{2}$ \cite{DM03}.
	We first focus on the case of unit line segments and show how our result extends to when line segments have bounded length ratio. This is in contrast with the \textbf{APX}-hardness of \cite{EFS09} for unit line segments with arbitrary orientation. Our result also implies a $(2+\eps)$-approximation
    for the case where the segments are axis-parallel (i.e. can be both vertical and horizontal).

%%%%%%%%%%%%%%%%%%%%%%%%%%%%%%%%%%%%%%%%%%%%%%%
\subsection{Related Work}
The work on TSPN is extensive, we list a subset of the  most notable and relevant work here and refer to the references of them for earlier works. All works listed are for $\RR^2$ metric.
Arkin and Hassin \cite{ArkinH94} presented constant factor approximations for several TSPN cases including when the neighbourhoods are parallel unit-length line segments (with ratio $3\sqrt{2}+1$). Very recently, PTASs were proposed for the case of unit disks and unit squares in \cite{soda25}.

Mata and Mitchell \cite{MM95} presented $O(\log n)$-approximation for general connected polygonal neighbourhoods.
Mitchell \cite{Mitchell07} gave a PTAS the case that regions are {\em disjoint} fat objects. This built upon his earlier work on PTAS for Euclidean TSP \cite{mitchell1999guillotine}. However, it is still open whether the problem is \textbf{APX}-hard for disjoint (general) shapes.

Dumitrescu and Mitchell \cite{DM03} presented several results, including $O(1)$-approximation for TSPN where objects are connected regions of the same or similar diameter. This implies $O(1)$-approximation for line segments of the same length in arbitrary orientation. They also 
improved the bounds for cases of parallel segments of equal length, translates of a convex region, and translates of a connected
region. For parallel (unit) line segments they obtain a $3\sqrt{2}$-approximation which remained the best ratio for this class until now.

Feremans and Grigoriev \cite{FG05}, gave a PTAS for the case that regions are disjoint fat polygons of similar size.
A similar result was obtained by \cite{BodlaenderFGPSW09} which gave a PTAS for TSPN for disjoint fat regions of similar sizes.
Mitchell \cite{Mitchell10} presented an $O(1)$-approximation for planar TSPN with pairwise-disjoint connected neighbourhoods of any size or shape (see also \cite{Berg05}). Subsequently, \cite{EFMS05} gave $O(1)$-approximation
for discrete setting where regions are overlapping, convex, and fat with similar size.

These results were further generalized by Chan and Elbassioni \cite{ChanE11} who presented a QPTAS for fat weakly disjoint objects in doubling dimensions (this allows limited amount of overlapping of the objects). This was further improved to a PTAS for the same setting \cite{ChanHJ16}. 
The TSPN problem is even harder if the neighbourhoods are disconnected. See the surveys of Mitchell \cite{Mitchell00,Mitchell04}. 
If the metric is planar and each group (or neighbourhood) is the set of vertices around a face, \cite{Bateni16} presents a PTAS for the group Steiner tree and a $(2+\eps)$-approximation for TSPN.

%%%%%%%%%%%%%%%%%%%%%%%%%%%%%%%%%%%%%%%%%%%%%%%%%%%%%%%%%%%%%%%%%%%%%%%%%%
\subsection{Our Results and Technique}
The main result of this paper is the following theorem.
\begin{theorem}\label{thm:main}
    Given a set of $n$ parallel (say vertical)
    line segments with lengths in $[1,\lambda]$ for a fixed $\lambda$
    as an instance of TSPN, there is an algorithm that finds a $(1+\eps)$-approximation solution in time $n^{O(\lambda/\eps^3)}$.
\end{theorem}

The algorithm we present is randomized but can be easily derandomized (like the PTAS for classic TSP).
To simplify the presentation, we give the proof for the case of unit line segments first.

This problem generalizes the classic (point) TSP (at a loss of $(1+\eps)$ factor). Given a point TSP instance, scale the plane so that 
the minimum distance between the points is at least $1/\eps$; call this instance $\calI$. Obtain instance $\calI'$ for line TSP by placing a vertical unit line segment over each point. It is easy to see that the optimum solutions of $\calI$ and $\calI'$ differ by at most an $\epsilon$ factor.

The difficult cases for line TSP are when the line segments are not too far apart (for e.g. they  can be packed in a box
of size $O(\sqrt{n})$ or smaller). There are two key ingredients to our proof that we explain here.
One may try to adapt the hierarchical decomposition by Arora \cite{Arora98} 
to this setting. Following that hierarchical decomposition,
the first issue is that some line segments might be crossing the horizontal dissecting lines, and so we don't have independent sub-instances, and it is not immediately clear in which subproblem these crossing segments must be covered by the tour. Note that 
the number of line segments crossing a dissecting line can be large.
Our first insight is the following:

\begin{insight}
    At a loss of $(1+\eps)$, we can drop the line segments crossing horizontal dissecting lines and instead requiring a subset of portals of each square (of the dissection) to be visited, provided 
    we continue the quad-tree decomposition until each square has size $\Theta(1/\eps)$.
\end{insight}

In other words, assuming all the squares in the decomposition have height at least $\Omega(1/\eps)$, then at a small loss we can show that a solution for the modified instance where line segments on the boundary of the squares are dropped, can be extended to a solution for the original instance. So, proving this property  allows us to work with the hierarchical (quad-tree) decomposition until squares of size $\Theta(1/\eps)$. This can be proved by a proper packing argument. But then we need to be able to solve instances where the height of the sub-problem instance is bounded by $O(1/\eps)$.

Let's define the notion of shadow of a solution (or in general, shadow of a collection of paths on the plane) as the maximum number of times a vertical sweeping line $\Gamma$ (that moves from left to right) intersects any of these paths. Our second insight is the following:

\begin{insight}
    If we consider a window that is a horizontal strip of height $O(h)$ and move this window 
    vertically anywhere over an optimum solution, then the shadow of the parts of optimum visible in this strip is at most $O(h)$.
\end{insight}

In other words, one expects that in the base case of the decomposition (where squares have height $\Theta(1/\eps)$), the shadow is bounded by $O(1/\eps)$. Despite our efforts, proving this appears to be more difficult than thought and it seems there are examples where, even in the unit length segments, the shadow may be large (see Figure \ref{fig:insight}). We were not able to prove this nor come up with an explicit counterexample.
However, we are able to prove the following slightly weaker version that still allows us to prove the final result:

\begin{figure}[h]
    \centering
    \includegraphics[width=0.3\textwidth]{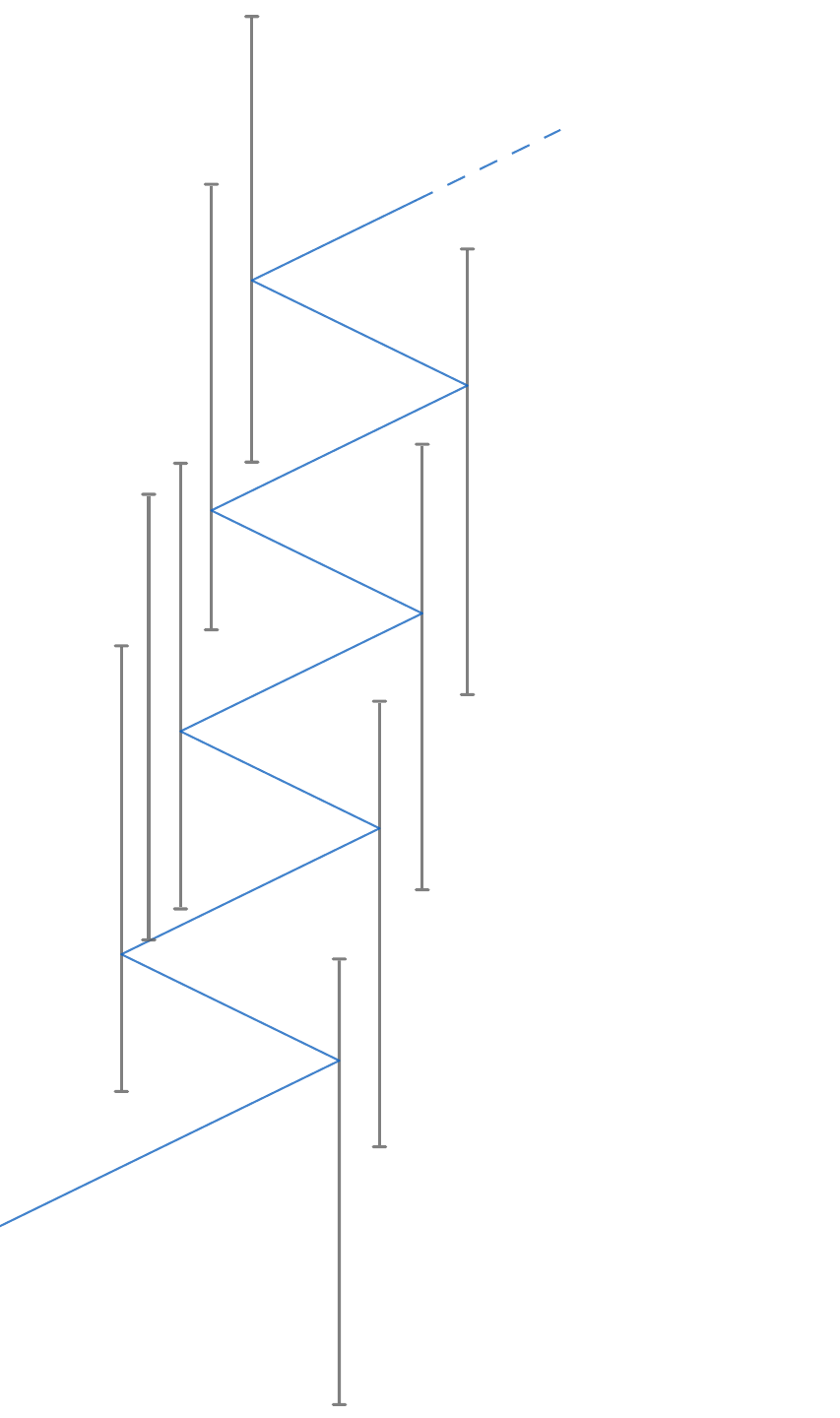}
    \caption{A potential arrangement of line segments where the solution has a large shadow}\label{fig:insight}
\end{figure}

\setcounter{insight}{1}
\begin{insight}[Revised]
    There is a $(1+\eps)$-approximate solution such that the shadow of any strip of height $h$ over that solution is bounded by $O(f(\eps)\cdot h)$ for some function $f(\cdot)$.
\end{insight}

The proof of this insight forms the bulk of our work. We characterize specific structures that would be responsible for having a large shadow in a solution and show how we can modify the solution to a $(1+\eps)$-approximate one with shadow $O(1/\eps)$.
Consider the unit line segments case, and suppose $\opt$ is the cost of an optimum solution.

\begin{theorem}\label{thm:onestrip-shadow}
    Given any $\eps>0$, there is a solution $\calO'$ of cost at most $(1+\eps)\cdot\opt$ such that in any  strip of height 1, the shadow of $\calO'$  is $O(1/\eps)$.
\end{theorem}

Proof of this theorem appears in Subsection \ref{sec:thm-onestrip-shadow}.
We will show that this near-optimum solution has in fact further structural properties that 
allow us to solve the bounded height cases at the base cases of the hierarchical decomposition using a Dynamic Program (DP) which, later on, is referred to as the inner DP. 
Proof of Theorem \ref{thm:onestrip-shadow} is fairly long and involves multiple steps that gradually proves structural properties for specific configurations. 

To get a very high level idea of the proof of this theorem, consider some fixed optimum 
solution $\OPT$. We decompose the problem into horizontal sections by drawing horizontal lines, called {\em cover-lines} that are 1-unit apart. The region of the plane between two consecutive cover-lines is called a {\em strip}. Note that each line segment crosses one cover-line (or might touch exactly two consecutive cover-lines). Let's consider one strip, say $S$, and consider the intersection of $\OPT$ with this strip. This intersection looks like a collection of paths that enter/exit this strip. 
We define the shadow of this strip similarly: consider a hypothetical vertical sweep line that moves left to right along the x-axis, the maximum number of intersections of this sweep line with these pieces of $\OPT$ restricted to $S$ is the shadow in $S$. We show that we can modify $\OPT$ to a near-optimum solution of cost at most $(1+\eps)$ times that of $\OPT$ so that the shadow in each strip is at most $O(1/\eps)$. To prove this, we show that there are certain potential structures that can cause $\OPT$ having a large shadow in a strip, one of which we call a zig-zag (see Figure \ref{fig:fig2}).
We show that we can modify $\OPT$ (at a small increase to the cost) so that the shadow becomes bounded along each zig-zag (or similar structures).

\subsubsection{Case Study: Restricted Special Instances}

To get an intuition of our approach, consider a special case where the given instance of the problem fits within a bounding box of height at most $ 3 $. 
If we draw a horizontal cover-line going through the bottom-most point of the top-most segment, and draw a second cover-line 1 unit below the previous one, then one can see that the entire optimum solution should fit within this strip of height 1 between these two cover-lines.
Every segment of the instance intersects with at least one of the two horizontal lines. We call the segments intersecting the top horizontal line {\em top segments}, and the rest {\em bottom segments}.
Any optimum solution can be partitioned into two subpaths: One subpath from the left-most point on the optimum solution towards the right-most point on the solution (call $P_1$), and one subpath the opposite way (call $P_2$), where $P_1$ is always "above" $P_2$
and is responsible of covering the top segments, and $P_2$ is responsible for covering the bottom segments. In this case the shadow of the solution is always 2.
This is similar to the classic bitonic TSP tour which can be easily solved using a dynamic program (where there is a sweeping line moving left to right and keeps track of the end-points of these paths $P_1$ and $P_2$).

\begin{figure}[h]
    \centering
    \includegraphics[width=.45\linewidth]{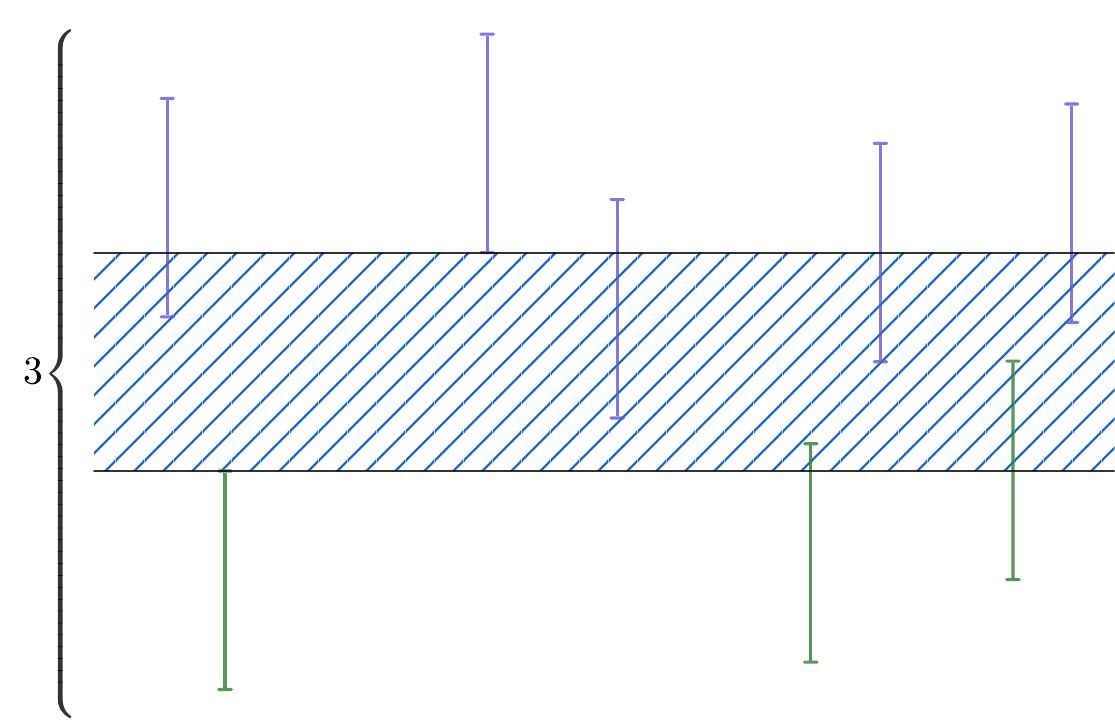}
    %\caption{The strip, and top/bottom segments}
    \hfill
    \includegraphics[width=.45\linewidth]{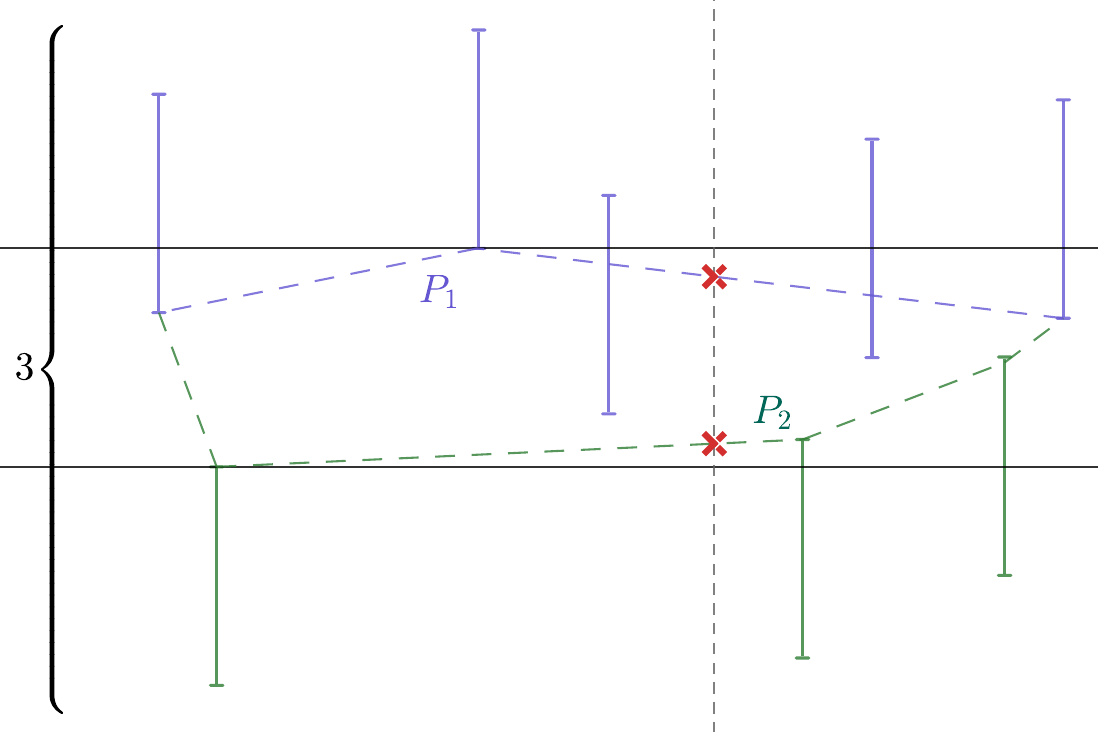}\caption{An optimum solution to these special case instances intersects with any vertical line at most twice}
    \label{fig:height3-shadow2}
\end{figure}

When the height of the bounding box is larger than 3, we still draw
horizontal cover-lines that are 1-unit apart.
In Figure \ref{fig:height5} you can see an instance with a bounding box of height 5 and a solution for it.

\begin{figure}[h]
    \centering
    \includegraphics[width=.5\tw]{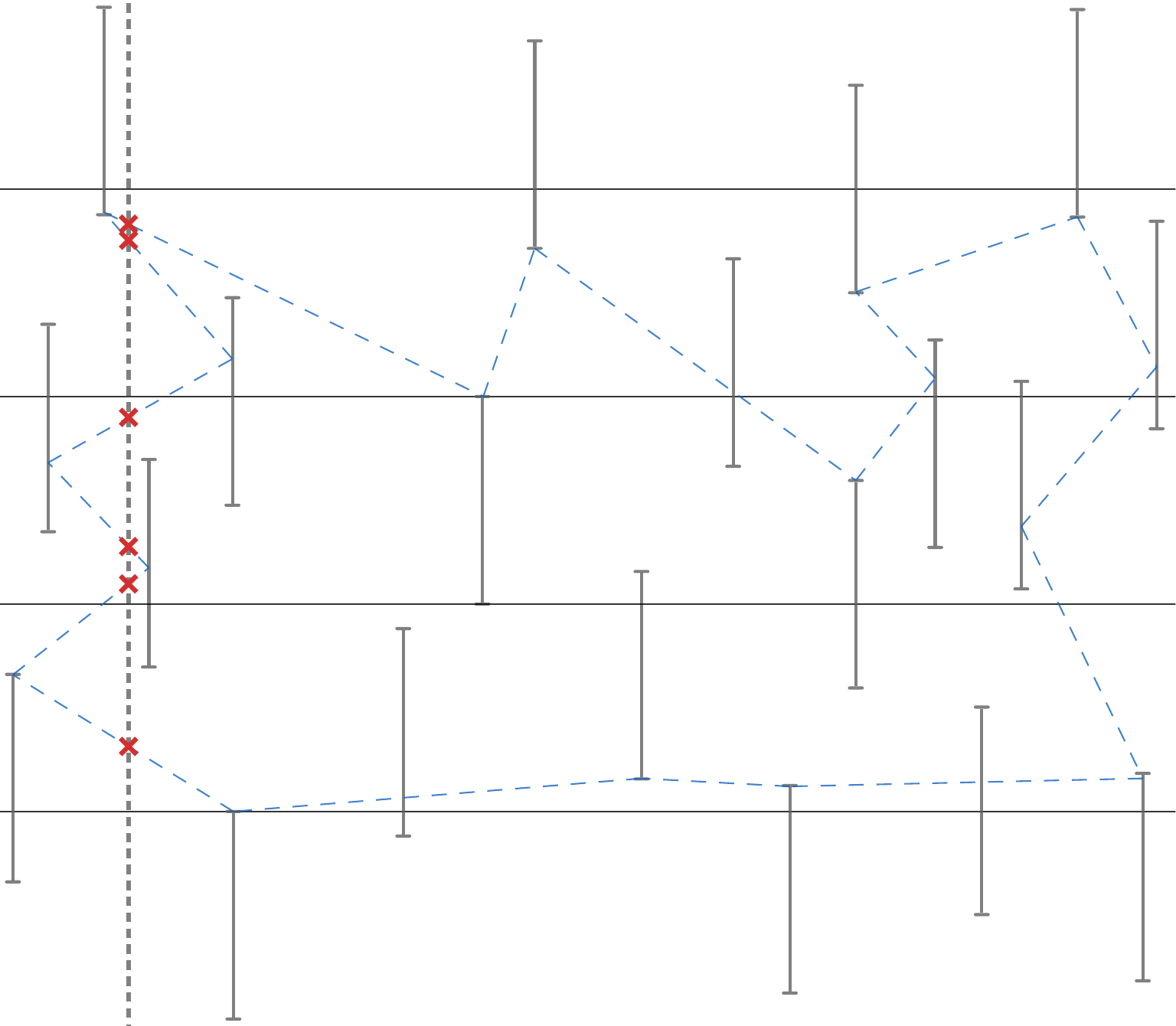}
    \caption{An example of an instance of the problem with bounding box of height at most 5. We have at most 3 strips in this case and maximum shadow of the given solution is 6}
    \label{fig:height5}
\end{figure}

If we can prove that the overall shadow of an optimum solution is bounded in every strip, then we can create an extension of the bitonic tour solution to obtain a dynamic program with low complexity (i.e. constant number of subpaths), to then find an exact solution by combining the subpaths we find.

%\begin{figure}[h]
%    \centering
%    \includegraphics[width=.65\tw]{new-images/e12(vertical).pdf}
%    \caption{A possible solution to the case with a bounding box of height 5. Is the shadow bounded by a function of the number of strips?}
%    \label{fig:height5-sol}
%\end{figure}

Unfortunately, it might be the case that the shadow of an optimum solution is much larger than a function of the height of the bounding box (i.e. the number of strips). Such configurations like Figure \ref{fig:insight} might happen, leading to large shadows. We will show that such configurations can be changed to another solution with a bounded shadow at a small increase in cost.

\subsubsection*{Organization of the paper:} 

We start by some preliminaries in the next subsection. 
In Section \ref{sec:properties}, we first list the properties of an optimum solution (in Subsection \ref{sec:opt-prop}), then the properties of a structured near-optimum solution (in Subsection \ref{sec:near-optimum-property}). In these two sections we list four major lemmas (Lemmas \ref{lemma:zig-zag} to \ref{lemma:overlapping-loops-ladders}) that are used
to prove Theorem \ref{thm:onestrip-shadow}.
We then prove Theorem \ref{thm:onestrip-shadow} in Subsection \ref{sec:thm-onestrip-shadow}. 

We describe the main algorithm in Section \ref{sec:main-alg},
which includes the outer DP and inner DP. The proof of our main theorem appears in Subsection \ref{sec:main-theo-proof}.

Section \ref{sec:proofs} has details of the proofs of the major lemmas (for properties of a near-optimum solution), as well as the proof of Theorem \ref{thm:exact-answer}.
Finally, we conclude our results in Section \ref{sec:conclusion} and discuss further extensions.

%%%%%%%%%%%%%%%%%%%%%%%%%%%%%%%%%%%%%%%%%%%%%%%%%%%%%%%
\subsection{Preliminaries}\label{sec:prem}
Suppose we are given $n$ vertical line segments $s_1,\ldots,s_n$ of lengths in the range $[1,\lambda]$ for some constant $\lambda \ge 1$, and the top and bottom points of each $s_i$ are denoted by $s^t_i$ and $s^b_i$, respectively. These end-points are also called {\em tips} of the segment. For any point $p$, let $ x(p) $ and $ y(p) $ denote the $x$ and $y$-coordinates of $p$, respectively. Similarly, for any segment or vertical line $s$, let $ x(s) $ denote its $x$-coordinate.
For two points $p,q$, we use $||pq||$ to denote the Euclidean distance between them.
A feasible (TSP) tour is specified by a sequence of points where each of these points is on one of the segments of the instance and each line segment has at least one such point, and the tour visits these points consecutively using straight lines. The line that connects two consecutive points in a tour is called a {\em leg} of the tour.
In our problem, the goal is to find a TSP tour of minimum total length. As mentioned earlier, we focus on the case where all the line segments have length 1 and then show how the proof easily extends to the setting where they have lengths in $[1,\lambda]$.
Fix an optimum solution, which we refer to by $\OPT$ and use
$\opt$ to refer to its cost. Our goal is to show the existence of a near-optimum (i.e. $(1+\eps)$-approximate) structured solution that allows us to find it using dynamic programming.

First, we show at a small loss we can assume all the line segments have different $x$-coordinates.
We assume that the minimal bounding box of these line segments has length $L$ and height $H$. For now, assume $ H > 3 $ (case of $H\leq 3$ is easier, see Theorem \ref{thm:exact-answer}). Let $ B = \max \{L, H-2\} $. So $\opt\geq 2B$; we can also assume $B\leq \frac{n}{\eps}$, because otherwise $\opt\geq 2n/\eps$ and if we consider an arbitrary point on each line segment (say the lower tip) and use a PTAS for the classic TSP for these points, then it will be a PTAS for our original instance as well (because we pay at most an extra +2 for each line for a total of $2n$ which is $O(\eps\cdot\opt)$).
For a given $ \epsilon > 0 $, consider a grid on the plane with side length $ \frac{\epsilon B}{n^2}$. Now move each line segment (parallel to the $y$-axis) so that the lower tip of each $s_i $ is moved to the nearest grid point where there is no other line segment $s_i$ with that $x$-coordinate. By doing this, all the segments will have different $x$-coordinates and each segment would move at most $\frac{\sqrt 2}2\cdot\frac{\eps B}n < \frac{\epsilon B}n $, and in total, all segments would move at most a distance of $\epsilon B$. So the optimum value of the new instance has cost at most $(1+\eps)\cdot\opt$.
For simplicity of notations, from now on we assume the original instance has this property and let $\OPT$ (and $\opt$) refer to an optimum (and its value) of this modified instance.

%%%%%%%%%%%%%%%%%%%%%%%%%%%%%%%%%%%%%%%%%%%%%%%%%%%%%%%%%%%%%%%%%5
\section{Properties of a Structured Near-Optimum Solution}\label{sec:properties}
We start by presenting some properties of an optimum solution and then some structural properties for a near 
optimum that allow us  to  prove Theorem \ref{thm:onestrip-shadow}.
\subsection{Properties of an Optimum Solution}\label{sec:opt-prop}

We start by stating some lemmas that give a better understanding of the geometrical properties of an optimum solution, and later build up the proof one of the major lemmas in Subsection \ref{sec:lemma-zigzag} 
(Lemma \ref{lemma:zig-zag}) from these properties.

One special instance of the problem is when there is a single horizontal line that crosses all the input segments.
This special case can be detected and solved easily. Otherwise, any optimum solution will visit at least
3 points that are not colinear. In such cases, like in the classic (point) TSP \cite{Arora98}, we can assume the optimum
does not cross itself, i.e. there are no two legs of the optimum $\ell$ (between points $p,q$) and $\ell'$ (between points $p',q'$) that intersect, as otherwise removing these two and adding the pair of $pq',p'q$ or $pp',qq'$ will be a feasible solution of smaller cost.

\begin{definition}\label{def:shadow}
    Given a collection $\mathcal  P$ of paths on the plane and a vertical line at point $x_0\in\RR$, the \textbf{shadow} at $x_0$ is the number of legs of the paths in $\mathcal  P$
    that have an intersection with the vertical line at $x_0$. The shadow of a given range $ [a, b] $ is defined to be the maximum shadow of values $ x_0 \in [a, b]$.
\end{definition}

Note that if a solution is self-crossing, the operation of uncrossing (which reduces the cost) does not increase the shadow

Suppose the sequence of points of $\OPT$ is $p_1,p_2,\ldots,p_\sigma$ and the straight lines connecting these points (i.e. legs of $\OPT$) are $\ell_1,\ell_2,\ldots,\ell_\sigma$
where $\ell_i$ connects two points $p_i,p_{i+1}$ (with $p_{\sigma+1}=p_1$), and each $s_i$ has at least one point $p_j$ on it. We consider $\OPT$ oriented in this order, i.e. going from $p_i$ to $p_{i+1}$. 
Since all segments have distinct $x$-coordinates, we can assume
%\begin{observation}
	no two consecutive points $p_i,p_{i+1}$ are on the same line segment by short-cutting (so no leg $\ell_i$ is vertical) and all points $p_i$ on distinct line segments have distinct $x$-coordinates.
%\end{observation}

As mentioned before, let the length of the sides of the minimal bounding box of an instance of the problem be $ L\times H $. 
The proof of the following theorem, which has the same setting as our toy example in the introduction, appears in Subsection \ref{sec:exact-answer}, where we use some of the definitions and lemmas given throughout this section. 
\begin{theorem}\label{thm:exact-answer}
    If $ H \le 3 $, then the shadow of an optimum solution is at most 2.
\end{theorem}

By proving this special case, we show we can instead focus on the cases that $H$ is large. So from now on, we assume that $ H > 3 $. Our first main goal is to prove Theorem \ref{thm:onestrip-shadow}. 

\begin{definition}\label{def:right-left}
    Given a segment $s$ and a leg $\ell$ incident to a point on $s$, we say
	$\ell$ is to the \textbf{left} of $s$ if $ \ell $ is entirely in the subplane $x \leq x(s)$; and
	$ \ell $ is to the \textbf{right} of $s$ if $ \ell $ is entirely in the subplane $x \ge x(s)$. 
\end{definition}
Since there are no vertical legs, there is no leg that is both to the left and to the right of a segment of the instance at the same time.
Consider any segment $s_i$  and suppose that $\ell_{j},\ell_{j+1}$ are the two legs of $\OPT$ with common end-point $p_j$ that is on $s_i$. 
Let $ s_i^t $ and $ s_i^b $ denote the top and the bottom tips of $s_i$. 
We consider 3 possible cases for the location of $p_j$ and the arrangement of $\ell_j,\ell_{j+1}$. Informally, one possibility is that the two legs $\ell_{j},\ell_{j+1}$ form a straight line that crosses $s_i$ at $p_j$; one possibility is that the two legs are touching $s_i$ at one of its tips (i.e. $p_j=s_i^t$ or $p_j=s_i^b$) such that one is to the left and one is to the right of $s_i$ and they don't make a straight line, and the third possibility is that the two legs $\ell_j,\ell_{j+1}$ are on the same side (both left or both right) of $s_i$.

\begin{observation}\label{obs:legs}
 Consider any segment $s_i$ (with top/bottom points $s_i^t$,$s_i^b$) and suppose that $\ell_{j},\ell_{j+1}$ are the two legs of $\OPT$ with common end-point $p_j$ that is on $s_i$. Then either:
	\begin{itemize}
		\item The subpath of $\OPT$ going through $p_{j-1},p_{j},p_{j+1}$ forms a straight line (i.e. $ \angle \ell_j p_j \ell_{j+1} = \pi $), and $\ell_j, \ell_{j+1}$ are on two sides (left/right) of $s_i$; then we call $p_j$ a \textbf{straight point}, or
	    \item $p_j$ is a tip of $s_i$ (i.e. $p_j=s_i^t$ or $p_j=s_i^b$), 
		$\angle \ell_j p_j \ell_{j+1} \ne \pi $ and $\ell_j$ and $\ell_{j+1}$ are on two sides of $s_i$ (one left and one right); in this case $p_j$ is called a \textbf{break point}, or
        \item both $\ell_j,\ell_{j+1}$ are on the left or on the right of $s_i$; in this case $p_j$ is called a \textbf{reflection point}.
	\end{itemize}
\end{observation}

For the case of a reflection point $p_j$ with two legs $\ell_j,\ell_{j+1}$, if both
legs are to the left of the segment it is called a {\em left reflection} point and otherwise it is a {\em right reflection} point.
Also note that if $\ell_j,\ell_{j+1}$ are on the two sides of $s_i$ and $\angle \ell_i p_j \ell_{i+1} \not= \pi $, then $p_j$ must be a tip, or else we could move $p_j$ slightly up or down and reduce the length of $\OPT$ (see Figure \ref{fig:straight-line}).
\begin{figure}[H]
    \centering
    \includegraphics[width=.29\textwidth]{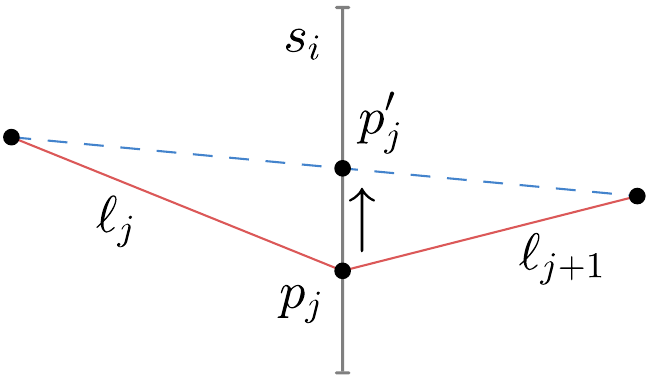}
    \caption{If $p_j$ isn't a tip of $s_i$, then $\ell_j, \ell_{j + 1}$ must be collinear}
    \label{fig:straight-line}
\end{figure}

\begin{lemma}\label{obs:further-point-reflection}
    If $P$ is a subpath of $\OPT$ with end-points $p,q$ where both are to the right of a vertical line $\Gamma$, and if $P$ crosses $\Gamma$, then
    the left-most point on $P$ to the left of $\Gamma$ is a right reflection point (symmetric statement holds for opposite directions).
\end{lemma}

\begin{proof}
    Let $r$ (on segment $s$) be the left-most point $P$ visits, so both subpaths $P_{pr}=p\to r$ and $P_{qr}=q\to r$ are entirely to the right of $r$, in particular the two legs $\ell^-$ and $\ell^+$ of $P$ incident to $r$ (which are the last two legs of the subpaths $P_{pr},P_{qr}$) must be on the right of $s$ which implies that $r$ is a right reflection point.
\end{proof}

\begin{definition}\label{def:ascending}
    Consider an arbitrary reflection point $r$ on a segment $s$. Let the two legs of $\OPT$ incident to $r$ visited before and after $r$ (on the orientation of $\OPT$) be $ \ell^- $ and $ \ell^+ $, respectively.
    $\ell^-$ is said to be \textbf{on top of} $\ell^+$  if all the points of $\ell^-$ have larger $y$-coordinate than all of the points of $\ell^+$. 
    In this case we also call $\ell^-$ the \textbf{upper leg} and $\ell^+$ the \textbf{lower leg}.
    Also, in this case, $r$ is called a \textbf{descending} reflection point.
    If $\ell^+ $ is on top of $\ell^- $, then $r$ is called an \textbf{ascending} reflection point.
\end{definition}

\begin{definition}
    If $\ell_j,\ell_{j+1}$ are two legs incident to a reflection point $p$ on a segment $s$, if 
    the angle between $\ell_j$ and $s$ is the same as the angle between $\ell_{j+1}$ and $s$ (i.e. $\ell_{j+1}$ is like the reflection of a ray $\ell_j$ on mirror $s$) then $p$ is called a \textbf{pure reflection} point.
\end{definition}

\begin{lemma}\label{obs:break-on-tip}
    Any reflection point that is not a tip of a segment is a pure reflection point.
\end{lemma}

\begin{proof}
    Suppose $p_j$ is a reflection point on $s_i$ and is not a tip of it. If the two legs $\ell_j,\ell_{j+1}$ don't have the same angle with $s_i$, then we
    can move $p_j$ along $s_i$ slightly up or down and one of the moves will decrease the cost of $\OPT$, a contradiction.
\end{proof}

We say a subpath contains a reflection point $p_j$ if $p_j$ is not the start or end vertex of the subpath (i.e. both legs of incident to $p_j$ belong to that subpath.)

\begin{lemma}\label{obs:shadow-increase}
    If a sweeping vertical line $\Gamma$ moves left to right on the $ x $-axis, the only values of $x$ for which the shadow at $\Gamma$ changes will be when $\Gamma$ hits a reflection point on that $x$-coordinate. Specifically, this means that any subpath of $\OPT$ that doesn't contain a reflection point, must have a shadow of 1 throughout its length. 
\end{lemma}

\begin{proof}
    According to Observation \ref{obs:legs}, we can see that straight points or break points will always contribute 1 to the shadow of $\Gamma$. But reflection points, depending on which direction the sweeping line moves, will either increase or decrease the shadow by 2.
    If a path doesn't contain any reflections, it means that it can only contain straight points or break points, meaning its shadow throughout its length will be equal to 1.
\end{proof}	

\begin{definition}\label{def:path-above-below}
    Let $P_1$ and $P_2$ be any two subpaths of $\OPT$. We say $P_1$ is \textbf{above} $P_2$ in range $I = [x_0, x_1]$ if for every vertical line $\Gamma$ with $x(\Gamma) \in I$, the top-most intersection of
    $\Gamma$ with these two paths is a point on $P_1$. We say $P_2$ is \textbf{below} $P_1$ if the bottom-most intersection of $\Gamma$ with $P_1, P_2$ is a point on $P_2$.
    Similarly, we say $L_1$ is to the \textbf{left} of $L_2$ in range $I' = [y_0, y_1]$ if for every horizontal line $\Lambda$ with $y(\Lambda) \in I'$, the left-most intersection point of $\Lambda$ with $L_1,L_2$ (i.e. one with the least $x$ value) always belongs to $L_1$. We say $L_2$ is to the \textbf{right} of $L_1$ if the right-most intersection of $\Lambda$ is with $L_2$.
\end{definition}
\begin{figure}[H]
    \centering
    \includegraphics[width=.4\textwidth]{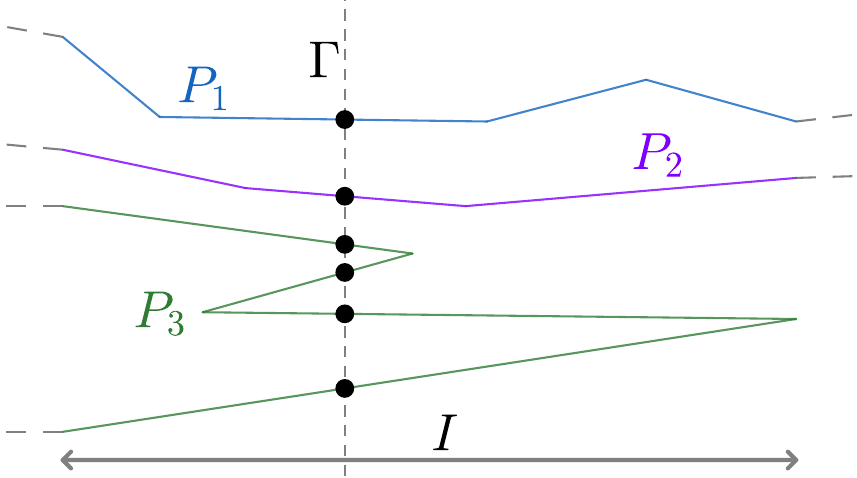}
    \caption{
    % (Definition \ref{def:path-above-below}) 
    In range $I$, $P_1$ is above $P_2, P_3$, and $P_2$ is above $P_3$}
    \label{fig:above-below-paths}
\end{figure}

\begin{lemma}\label{obs:change-in-x} 
For any distinct points $p_j$ and $p_{j'}$ on $\OPT$, following $\OPT$ according to its orientation, either the path from $p_j$ to $p_{j'}$ or the path from $p_{j'}$ to $p_j$ must contain at least one reflection point. 
\end{lemma}

\begin{proof}
    Without loss of generality, assume $ x(p_j) < x(p_{j'})$, and following the orientation of $\OPT$ starting from $p_j$, suppose the path from $p_j$ to $p_{j'}$  does not contain any reflection points (or the statement of lemma holds). According to Observation \ref{obs:legs}, the $x$-coordinate of points on 
    $\OPT$ will not decrease if and only if the path contains only straight points or break points. The path from $p_{j'}$ to $p_j$ has to have a decrease in the $ x $-coordinate, due to $ x(p_{j'}) > x(p_j)$, which is only possible if there is a reflection in this part of the path.
\end{proof}

\begin{lemma}\label{obs:reflection-paths-above-below}
    Let $r_j$ be any reflection point on $\OPT$, say it is a right reflection point, with incident legs $\ell_i, \ell_{i + 1}$. Without loss of generality, assume that $\ell_i$ is above $\ell_{i + 1}$. Take any two subpaths $P_1$ and $P_2$ of $\OPT$ both starting at $r_j$ with shadow of $1$ such that $\ell_{i}\in P_1$ and $\ell_{i + 1}\in P_2$. 
    If there is a vertical line $\Gamma$ with $x(\Gamma) > x(r_j)$ crossing both $P_1$ and $P_2$, then $P_1$ will be above $P_2$ in range $I = [x(r_j), x(\Gamma)]$.
\end{lemma}

\begin{proof}
    Note that for any vertical line $\Gamma'$ with $x(\Gamma') \in I$, both $P_1$ and $P_2$ will intersect with it.
    Now assume the contrary, that $P_1$ is not above $P_2$. This means for some vertical line $\Gamma'$ with $x(\Gamma') \in I$, there are points $p_1$ and $p_2$ on $\Gamma'$ such that $p_1\in P_1$, $p_2\in P_2$, and $y(p_2) > y(p_1)$. Since both $P_1$ and $P_2$ have a shadow of $1$, then using Lemma \ref{obs:shadow-increase}, we get that neither of them have a reflection point; this implies that the value of the $x$-coordinate on both $P_1$ and $P_2$ is monotone (or else there must be a reflection point). 
    Since $P_1$ travels from $r_j$ to $p_1$ and $P_2$ travels from $r_j$ to $p_2$, both are crossing the same vertical lines (at $x=x(r_j)$ and $\Gamma'$). 
    Now, because $\ell_i$ is above $\ell_{i + 1}$ but $p_1$ is below $p_2$, we conclude that $P_1$ and $P_2$ will intersect with each other in the area between the vertical lines $\Gamma'$ and $x = x(r_j)$. This is a contradiction, hence the lemma.
\end{proof}

\begin{lemma}\label{obs:consecutive-reflections}
    Among the set of points visited by $\OPT$ following its orientation, suppose $p_{j},p_{j'}$, $j < j'$ (on segments $s_i$, $s_{i'}$, respectively) are two consecutive reflection points (i.e. no other reflection point exists in between them). Then $p_j$ and 
    $p_{j'}$ cannot be both left or both right reflection points. Furthermore, if $s_i$ is to the left of $s_{i'}$ then $p_j$ is a
    right reflection and $p_{j'}$ is a left reflection (the opposite holds if $s_{i'}$ is to the left of $s_i$).
\end{lemma}

\begin{proof}
    Without loss of generality, assume that $s_i$ is to the left of $s_{i'}$, meaning $x(p_j) < x(p_{j'})$.
    By way of contradiction, first suppose both $p_j$ and $p_{j'}$ are right reflection points, i.e. the two legs incident to $p_j$ ($\ell_j,\ell_{j+1}$)
    and the two legs incident to $p_{j'}$ ($\ell_{j'},\ell_{j'+1}$) are on the right of $s_i$ and right of $s_{i'}$, respectively.
    This means following the orientation on $\OPT$, along $\ell_j$ we have a decrease in $x$-coordinate, then
    following $\ell_{j+1}$ have an increase, then again following $\ell_{j'}$ have a decrease and following $\ell_{j'+1}$ have an increase. 
    So the value of the $x$-coordinate isn't monotone in the subpath of $\OPT$ from $p_j$ to $p_{j'}$ (excluding these two points themselves), because the legs $\ell_{j + 1}$ and $\ell_{j'}$ are visited in this path in this order. Similar to the proof in Lemma \ref{obs:change-in-x}, we see that this is only possible if there is a reflection point on this subpath, which contradicts the assumption that $p_j,p_{j'}$ are consecutive.
    Similar argument implies that we cannot have both $p_j, p_{j'}$ being left reflections
    or $p_j$ being a left reflection and $p_{j'}$ being a right reflection; otherwise, the leg after visiting $p_j$ will have decreasing $x$-value while it will have to visit $p_{j'}$ eventually, which has a larger $x$-value. So the path from $p_j$ to $p_{j'}$ must include another reflection point, again a contradiction.
\end{proof}

\begin{corollary}\label{cor:alter}
    Consecutive reflection points in $\OPT$ alternate between left reflections and right reflections.
\end{corollary}

\begin{lemma}\label{obs:reflection}
If a segment $s_i$ has a reflection point $p_j$, then it cannot have any other intersections with $\OPT$ (i.e. no other point $p_j'$ can be on $s_i$).
\end{lemma}

\begin{proof}
    Assume otherwise, that a segment $s_i$ contains a reflection point $p_j$ with legs $ \ell_j$ and $\ell_{j+1}$, and another point $p_{j'}$ on $s_i$. We can by-pass $p_j$ locally and reduce the length of $\OPT$ which would be a contradiction.
    More specifically,
    let $ R^-\in \ell_j $ and $ R^+\in \ell_{j+1}$ be points on the legs that have a distance of $\delta>0$ from $p_j$.
    By replacing the subpath $R^-\to p_j\to R^+ $ with $ R^-\to R^+ $, the total cost of $\OPT $ will decrease, which gives us a contradiction. 
\end{proof}
\begin{figure}[H]
    \centering
    \includegraphics[width=.2\textwidth]{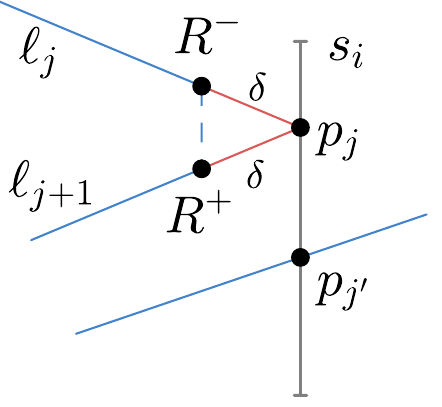}
    \caption{There can't be another $p_{j'}\in s_i$ if $p_j\in s_i$ is a reflection}
    \label{fig:reflection-shortcut}
\end{figure}

%%%%%%%%%%%%%%%%%%%%%%%%%%%%%%%%%%%%%%%%%%%%%%%%%
% \subsection{Strips, Zig-zag, Sink}
\label{sec:bounded-shadow}
We decompose the problem into horizontal {\em strips} by drawing some parallel horizontal lines, which we call {\em cover-lines}.
Starting from the bottom tip of the top-most segment, draw parallel horizontal lines that are 1-unit apart, these are our {\em cover-lines}. Each input segment is considered "covered" by the top-most (i.e. the first) cover-line that intersects it. Let's label these cover-lines by $C_1,C_2,\ldots$.

\begin{definition}[strip, top/bottom segments]\label{def:strip/top-bottom}
	The region of the plane between two consecutive cover-lines $C_\tau,C_{\tau+1}$ is called a \textbf{strip} and denoted by $S_\tau$. 
	We consider $C_\tau,C_{\tau+1}$ part of $S_\tau$ as well.
	The input line segments that are intersecting the top cover-line of $S_\tau$ ($C_\tau$) are called \textbf{top segments,} and the segments covered by the bottom cover-line ($C_{\tau + 1}$) are called \textbf{bottom segments} of the strip.  
\end{definition}

We show the near-optimum solution guaranteed by Theorem \ref{thm:onestrip-shadow} has more structural properties that will be defined later. Note that once we prove that theorem, it follows that if we restrict a solution to $h>1$ many strips, then the shadow is bounded by $O(h/\eps)$ as well.

For now, let us focus on an (arbitrary) strip $S_\tau$ and imagine we cut the plane along $C_\tau,C_{\tau+1}$ and look at the pieces of line segments of the instance left inside this strip, along with pieces of $\OPT$ inside $S_\tau$. 
Each top segment is now a partial segment in $S_\tau$ that has one end on $C_\tau$; and each bottom segment has one end on $C_{\tau+1}$.
Let $\OPT_\tau$ be the restriction of $\OPT$ to $S_\tau$. 
For each leg of $\OPT$ that intersects $C_\tau$ or $C_{\tau+1}$, we add a dummy point at the intersection(s) of that leg with $C_\tau$ and $C_{\tau+1}$ (so that the components of $ \OPT_\tau $ become consistent with our definition of legs).
So $\OPT_\tau$ can be seen as a collection of subpaths within $S_\tau$ (possibly along $C_\tau$ or $C_{\tau+1}$). Following the orientation of $\OPT$, each subpath of $\OPT_\tau$ is formed by it intersecting with $S_\tau$, traveling within $S_\tau$ (possibly along one of the cover-lines), and then exiting $S_\tau$. Using the dummy points added, each path in $\OPT_\tau$ is a subpath of $\OPT$ that is between two points on cover-lines (these are called the entry points of the path with the strip. A formal definition is provided below).

Recall Definition \ref{def:path-above-below} of paths being above or below each other. Having the definition of top/bottom segments, we get the following:

\begin{observation}\label{cor:path-above-top-segment}
    Consider $\OPT_\tau$, the restriction of $\OPT$ to any strip $S_\tau$. Take any two subpaths of $\OPT_\tau$ like $P_1$ and $P_2$ such that $P_1$ is above $P_2$ in some range $I$. If $s_t$ is any top segment in range $I$ that $P_2$ intersects with, then $P_1$ will also intersect with it. Similar statement holds for bottom segments if $P_2$ is below $P_1$.
\end{observation}

\begin{definition}[entry points, loops, ladders]\label{def:loops/ladders}
	For each subpath $P_j$ of $\OPT_\tau$, let $e_j$ and $o_j$ be the first and last intersections of $P_j$ with the interior of $S_\tau$. Points $e_j$ and $o_j$ are called the \textbf{entry points} of $P_j$.\\
	If both $e_j$ and $o_j$ lie on the same cover-line (either $C_\tau$ or $C_{\tau+ 1}$), then $P_j$ is called a \textbf{loop}, otherwise it's called a \textbf{ladder}.
	If a subpath of $\OPT_\tau$ enters $S_\tau$ at $e_j$ on a cover-line and follows on that cover-line to point $o_j$ and exits the strip, it is a special case of loop that we refer to as a \textbf{cover-line loop}.
	
\end{definition}

\begin{figure}[H]
    \centering
    \includegraphics[width=.7\textwidth]{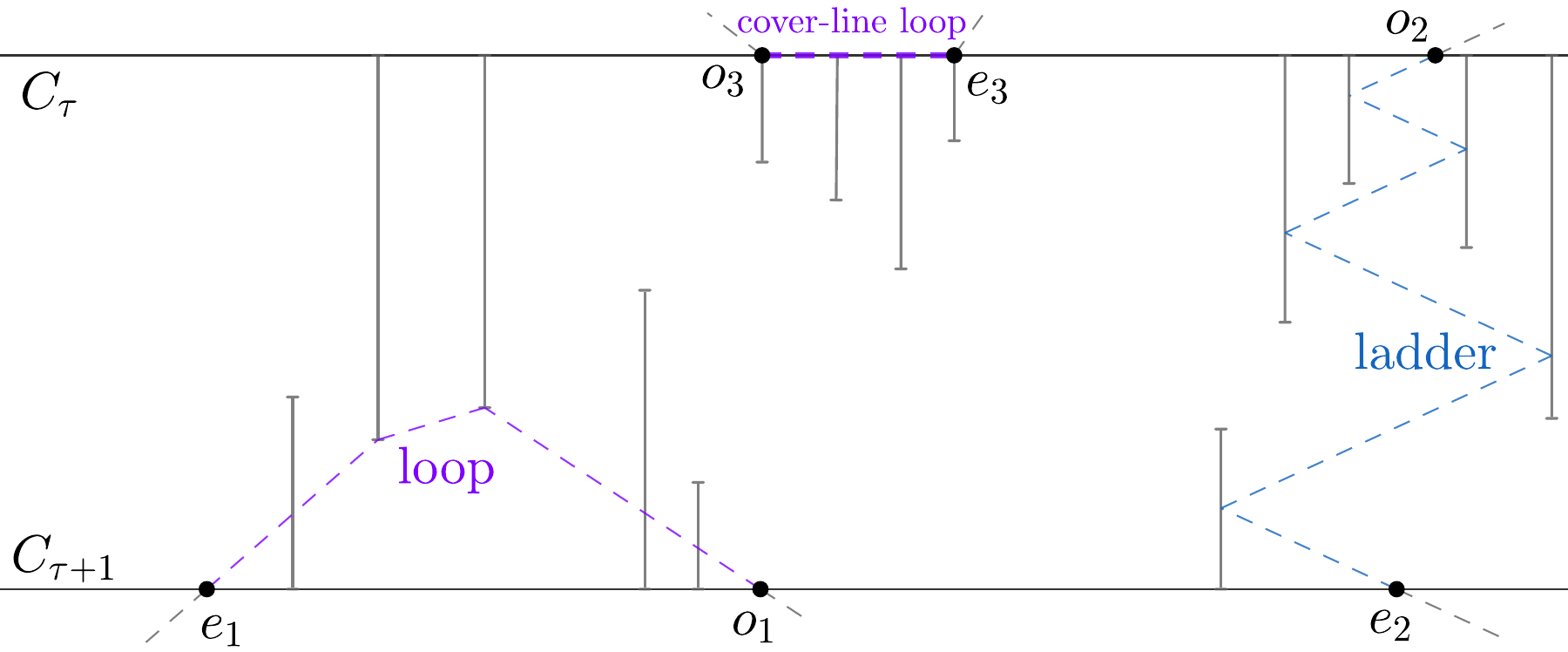}
    \caption{An example of loops and ladders in a strip $S_\tau$ (i.e. area between cover-lines $C_\tau, C_{\tau + 1}$)}
    \label{fig:loops-ladders}
\end{figure}

Since we're assuming $ H > 3 $ (see Theorem \ref{thm:exact-answer}), we get that $ \OPT $ is not limited to a single strip, and that it has to indeed enter and exit any given strip that it intersects with (i.e. there is no strip that $ \OPT $ completely lies inside it).

Note that if a path of $\OPT_\tau$ is a cover-line loop, i.e. a section of the line $C_\tau$ or $C_{\tau+1}$, then the entry points of that path must be the two end-points of this section. In other words, if for a cover-line loop of $\OPT_\tau$ the first point is $e_j$ on (say) $C_\tau$, and the last point is $o_j$ on $C_\tau$, then this subpath must be travelling straight from $e_j$ to $o_j$ without  any change of direction. This is true because otherwise, that cover-line loop would have to go back and forth on some portion on a cover-line, which is only possible if it's self-intersecting; but this is against our assumption that $ \OPT $ is not self-crossing.

The two structures defined below (called a zig-zag and a sink) are the two configurations that can cause a large shadow.
\begin{definition}[Zig-zag/Sink]\label{def:zig-zag}
	Consider any loop or ladder of  $\OPT_\tau$, call it $P$. Let $\mathcal R = r_{1},{r_{2}},\dots, {r_{t}}$ be the sequence of points of $P$ that are reflection points (indexed by the order they're visited). 
	Consider any maximal sub-sequence $r_j,r_{j+1},\ldots,r_q$ of 
	$ \mathcal R $ (with $ q \ge j + 1 $)
	such that all are ascending or all are descending reflections, and the segments
	containing them alternate between top and bottom segments; then the subpath 
	$P$ that starts at $r_j$ and ends at $r_q$ is called a \textbf{zig-zag}.\\
	If $r_j,r_{j+1},\ldots,r_q$ is a maximal sub-sequence of 
	$\mathcal R$ that are all ascending or all descending, and all belong to top segments or all belong to bottom segments; then the subpath $P$
	that starts at $r_j$ and ends at $r_q$ is called a \textbf{sink}
	(see Figure \ref{fig:fig2}).
\end{definition}
\begin{figure}[h]
 \centering
	\begin{subfigure}[t]{.45\tw}
		\centering
		\includegraphics[width=.4\textwidth]{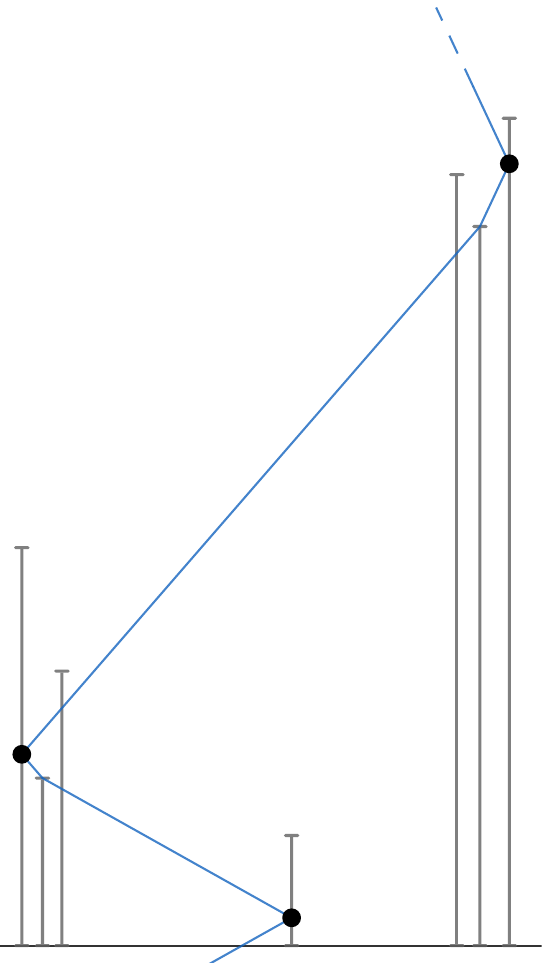}
		\caption{A sink}
	\end{subfigure}
	\begin{subfigure}[t]{.45\tw}
		\centering
		\includegraphics[width=.4\textwidth]{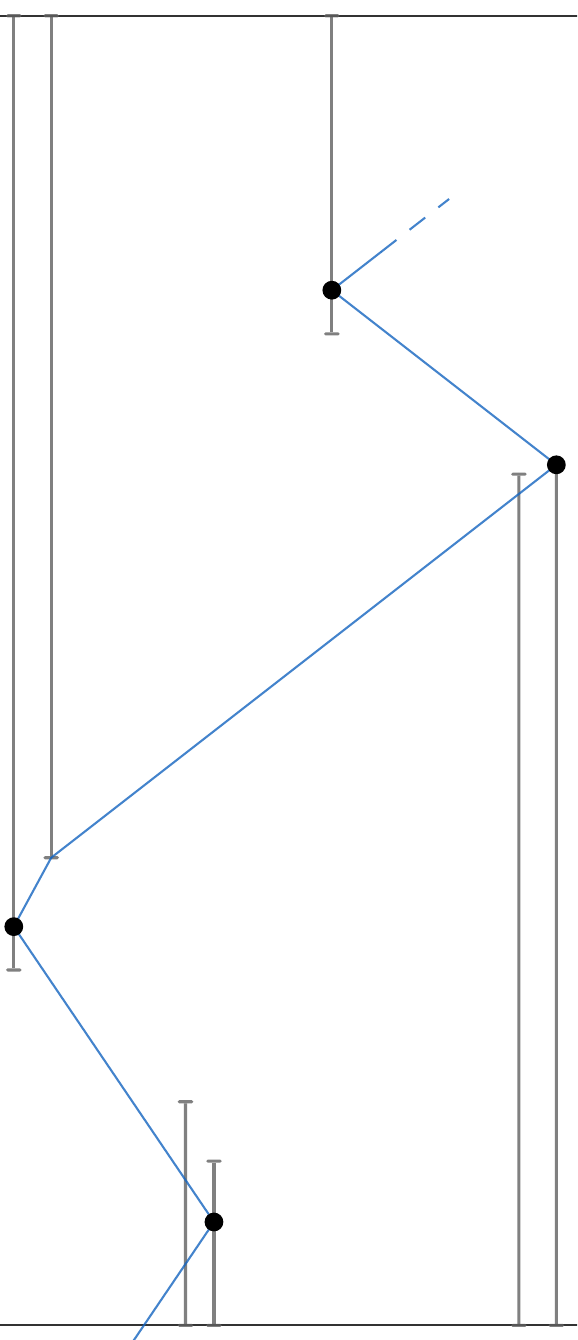}
		\subcaption{A zig-zag}
	\end{subfigure}
    
    \caption{Examples of sinks and zig-zags. The bold black dots represent the reflection points along these paths}\label{fig:fig2}
\end{figure}
% \begin{figure}[h]
%  \centering
%     \includegraphics[width=.23\textwidth]{new-images/sink-example3.pdf}
%     \quad
%     \includegraphics[clip, trim=0 0.2cm 0 0, width=.23\textwidth]{new-images/zig-zag-example3.pdf}
%     \caption{Examples of a sink (left) and a zig-zag (right). Bold dots indicate the reflection points.}\label{fig:fig2}
% \end{figure}

Note that each zig-zag has at least two reflection points (or else it is a sink instead).
Also, using Corollary \ref{cor:alter}, the reflection points
in a zig-zag or sink alternate between left and right reflections.
The next lemma is used critically to show that very specific structures (made by zig-zags and sinks) are responsible for
having a large shadow along a ladder or loop in $\OPT_\tau$. Additionally, we can partition each ladder/loop into parts (subpaths), such that the shadow
of the ladder/loop is equal to the maximum shadow among these parts, and each part is a path that consists of up to three sinks and/or zig-zags. So the shadow of a loop/ladder is within $O(1)$ of the maximum shadow of zig-zag/sinks along it.

\begin{lemma}[\bf First major lemma]\label{lemma:zig-zag}
    Consider any strip $S_\tau$ and any ladder or loop $P\in \OPT_\tau$ within $S_\tau$. Suppose the sequence of reflection points of $P$ is $r_1,\ldots,r_q$.
    These reflection points can be partitioned into disjoint parts, say part
    $i$ consists of reflection points $r_{a_i},r_{a_i+1},\ldots,r_{a_j}$,
    where the subpath of $P$ from $r_{a_i}$ to $r_{a_j}$ is concatenation of up to three sections in the following order:
 % \begin{enumerate}
 %     \item [a)]
   \textbf{a)} 
    A sink
   , followed by
     % \item [b)]
   \textbf{b)} 
    A zig-zag
   , followed by
     % \item [c)]
   \textbf{c)} 
    A sink,
 % \end{enumerate}
    where any of these three sections can possibly be empty, and the last reflection of a section is common with  the first reflection of the next section.
    Furthermore, for any vertical line $\Gamma$, there is at most one of these parts (of the partition) that intersects with it, i.e. the shadow of the ladder/loop is the maximum shadow among the parts plus 2.
\end{lemma}

The proof of this lemma is rather involved and appears in Subsection \ref{sec:zig-zag-lem}.
To give an idea of the proof, we essentially show that for any loop or ladder in any strip, the vertical line at which the largest shadow for that loop or ladder happens, can intersect with at most two sinks and a zig-zag. So the shadow of a loop or ladder is within $O(1)$ of the maximum shadow of the zig-zags and sinks along that.

%%%%%%%%%%%%%%%%%%%%%%%%%%%%%%%%%%%%%%%%%%%%%%%%%%%%%%%%%%%
\subsection{Properties of a Near-Optimum Solution}
\label{sec:near-optimum-property}
In this section we introduce some lemmas that describe structures of a near-optimum solution with small shadow. The next three major lemmas, namely Lemmas
\ref{lem:strip-bounded-shadow}, \ref{lemma:reflections}, and \ref{lemma:overlapping-loops-ladders} are the main lemmas used in the proof of Theorem \ref{thm:onestrip-shadow}. 
We prove these in Subsections \ref{sec:lemma-shadow-of-sink/zigzag}, \ref{sec:lemma-pure-reflection-size}, and \ref{sec:lemma-overlapping-loops/ladders}, respectively. We make alterations to an assumed optimum solution such that some new structural properties hold; we ensure that the alterations have a limited additional cost.

\begin{lemma}[\bf Second major lemma]\label{lem:strip-bounded-shadow}
    Consider $\OPT_\tau$ for an arbitrary strip $S_\tau$ and
    let $\opt_\tau$ be the total cost of $\OPT_\tau$. Given any 
    $\eps>0$, we can change $\OPT_\tau$ to a solution of cost at most $(1+O(\eps))\opt_\tau$ where the shadow of each zig-zag and sink is at most $O(1/\eps)$.
\end{lemma}

The proof of this lemma appears in Subsection \ref{sec:lemma-shadow-of-sink/zigzag}. The following corollary immediately follows from Lemmas \ref{lemma:zig-zag} and \ref{lem:strip-bounded-shadow}:
\begin{corollary}\label{cor:loop-bounded-shadow}
    There is a $(1+\eps)$-approximate solution in which all loops/ladders have shadow $O(1/\eps)$.
\end{corollary}

\begin{definition}\label{def:pure-reflections}
    Let $\mathcal R = p_i, p_{i+1},\ldots, p_q$ be any sequence of consecutive points in $\OPT$ such that $p_i$ and $p_{q}$ are reflection points. If none of $p_j$'s ($i<j<q$) in $\mathcal R$ is a tip of a segment, then $\mathcal R$ is called a \textbf{pure reflection sequence}.
\end{definition}
So each point in $\mathcal R$ is either a straight point or a pure reflection according to Lemma \ref{obs:break-on-tip}. 
The following lemma (whose proof appears in Section \ref{sec:lemma-pure-reflection-size}) shows the existence of a near-optimum with bounded size pure reflection sequences.

\begin{lemma}[\bf Third major lemma]\label{lemma:reflections}
    Consider $\OPT_\tau$ for an arbitrary strip $S_\tau$ and suppose the total length of legs of $\OPT_\tau$ is $\opt_\tau$.
    Given $\eps>0$, we can change $\OPT_\tau$ to a solution of cost at most $(1+\eps)\opt_\tau$ in which 
    the size of any pure reflection sequence is bounded by  $O(\frac1\epsilon)$.
\end{lemma}

Our next goal is to show that for any vertical line, it can intersect at most $O(1)$ many loops or ladders of
$\OPT_\tau$ in a strip $S_\tau$. This together with the above corollary implies there is a $(1+\eps)$-approximate solution s.t. the shadow in each strip $S_\tau$ is $O(1/\eps)$.

\begin{definition}
    A collection of loops and or ladders are said to be \textbf{overlapping} with each other if there is a vertical line that intersects
    all of them.
\end{definition}

\begin{lemma}[\bf Fourth major lemma]\label{lemma:overlapping-loops-ladders}
    Consider $\OPT_\tau$, the restriction of $\OPT$ to any strip $S_\tau$. We can modify the solution (without increasing the shadow or the cost) such that
    there are at most $O(1)$ loops or ladders in $ \OPT_\tau $ that all are overlapping with each other. 
\end{lemma}
We prove this Lemma in Section \ref{sec:lemma-overlapping-loops/ladders} as it consists of multiple parts.

\subsection{Proof of Theorem \ref{thm:onestrip-shadow}}
\label{sec:thm-onestrip-shadow}
Assuming the correctness of main lemmas defined in this section (i.e. Lemmas \ref{lemma:zig-zag}, \ref{lem:strip-bounded-shadow}, \ref{lemma:reflections}, and \ref{lemma:overlapping-loops-ladders}), we can now prove Theorem \ref{thm:onestrip-shadow}.

If the height of the bounding box is at most 3, refer to Theorem \ref{thm:exact-answer}.
Otherwise, consider any strip $S_\tau$ (to be more precise, $S_\tau$ can be any arbitrary strip of height $1$ in the plane).
Using Lemma \ref{lem:strip-bounded-shadow} for parameter $ \eps_1 = \frac\eps2 $, there is a solution $\mathcal O''$ of cost at most $(1+\frac\eps2)\cdot\opt$
where the shadow of each sink and zig-zag is bounded by $O(\frac1{\eps/2}) = O(1/\eps)$.
By Lemma \ref{lemma:zig-zag}, each loop or ladder in $S_\tau$ has a shadow that is at most 3 times the maximum  shadow of a sink or zig-zag in it, plus two. So each loop or ladder has shadow $O(1/\eps)$.
Finally, Lemma \ref{lemma:overlapping-loops-ladders} shows that there can be at most $O(1)$ overlapping loops or
ladders in a strip. Thus, the overall shadow of $\mathcal O''$ in $S_\tau$ is bounded by $O(1/\eps)$.
Furthermore, we apply Lemma \ref{lemma:reflections} on $ O'' $ for parameter $ \eps_2 = \frac\eps{\eps + 2} $ to get a solution $ \mathcal O' $. This new solution has the property that with an additional cost of factor $ (1 + \frac\eps{\eps+ 2}) $ compared to $ \mathcal O'' $, the size of any pure
reflection sequence is bounded by $O(\frac{\eps + 2}{\eps}) = O(1/\eps)$. The total cost of $ \mathcal O' $ is at most 
\begin{align*}
    (1 + \eps_1)\cdot (1+\eps_2)\cdot \opt &= (1 + \tfrac\eps2) \cdot (1 + \tfrac\eps{\eps + 2})\cdot \opt\\
    &= (1 + \tfrac\eps2 + \tfrac\eps{\eps + 2} + \tfrac{\eps^2}{2(\eps + 2)})\cdot\opt\\
    &= (1 + \eps (\tfrac12 + \tfrac1{\eps + 2} + \tfrac\eps{2(\eps + 2)}))\cdot \opt\\
    &= (1 + \eps)\cdot \opt,
\end{align*}
resulting in the statement of the theorem. \qed

\vspace{.5cm}
\noindent \textbf{Note:} Although having a bound on the  length  of pure reflection sequences is not in the statement of the theorem, we use this extra property crucially in designing our DP to find a near-optimum solution.
        
%%%%%%%%%%%%%%%%%%%5%%%%%%%%%%%%%%%%%%%%%%%%%%%%%%%%%%%%%%%%%%%%%%
\section{Main Algorithm and Reduction to Structured Bounded Height Instances}\label{sec:main-alg}
As mentioned in the introduction, we follow the paradigm of Arora \cite{Arora98} for designing a PTAS for classic Euclidean TSP with some modifications. We focus more on defining the modifications that we need to make to that algorithm. 

First, we describe the main algorithm and how it reduces the problem into a collection of 
instances with a constant-height bounding box. 
We show how those instances can be solved using another DP (referred to as the {\em inner DP}), and how we can combine the solutions for them using another DP (referred to as the {\em outer DP}) to find a near-optimum solution of the original instance.
Recall that in Subsection \ref{sec:prem}, we assumed the minimal bounding box of the instance has length $L$ and height $H$
and we defined $ B = \max \{L, H-2\} $, and also we can assume that $B\leq \frac{n}{\eps}$. We moved each line segment to be aligned with a grid point with side length $\frac{\epsilon B}{n^2}$ (at a loss of $(1+\eps)$ at approximation). Now, we scale the grid (as well as the line segments of the instance) by a factor of $\rho=\frac{4n^2}{\epsilon B}$ so that each grid cell has size 4. We obtain
an instance where each line segment has length $\rho$, all have even integer coordinates, any two segments are at least 4 units apart, and the bounding box has size $N=O(n^2/\eps)$. Let this new instance be $\mathcal I$. 
Note that if we define cover-lines as before but with a spacing of $\rho$,
all the arguments for the existence of a near-optimum solution with a bounded shadow in any strip (the area between two consecutive cover-lines) still hold. We will present a PTAS for this instance. It can be seen that this implies a PTAS for the original instance of the problem. From now on, we use $\OPT$ to refer to an optimum solution of instance
$\mathcal I$, and $\opt$ to refer to its cost. Note that since the bounding box has side length $N$, then $\opt\geq 2N$.

\subsection{Dissecting the Original Instance into Smaller Subproblems}\label{sec:dissection}

Similar to Arora's approach, we do the hierarchical dissectioning of the instance into nested squares using random axis-parallel dissectioning lines, and put portals at these dissecting lines. We continue this dissectioning process until the distances between horizontal (and so vertical) dissecting lines is $h\cdot\rho$ for $h=\lceil 1/\eps\rceil$. So at the leaf nodes of our recursive decomposition quad-tree, each square is $(h\cdot\rho)\times (h\cdot\rho)$, and the height of the decomposition is $\log (N/\rho h)=O(\log n)$ since $B\leq\frac{n}{\eps}$. We choose vertical dissecting lines only at odd $x$-coordinates so no line segment of the instance will be on a vertical dissecting line.

We define our cover-lines $C_\tau$ based on these horizontal dissecting lines carefully.
Consider the first (horizontal) dissecting line we choose, this will be a cover-line, and then moving in both up and down directions from this line, we draw horizontal lines that are $\rho$ apart. These will be all the cover-lines.
Label the cover-lines from the top to bottom by $C_1, C_2,\dots, C_\sigma$ in that order. As before, the smallest index $\tau$ such that $C_\tau$ crosses a line segment is the cover-line that "covers" that line segment. 

We partition the cover-lines into $h$ groups based on their indices: Group $G_j$ contains all those cover-lines with index $\tau$ where $j=\tau \pmod h$.
Let $G_{j*}$ be the group of cover-lines that includes the first horizontal dissecting line, and hence all the other horizontal dissecting lines as well.
The arguments for the case of unit-length line segments to show there is a near-optimum solution in which the shadow in each strip of height $1$ is $O(1/\eps)$ (Theorem \ref{thm:onestrip-shadow}), also imply the same for the scaled instance $\mathcal I$. Furthermore, if we consider $h$
consecutive strips, i.e. the area between two consecutive cover-lines in the same group $ G_j $,
then there is a near-optimum solution that has shadow $O(h/\eps)=O(1/\eps^2)$.
Our goal is to show that, at a $(1+\eps)$-factor loss, we can simply drop the line segments that are intersecting the horizontal dissecting lines (i.e. all those intersecting cover-lines in $G_{j^*}$) with appropriate 
consideration of portals (to be described). Removing the line segments that cross the dissecting lines allows us to decompose the instance into "independent" sub-instances that interact only via portals.

For each cover-line $C_\tau$, we define a set $B_\tau$ of disjoint (horizontal) {\em intervals} of length $\rho$ placed on it so that each line segment covered by $C_\tau$, is intersecting one of these intervals. To define $B_\tau$, move on $ C_\tau $, from left to right, start by placing the left corner of the first interval of $B_\tau$ on it at the intersection of the left-most segment covered by $C_\tau$; all the segments covered by $C_\tau$ intersecting this interval are considered "covered" by this interval. Next, pick the first segment to the right of the latest interval that is intersecting $C_\tau$, but not intersecting (and so not covered by) the previous intervals, and place the left point of the next interval of $B_\tau$ at that intersection (all the segments intersecting $C_\tau$ and this interval are now covered by this interval). Continue this process until all segments on $ C_\tau $ are covered by an interval (see Figure \ref{fig:s1-s2-s3-s4}). 
Let $\mathcal B=\cup_{\tau=1}^\sigma B_\tau$.

\begin{observation}\label{obs:box-apart}
    A segment covered by an interval of cover-line $C_\tau$ and another segment covered by an interval of cover-line $C_{\tau+2}$ are at least $\rho$ apart ($\tau\leq \sigma-2$). 
\end{observation}

\begin{lemma}\label{lemma:box-size-opt-ineq}
    $\displaystyle\opt\geq \frac{\rho\cdot |\mathcal B|}{6}.$
\end{lemma}
\begin{proof}
For each $B_\tau$, let $i_1, i_2,\dots, i_\eta$ be the intervals on $ C_\tau $ ordered from left to right. Now partition $ B_\tau $ into $ O_\tau \cup E_\tau $ where
	$O_\tau$ consists of intervals $i_q$ with an odd $q$, and $E_\tau$ consists of those with even $q$'s. We also partition $C_\tau$'s into 3 groups based on the value of $\tau\pmod 3$.
	We get a partition of all intervals into 6 groups based on: Whether an interval on $C_\tau$ is in $O_\tau$ or $E_\tau$ (two choices), and what $\tau \pmod 3$ is (three choices). 
	Let $N_j$'s ($1\leq j\leq 6$) be the total number of intervals in these 6 parts. 
	Note that $\sum_{j=1}^6 N_j=|\mathcal B|$, and any two segments $s,s'$ covered by intervals from different groups are at least $\rho$ apart (if they there are 
	covered by intervals in the same cover-line then they are $\rho$ apart horizontally and if they covered by intervals in different cover-lines then by Observation \ref{obs:box-apart} they are at least $\rho$ apart).
	So an optimum solution for the instance that only contains segments covered by intervals of part $N_j$, must have cost at least $\rho N_j$ as it must have at least $N_j$ legs of size at least $\rho$. Since one of these parts has size at least $|\mathcal B|/6$, the statement follows.
\end{proof}

\begin{lemma}\label{lem:cover-line-group}
    For a $j$ chosen randomly from $[1..h]$, we have 
    \[
        \E[\; \rho\sum_{C_\tau\in G_j}|B_\tau|\; ]=O(\eps\cdot\opt).
    \]
\end{lemma}
\begin{proof}
    For each $1 \le j \le h$, let $\mathcal B_j = \bigcup_{C_\tau\in G_j}{B_\tau}$. Using Lemma \ref{lemma:box-size-opt-ineq}, we have $\sum_{j = 1}^h \abs{\mathcal B_j} = \abs{\mathcal B} \le 6\cdot opt / \rho$. Now we obtain
    \begin{align*}
        \E[\rho \sum_{C_\tau\in G_j}\abs{B_\tau}] &= \rho\cdot \E[\sum_{C_\tau\in G_j}\abs{B_\tau}]\\
        &= \rho \cdot \E[\abs{\mathcal B_j}] \\
        &= \frac{\rho}{h}\cdot|\mathcal B|\\
        &\le \frac{\rho}{h} \cdot \frac{6\cdot opt}{\rho}  = O(\opt/h) = O(\eps\cdot \opt).
    \end{align*}
\end{proof}

Similar to Arora's scheme for TSP, for $m=O(\frac{1}{\eps}\log(N/\rho h))$, we place portals at all 4 corners of a square in the decomposition, plus an additional $m - 1$ equally distanced portals along each side (so a total of $4m$ portals on the perimeter of a square of the dissection). For simplicity, we assume $m$ is a power of 2 and at least $\frac{4}{\eps}\log(N/\rho h)$.
We say a tour is {\em portal respecting} if it crosses between two squares in our decomposition only via portals of the squares. A tour is {\em $r$-light} if it crosses the portals on each side of a square of the dissection at most $r$ times.

For classic (point) TSP, it can be shown that there is a near-optimum solution that is portal respecting and $r$-light for $r=O(1/\eps)$. Our goal is to
show a similar statement, except that we want the restriction of the tour to each "base" square of side length $O(h\cdot\rho)$ to have bounded (by $O(h/\eps)=O(1/\eps^2)$) shadow as well. We then show that we can find an optimum solution with a bounded shadow for the base cases using a DP. This will be our {\em inner DP}. We then show how the solutions of for the 4 sub-squares of a square in our decomposition can be combined into a solution for the bigger subproblem using the outer DP.

We show that at a small loss in approximation (i.e. $O(\eps\cdot\opt)$), we can drop all the line segments of input that are intersecting the horizontal dissecting lines (i.e. covered by a cover-line in group $G_{j^*}$), solve appropriate subproblems, and then extend the solutions to cover those dropped segments. This modification requires certain portals of each square in the decomposition to be visited in the solution for that square.

We show there is a feasible solution that visits all the remaining segments as well as the "required" portals, of total cost at most $(1+\eps)\cdot\opt$, and that such a solution can be extended to a feasible solution visiting all the segments of the original instance (i.e. including the ones that we dropped) at an extra cost of $O(\eps\cdot\opt)$.

%%%%%%%%%%%%%%%%%%%%%%%%%%%%%%%%%%%%%%%%%%%%%%%%%%%%%%%%%%%55
\subsubsection[Dropping the Segments hitting the Dissecting Lines]{Dropping the segments intersecting horizontal dissecting lines}\label{subsec:dropping-horizontal}

This Subsection is dedicated to proving the following lemma:

\begin{lemma}\label{lemma:horizontal-cover-lines}
    Given instance $\mathcal I$, there is another instance $\mathcal I'$ that is obtained by removing all the segments that are crossing cover-lines in $G_{j^*}$ (i.e. intersecting horizontal dissecting lines),
    and instead some of the portals around (more precisely, the top and bottom sides of) 
    each square of quad-tree dissection are required to be covered (visited); such that there is a solution for 
    $\mathcal I'$ of cost at most $(1+O(\eps))\cdot\opt$, and such a solution can be extended to a feasible solution of $\mathcal I$ of cost at most $(1+O(\eps))\cdot\opt$. Furthermore, the shadow of the solution for $\mathcal I'$
    between any two consecutive cover-lines in $G_{j^*}$ is at most
    $4$ more than the shadow of $\OPT$ between those two lines.
\end{lemma}

We say the edges of the bounding box are {\em level 0} dissecting lines, the first pair of dissecting lines are {\em level 1} dissecting lines, and so on.
Consider a square $S$ in our hierarchical decomposition and suppose
it is cut into four squares $S_1,S_2,S_3,S_4$ by two dissecting lines
where the horizontal one, line $\Gamma$, is the cover-line $C_\tau$ from $G_{j^*}$, 
and is a level $j$ dissecting line.
Recall that we place a total of $2m$ portals 
along $\Gamma$ inside $S$; $m$ portals on the common sides of $S_1,S_4$ and $m$ along the common side of $S_2,S_3$.
Define $B_\tau(S)$ to be the set of intervals in $B_\tau$ (intervals of $C_\tau$)
that cover a segment that lies inside $S$ (and so intersects with $\Gamma$) (see Figure \ref{fig:s1-s2-s3-s4}).

\begin{figure}[h]
    \centering
    \includegraphics[width=.55\textwidth]{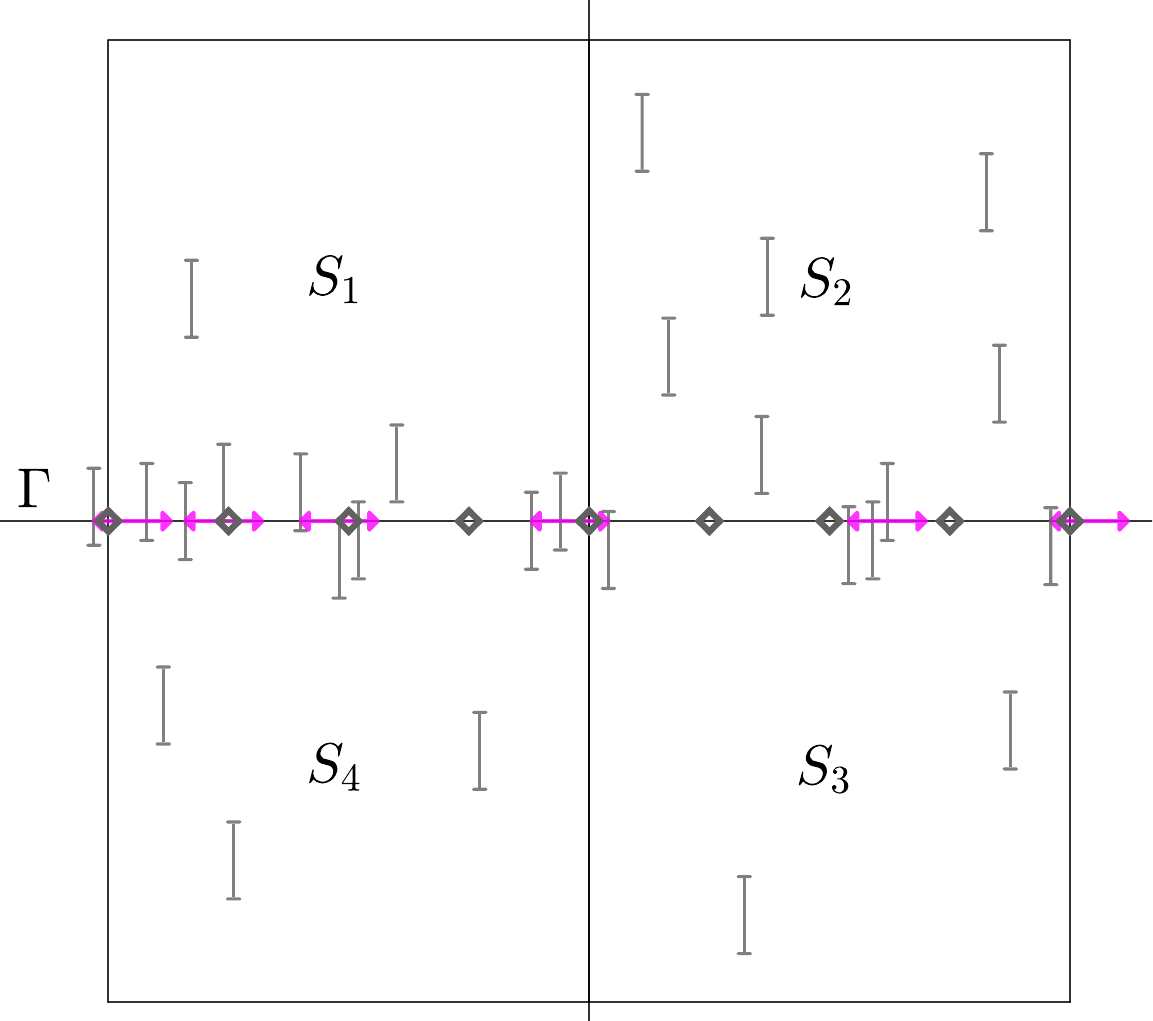}
    \caption[Breaking a square $S$ into $4$ smaller squares]{Breaking a square $S$ into $4$ smaller squares. The magenta parts on line $\Gamma$ (i.e. the cover-line $C_\tau$) show the interval set $B_\tau(S)$}
    \label{fig:s1-s2-s3-s4}
\end{figure}

For each $b\in B_\tau(S)$, suppose $p(b)$ is the nearest portal to it in $S$ among the portals on $\Gamma$, and let $s(b)$ be the left-most segment covered by $b$ that is in $S$.
We are going to modify $\OPT$ in the following way:
Consider a point $p_s$ on $s(b)$ visited by $\OPT$.
%and suppose $q$ is the  point visited by $\OPT$ after $p_s$. Remove the leg $p_s\to q$ and instead 
Insert the following "legs" to the path: travel from $p_s$ vertically along $s(b)$ until you arrive at its intersection with $\Gamma$,  i.e. arrive at interval $b$ (this length is at most $\rho$), then travel along $\Gamma$ to the right-most segment covered by $b$ (this is also at most $\rho$), and then travel to $p(b)$, and then travel back to $p_s$.
For every other segment $s'$ covered by $b$ in $S$, we are going to short-cut any point on $s'$ that was visited by $\OPT$ as all these segments are now covered by the newly added legs (see Figure \ref{fig:modified}). We also 
short-cut the second visit to $p_s$.

\begin{figure}[h]
 \centering
    \includegraphics[width=.4\textwidth]{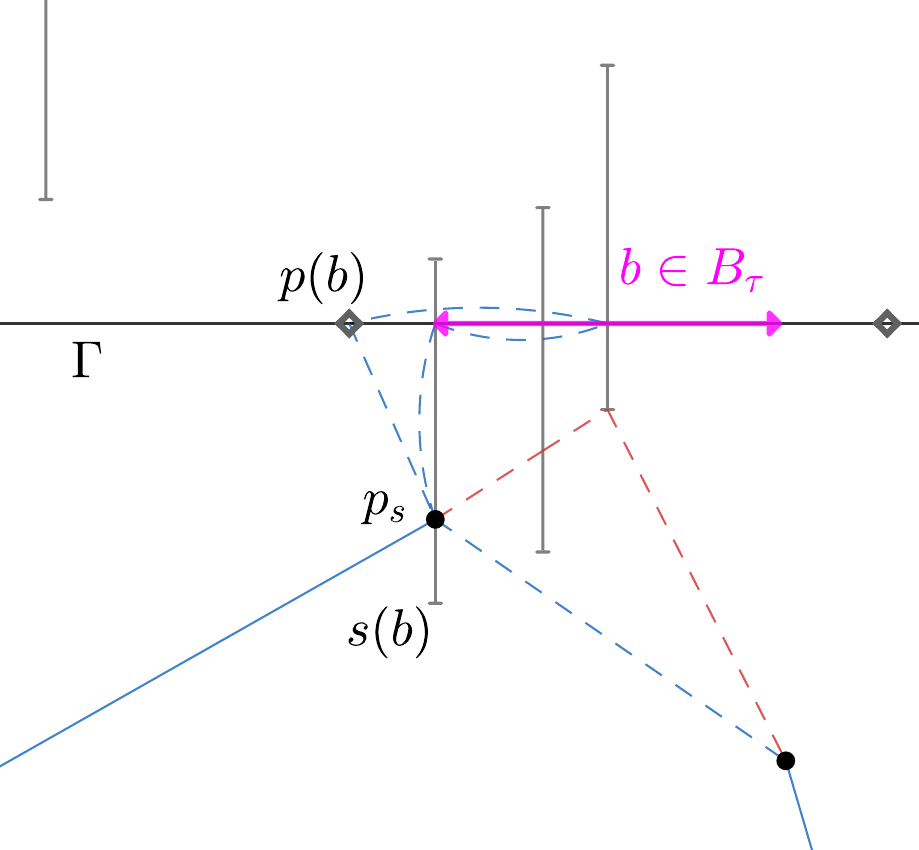}
    \caption[The modified solution for dropping segments on dissecting lines.]{The modified solution for dropping segments crossing horizontal dissecting lines: follow the blue dashed lines from $p_s$; red dashed lines are the discarded parts of the original path}
    \label{fig:modified}
\end{figure}

Using triangle inequality, the expected length of the new legs will increase the cost of the solution by at most $2\rho+2||p_sp(b)||\leq 2(\rho+\frac{N}{2^jm})$. We do this for all the intervals on $\Gamma$ and inside $S$, i.e. if $\OPT$ visits a segment covered by that interval $b$, we change $\OPT$ to make a detour to visit $p(b)$ as well.
Note that each interval $b\in B_\tau$ can belong to at most two $B_\tau(S)$'s (two adjacent squares that $b$ intersects with), and the intervals for which this modification can happen for, are at least $h\cdot\rho$ apart because that is the minimum size of a square of the dissection.
 
Given the random choice of our dissecting lines, since dissecting lines are $h\cdot\rho$ apart, are randomly chosen, and each interval has length $\rho$, the probability that an interval $b\in B_\tau$ appears in two $B_\tau(S)$'s (i.e. cut by a dissecting line), is at most $1/{h}=\eps$. Also, each cover-line in $G_{j^*}$ is a level $j$ dissecting line with probability $2^{j-1}/(N/\rho h)$.
Thus, the expected increase in the cost by this modification for all the
interval of $C_\tau$ is at most

\begin{align*}
    \sum_{j=1}^{\log (N/\rho h)}& \Pr[\mbox{$\Gamma$ is level $j$}]\cdot (1+\eps)\cdot|B_\tau|\cdot 2(\rho+\frac{N}{2^jm})\\
    &\le  2(1+\eps)\cdot|B_\tau|\cdot\sum_{j=1}^{\log(N/\rho h)} \frac{2^{j-1}}{N/\rho h}\cdot\left.(\rho+\frac{N}{2^j m}\right.)\\
    &\le  2(1+\eps)\cdot |B_\tau|\cdot \frac{\rho h}{N} \cdot \left(\frac{N}{h}+\frac{N\log(N/\rho h)}{2m}\right)\\
    &\le  (1+\eps)\cdot|B_\tau|\cdot\rho\cdot(1+\eps h)\\
    &\leq 4\rho\cdot|B_\tau|.
\end{align*}

Considering all cover-lines in $G_{j^*}$, this implies the total expected increase
in the cost is at most $\sum_{C_\tau\in G_{j^*}} 4\rho|B_\tau|$, which combined with Lemma 
\ref{lem:cover-line-group}, implies with probability at least $1/2$, the increase in total cost is at most $O(\eps\cdot\opt)$. Each portal $p$ that is visited by a detour as described above is called a {\em required portal}. 

In fact, we can short-cut more paths so that the number of detours to each portal is bounded by 2. Informally, only the left-most interval to the left of $p$ that has made a detour to $p$, along with the right-most interval to the right of $p$ that has made a detour to $p$ are sufficient to cover all the segments of the intervals in between them. More specifically, consider a portal $p$ on $\Gamma$ and let $b_L(p)$ be the left-most interval in $B_\tau(S)$ to the left of $p$ that covered a segment whose path was detoured to visit $p$ (null if there is no such interval). Similarly, let $b_R(p)$ be the right-most interval among $B_\tau(S)$ to the right of $p$ that covered a segment whose path was detoured to visit $p$ (this too can be  null if there is no such interval).
 In other words, there was a segment $s_L=s(b_L(p))$ and a segment $s_R=s(b_R(p))$ that
 were visited by $\OPT$, and we made a detour to $p$ when $\OPT$ visited $s_L$ and $s_R$. The detour from $s_L$ to $p$ covers all the segments of intervals between $b_L(p)$ and $p$. Similarly, the detour from $s_R$ to $p$ covers all the segments of interval between $p$ and $s_R$.
 Thus, for any interval $b'$ between $b_L(p)$ and $b_R(p)$, all the segments covered by $b'$ are also covered by the detours of $s_L$ and $s_R$. This means for
 all those intervals $b'$, we can short-cut the segments covered by them entirely (in particular, they don't need to make a detour to $p$). 
 Therefore, at most two intervals will have detours to $p$, namely $b_L(p)$
 and $b_R(p)$. And the detours to different portals are disjoint,
 so the added detours don't overlap on $\Gamma$, and since short-cutting doesn't increase the shadow, we only add a shadow of at most 2 per cover-line to the solution. This implies that if we focus on the modified solution restricted to the strip between two cover-lines in $G_{j^*}$, it still has a bounded shadow.
These arguments complete the proof of Lemma \ref{lemma:horizontal-cover-lines}. \qed

\subsection{Dynamic Program}\label{sec:DP}
As mentioned previously, the outer DP is responsible of combining the solutions of each "base-case" subproblem into a solution of the main instance of the problem; and the inner DP is responsible of solving each base-case (i.e. subproblems with a bounding box of size $O(1/\eps)$).

\subsubsection{Outer DP}\label{sec:outer-dp}

The outer DP based on the quad-tree dissection is similar to the classic PTAS for Euclidean TSP.
One can show that for $r=O(1/\eps)$, there is an $r$-light portal respecting tour for $\mathcal I'$ with cost
at most $(1+\eps)\cdot\opt'$, where $\opt'$ is the cost of an optimum solution for $\mathcal I'$. 
The base case of this DP will be instances with bounding box of size $\rho\cdot h$. For such instances,
we solve the problem using an inner DP that is described in the Subsection \ref{sec:inner-dp}.

We will use the "patching lemma" the same way it is described in Arora's approach. We show there is a near-optimum solution for $\mathcal I'$ that is
portal respecting and {\em $r$-light}, meaning each square in our quad-tree decomposition is crossed by the solution only $r$ many times on each side for a parameter $r=O(1/\eps)$. Then a DP similar to the point TSP (outer DP) will combine the solutions for the subproblems to find the solution for a bigger subproblem.
Since we don't know which portals for each square are supposed to be "required" in $\mathcal I'$ (so that the solution can be extended to cover the dropped line segments), for each such square we "guess" the set of required portals in our DP; i.e. we will have an entry for each guessed set of portals on the horizontal sides of a square as the set of required portals in our DP. Since the number of portals is logarithmic, this guessing remains polynomially bounded. For now, assume that we know all the required portals, and hence,
instance $\mathcal I'$ itself (even though $\mathcal I'$ is defined based on $\OPT$ which we don't know).

Consider instance $\mathcal I'$ and let $\OPT'$ be the optimum solution for it, and let the cost of that solution be $\opt'$. For each dissecting line $\Gamma$ (vertical or horizontal), let $t(\Gamma)$ be the number of intersections of $\OPT'$ with $\Gamma$ and $T = \sum_\Gamma t(\Gamma)$. 
\begin{lemma}[\cite{Arora98}]\label{lemma:T-less-2opt}
    $T\le 2\cdot\opt'/(\rho h)$.
\end{lemma}
\begin{proof}
    Let $\ell = (x_1, y_1)\to (x_2, y_2)$ be any leg of $\OPT'$. Let $\Delta x = \abs{x_1 - x_2}$ and $\Delta y = \abs{y_1 - y_2}$. The contribution of $\ell$ to $\opt'$ is its length, i.e. $L_\ell = \sqrt{(\Delta x)^2 + (\Delta y)^2}$. Note that due to the scaling, we have $L_\ell \ge 4$. Since the dissecting lines are $\rho h$ apart, there are at most $(\Delta x + \Delta y + 2)/(\rho h)$ dissecting lines that intersect with $\Gamma$; so the contribution of $\Gamma$ to $T$ is at most the same amount. Using the Cauchy-Schwarz inequality, we have
    $2((\Delta x)^2 + (\Delta y)^2)\ge (\Delta x + \Delta y)^2$, which implies
    \begin{align*}
        (\Delta x + \Delta y + 2)/(\rho h) &\le (\sqrt{2((\Delta x)^2 + (\Delta y)^2)} + 2 )/(\rho h)= (\sqrt 2\cdot L_\ell + 2)/(\rho h).
    \end{align*}
    It suffices to show $ \sqrt 2\cdot L_\ell + 2 \le 2L_\ell$; this is easily seen to be true because $L_\ell \ge 4$. Therefore, if we add these inequalities for all legs $\ell$ of $\OPT'$, we get the lemma's statement as the result.
\end{proof}

The following lemma is essentially the same as the one in the  case of point TSP (except we have different stopping points):
\begin{lemma}[\cite{Arora98}]\label{lemma:portal-respecting}
    Considering the randomness of the dissecting lines, with probability of at least $\frac12$, there exists a portal-respecting solution for $\mathcal I'$ with cost at most $(1 + \eps)\cdot \opt'$ for portal parameter $m = O(\frac{1}\eps\cdot \log \frac N{\rho h})$
\end{lemma}

\begin{proof}
	The proof is similar to the survey in \cite{Vazirani}.
    Consider any dissecting line $\Gamma$ of level $j$ and focus on the intersections of $\OPT'$ with that line. Consider any leg $\ell = ab$ of $\OPT'$ which intersects $\Gamma$, say, at a point $q$ and suppose $p$ is the nearest portal of $\Gamma$ to $q$. Replace $\ell$ with
    with two new "legs" $\ell_1 = ap$ and $\ell_2 = pb$. Let $d$ be the distance of $q$ to $p$. Using triangle inequality, it can be seen that $\ell_1 + \ell_2 \leq \ell + 2d$; meaning the additional cost for going through portal $p$ is at most $2d$. The distances between the portals on level $j$ line $\Gamma$ are $d_j = \frac N{2^jm}$, and clearly $d\le d_j$. Recall that $\OPT'$ intersects with $\Gamma$, $t(\Gamma)$ times. Thus, the expected increase in cost for any dissecting line $\Gamma$ is at most
    \begin{align*}
        \sum_{j = 1}^{\log N/\rho h}\Pr[\Gamma {\rm \ is\ level\ } j]\cdot t(\Gamma)\cdot 2\cdot\frac N{2^jm} &\le \sum_{j = 1}^{\log N/\rho h}\frac{2^{j - 1}}{N/\rho h}\cdot t(\Gamma)\cdot2\cdot \frac N{2^j m}\\
        &= \frac{\rho h}m\cdot \sum_{i = 1}^{\log N/\rho h} t(\Gamma)\\
        &= \frac{\rho h}m \cdot \log \tfrac N{\rho h} t(\Gamma).
    \end{align*}
    For
    $m \ge \frac{4}\eps \log\frac N{\rho h}$, the last value above is at most $ \frac{\eps\rho h}{4}\cdot t(\Gamma)$. Adding all these inequalities over different $\Gamma$'s gives us $\frac{\eps\rho h}{4}\cdot T$, which according to Lemma \ref{lemma:T-less-2opt} is at most $\frac\eps2\cdot \opt'$. Using Markov's inequality the statement of the lemma follows.
\end{proof}

The patching Lemma (stated below) for classic Euclidean TSP holds in our setting as well.
\begin{lemma}[The patching Lemma \cite{Arora98}]\label{lemma:patching-lemma}
    For any dissecting line segment $\tau$ with length $L_\tau$, if a tour crosses $\tau$ more than twice, it can be altered to still contain the original tour, but intersect with $\tau$ at most twice with an additional cost not greater than $6L_\tau$. 
\end{lemma}
\begin{proof}
    The same proof as in \cite{Arora98} applies here.
    % Consider any tour that has $k\ge 3$ intersections with $\tau$. For each intersection point, break the tour. Add 2 copies for each of the $k$ intersection points. Now on each side of $\tau$, create a path including all the $k$ new points. If $k$ is odd, match the first $k - 1$ points into $\frac{k-1}2$ pairs, and if $k$ is even, match all the points into $\frac k2$ pairs. If $k$ is odd, connect the first new point on one side to its copy on the other side, and if $k$ is even, connect the first and last points on one side to their copies to the other side. Seeing each of these new connections as an edge, we can see that the new points form a Eulerian graph, and the new configuration intersects with $\tau$ at most twice. We can traverse the corresponding graph (and shortcutting at appropriate places) to reach a new tour. The extra cost of the new tour is at most $2L_\tau$ for the two paths, plus $2L_\tau$ for the pairings, totaling at most $4L_\tau$. Note that a portal-respecting tour will remain portal-respecting after applying this lemma.
\end{proof}

\begin{observation}\label{cor:portal-twice}
	A single point can be seen as a $ 0 $-length segment. By using Lemma \ref{lemma:patching-lemma}, we get that at no additional cost (i.e. extra cost of $ 6\times 0 $), each portal is visited at most twice.
\end{observation}

The next lemma shows the existence of a near-optimum solution that is 
$r$-light and portal respecting for $r=O(1/\eps)$:
\begin{lemma}\label{lemma:p-tour-r-light}
    Given the randomness in picking the dissecting lines, with probability at least $\frac 12$, there is an $r$-light portal respecting tour for $\mathcal I'$ with cost at most $(1+\eps)\cdot \opt'$ for 
    $r = O(\frac{1}\eps)$.
\end{lemma}
\begin{proof}
    This is implied by the {\em Structure Theorem} in \cite{Arora98}, and the similar proof works here.
\end{proof}

\subsubsection*{Outer DP Table and Time Complexity} 

The outer DP is similar to the DP for classic Euclidean TSP except that we need to
take care of required portals that are going to be guessed and passed down to the subproblems.
Note that there are $O(n)$ subproblems in each level of the dissection tree, and so a total of $O(n\log n)$
squares to consider.
For each square $S$ with $4m$ portals around it, we guess a subset of portals on the horizontal sides of $S$
to be required. The number of such guesses is $2^{2m}$ where $m=O(\frac{1}{\eps}\cdot\log\frac{N}{\rho h})=O(\log n/\eps)$.
There are $(4m+1)^{4r}$ guesses for up to $4r$ portals to be chosen for an $r$-light portal respecting, and
at most $(4r)!$ for the pairings of these portals.
So the size of the DP table  is at most $O(n\log n\cdot 2^{2m}\cdot (4m+1)^{4r}\cdot (4r)!)=O(n\log^{O(r)} n)$.

The DP table is filled bottom up. The base cases are when we have a square of side length $\rho\cdot h$.
These subproblems are solved using the inner DP described in the next section.
For every other square $S$ that is broken into 4 squares $S_1,\ldots,S_4$, we solve the subproblem of $S$
after we have solved all subproblems for $S_1,\ldots,S_4$.
The way we combine the solutions from those of the sub-squares to obtain the solution for $S$ is very much like
the classic point TSP. However, we have to extend the solutions so that the line segments that were intersecting the horizontal dissecting line that split $S$, are now fully covered by the guessed required portals for $S_1,\ldots,S_4$.
More specifically, suppose $\Gamma$ is the horizontal dissecting line that corresponds to a cover-line $C_\tau$ from group $G_{j^*}$ (and hence we removed all the segments crossing $C_\tau$ and instead made some of the portals along $C_\tau$ as required). We add those segments of the instance back, and we extend the solutions from the require portals to travel left and right to cover these segments.
Similar to the classic TSP, the total time to fill in the outer DP table is $O(n \log^{O(r)} n)$.

\subsubsection{Inner DP}\label{sec:inner-dp}
Recall that each base case of the quad-tree decomposition is a subproblem defined 
on a square $S$ with size $\rho h \times \rho h$, and has $4m$ portals around it.
Since we assume the solution we are looking for is $r$-light, it means the instance defined by $S$ has also
a set $P$ of size at most $4r$ of portal pairs (where $r=O(1/\eps)$). Each pair $(p_i,q_i)\in P$ specifies that the solution restricted to $S$, has a $p_i,q_i$-path. We are also given a guessed subset $Q$ of the portals around $S$ (specifically  on the top and bottom side of $S$) as the required portals. The goal is to find a minimum cost collection of paths that start/end at the given set of portal pairs $P$ that cover all the line segments in $S$, as well as visit all the required portals in $Q$. Let us denote this instance by $(S,P,Q)$. Note
that by Theorem \ref{thm:onestrip-shadow}, Lemma \ref{lemma:horizontal-cover-lines}, and Lemma \ref{lemma:p-tour-r-light}, there is a near-optimum solution such that it is $r$-light for each square of the dissection, is portal respecting, covers all the required portals, and has shadow bounded by $O(1/\eps^2)$. Also using Lemma \ref{lemma:reflections}, the length
of any pure reflection sequence in it is bounded by $O(1/\eps)$. We describe the inner DP to find 
an optimum solution with bounded shadow (and pure reflection sequence bounded to $O(1/\eps)$ elements) restricted
to subproblem $(S,P,Q)$. For square $ S $, let us use $\OPT_S$ and $\opt_S$ to denote such a bounded shadow optimum solution and its value, respectively.

Informally, the DP is a (nontrivial) generalization of the DP for the classic (and textbook example) bitonic TSP in which the shadow is 2. In our case, the shadow is $O(1/\eps^2)$.
We are going to consider a sweeping vertical line $\Gamma$ in $S$ (that moves left to right) and "guess" the intersections of $\OPT_S$ with it. 

We define an {\em event point set} in the following way:
\begin{definition}[Event Point]
	Given a subproblem triplet $ (S, P, Q) $, each line segment in $S$ is in the event point set. Also, each portal that is on a horizontal side of $S$
	and is either in $Q$, or participates in a pair of $P$, is also in the event point set.
\end{definition}
We consider an ordering of all the elements in the event point set from left to right (i.e. increasing $x$-coordinate), say $v_1,v_2,\ldots,v_{n_S}$, where $n_S$ is the number of event points; note that $n_S=O(n)$.
There are $n_S - 1$ equivalent classes for positions of $\Gamma$, where each class corresponds to when $\Gamma$ is located between $v_i,v_{i+1}$. A sweep line between $v_i,v_{i+1}$ is denoted by $\Gamma_i$.
Since the shadow of $\OPT_S$ is bounded, the intersection of $\Gamma_i$ with $\OPT_S$ has a low complexity. We will give a more concrete explanation of that complexity below.

Recall Observation \ref{obs:legs} and the types of points in a solution (straight point, break point, or reflection point). Also recall the definition of a pure reflection point (a reflection point that  is not
at a tip of a segment of the instance). Consider the global optimum solution that is $r$-light and portal respecting with bounded shadow
and bounded pure reflection sequence that also covers the required portals of each square.
Suppose $p_{a_1},p_{a_2},\ldots,p_{a_k}$ is the sequence of points in $S$ visited by $\OPT_S$ in this order that are 
{\em not} a straight point nor a pure reflection point; so each of them is a break point (tip of a segment) or perhaps a required portal in $Q$, or a portal in $ P $ (i.e. is an entry or exit point in some pair belonging to $ P $).
So any point visited by $\OPT_S$ between $p_{a_i},p_{a_{i+1}}$ (if there is any) is either a straight point
or a pure reflection point. We define subpaths of $ \OPT_S $ named {\em large legs} as follows:

\begin{definition}[Large Leg]
	The path of $ \OPT_S $ from $p_{a_i}$ to $p_{a_{i+1}}$ is a large leg. Each large leg starts and ends from a portal or a tip of a segment, and all the points in between are either straight points or pure reflection points.
\end{definition}

It follows from Lemma \ref{lemma:reflections} that the number of pure reflection points in each large leg is bounded by $O(1/\eps)$. Each large leg can be guessed by making at most $O(1/\eps)$ guesses for segments or points: guess the two end-points of the large leg (which are either portals or tips of segments), then guess at most $O(1/\eps)$ segments that have pure reflection points on them; once we guess the two end-points and the segments for pure reflections, the pure reflection points are uniquely determined. Since there are $O(n^2)$ choices for the end-points and $O(n^{1/\eps})$ choices for the segments of pure reflection points, the total number of possible large legs is bounded by $n^{O(1/\eps)}$. 
Now since we assume $\OPT_S$ has bounded shadow of $O(1/\eps^2)$, for any sweep line $\Gamma_i$, there are at most $O(1/\eps^2)$ large legs of $\OPT_S$ that can cross $\Gamma_i$. 

So for a fixed $i$ (and sweep the line $\Gamma_i$), let
$\mathcal L_i= L_1,\ldots,L_\sigma$ be the sequence of large legs ($\sigma=O(1/{\eps^2})$) of $\OPT_S$
that cross $\Gamma_i$; where each large leg is specified by the end-points as well as the intermediate segments for pure reflections (if there are any). Then the number of possible choices for $\mathcal L_i$ is $n^{O(1/\eps^3)}$. 
Given $i$ and $\mathcal L_i$, let $S^L_{i},S^R_i$ be the left and right part of $S$ (cut by $\Gamma_i$).
If we ignore the segments covered by $\mathcal L_i$ in $S^L_i$, and consider
the end-points of each $L_j$ as portals too, then the restriction of $\OPT_S$ to $S^L_i$ is a collection of paths that start/end at portals of $P$ in $S^L_i$ or end-points of $L_j$'s in $S^L_i$
that cover all the segments in $S^L_i$ not already covered by $\mathcal L_i$, as well as points in $Q\cap S^L_i$.
More specifically, each part of $\OPT_S$ in $S^L_i$ is a path that starts at a $p_j$ for a pair $(p_j,q_j)\in P$, or  at an end-point of $L_k$ that is in $S^L_i$ and ends at a point $p_{j'}$ (or $q_{j'}$) of another pair
in $P$ that is also in $S^L_i$, or at another end-point of some $L_{k'}$ that is in $S^L_i$. So this induces
some pairs of points, denoted by $P^L_i$:
\begin{definition}[Path-wise Pairing $ P^L_i $]
	Set $ P^L_i $ of pairs of points is said to be the \textbf{path-wise pairing} for $ S^L_i $, if there is a path in $ S^L_i $ between the two points of any given pair $ (a, b)\in P^L_i $. 
	Furthermore, each point in a pair $(a,b)\in P^L_i$ is either a portal in $S^L_i$ that is part of a pair in $P$,
	or is an end-point of a large leg $L_j$ that is in $S^L_i$. 
	
	For any such point in $ S^L_i $, say $ p $, there must be a pair in $ P^L_i $ containing that point. We also assume $ (p, p)\in P^L_i $, and if $ p $ is an end point of a large leg in $ \mathcal L_i $ or $ S_i^L $, and if $ q $ is the other end point of that large leg, then $ (p, q)\in P^L_i $.
\end{definition}

We say a set of pairs $P^L_i$ is not {\em promising} if given $\mathcal L_i$, there is no feasible solution in the
entire $S$ whose restriction to $S^L_i$ defines subpaths consistent with $P^L_i$ 
(i.e. they start and end on the same pairs as specified by $P^L_i$). Otherwise, we consider it promising.
For example if $(p_j,q_j)\in P$,
both $ p_j, q_j $ belong to $S^L_i$, and if $(p_j,u),(q_j,v)\in P^L_i$ where $u$ is one end of a long leg $L_1$
and $v$ is one end of a long leg $L_2$, it must be the case that it is possible to have a path from the other end of $L_1$ to the other end of $L_2$. This would be impossible if, for instance, those other ends of $L_1,L_2$ are paired up with other portals in $P^L_i$. 
Note that since there are at most $4r$ pairs in $P$ and $O(1/\eps^2)$ end-points in $\mathcal L_i$,
the number of possible choices for $P^L_i$ is $(1/\eps)^{O(1/\eps)}$. Also, a given $P^L_i$  (together with $\mathcal L_i$), it can be checked if $ P^L_i $ is promising or not in poly-time in $ n $.

This suggests how we can break the instance $(S,P,Q)$ into polynomially many sub-instances.
For a fixed $i$, guess $\mathcal L_i$ among all those with shadow $O(1/\eps^2)$, break $S$ into $S^L_i,S^R_i$, let $Q^L_i=Q\cap S^L_i$, and guess
the new pairs $P^L_i$ (for $S^L_i$) that are promising. We solve $(S^L_i,P^L_i,Q^L_i)$ for each $S^L_i,P^L_i,Q^L_i$ obtained this way. 
We can solve each such subproblem assuming we have solved all subproblems defined by each $\Gamma_j$ for $j<i$. So formally, let us define a {\em configuration}:
\begin{definition}[Configuration]
 A \textbf{configuration} is a vector $ (i, \mathcal L_i, P^L_i) $ where the components are:
\begin{itemize}
    \item $i$ (indicating $\Gamma_i$ and defining $S^L_i$),
    \item The large legs of $\OPT_S$ crossing $\Gamma_i$, denoted by $\mathcal L_i$, $|\mathcal L_i|=O(1/\eps^2)$,
    \item The pairing $P^L_i$ defined by $\mathcal L_i$, $P$, and the restriction of $\OPT_S$ to $S^L_i$.
\end{itemize}
\end{definition}

This configuration (see Figure \ref{fig:event-point}), defines a subproblem: Suppose $\mathcal L_i$ is a given set of large legs crossing $\Gamma_i$.
Find a collection of paths in $S^L_i$ such that $P^L_i$ specifies the start/end of these paths (and is promising), such that these paths cover all the segments in $S^L_i$ (excluding those already covered by $\mathcal L_i$), and also cover all the points in $Q\cap S^L_i$, with shadow at most $O(1/\eps^2)$. 

\begin{figure}
    \centering
    \includegraphics[width=.4\textwidth]{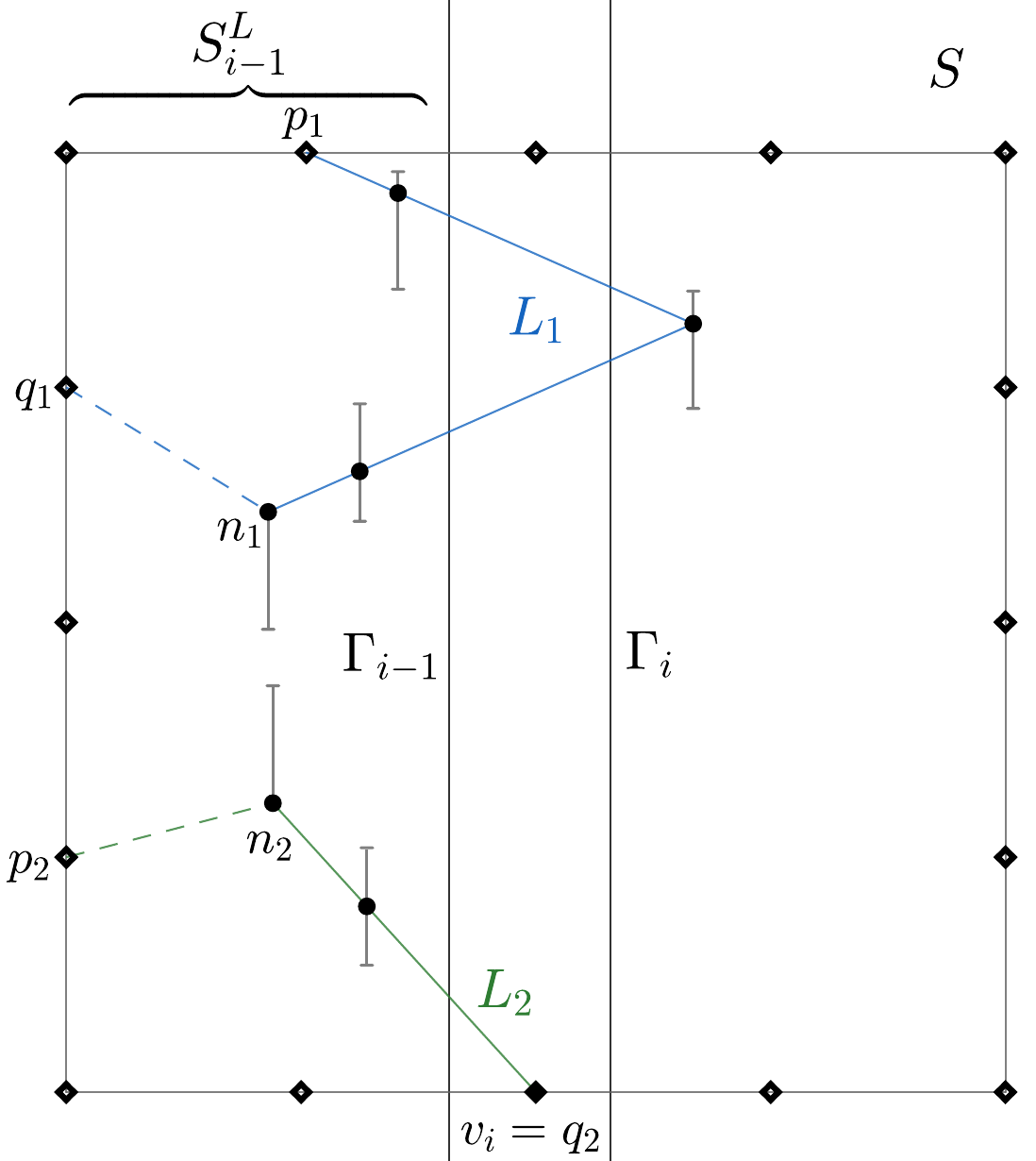}
    \caption[An example of an event point $v_i$ and consecutive vertical lines $\Gamma_{i - 1}, \Gamma_i$.]{An example of an event point $v_i$ and vertical lines $\Gamma_{i - 1}, \Gamma_i$ from two consecutive equivalent classes in square $S$. In this figure, $\mathcal L_{i - 1} = L_1, L_2$ and $\mathcal L_i = L_1$; plus, it is the case that $(p_2, n_2)\in P_{i - 1}^{L}, (q_1, n_1)\in P_{i - 1}^{L},$ and $(q_1, n_1)\in P_{i}^{L}$}
    \label{fig:event-point}
\end{figure}

\noindent The cost of this solution is defined to be the sum of the costs of all the edges that are entirely (i.e. both end-points) in $S^L_i$ (including those legs of a large leg in $\mathcal L_i$ that are entirely in $S^L_i$, but not those that are crossing $\Gamma_i$).
Entry $A[i,\mathcal L_i,P^L_i]$ of the inner DP, stores the minimum cost of such a solution. Recall that there are $n_s=O(n)$ choices for $i$ (and so for $\Gamma_i$), $n^{O(1/\eps^3)}$ choices for $\mathcal L_i$, and $(1/\eps)^{O(1/\eps)}$ choices for $P^L_i$. So there are $n^{O(1/\eps^3)}$ possible configurations, which is the size of our DP table as well.

We fill in the entries of this table $A[.,.,.]$ for increasing values of $i$.
For $i=O(1)$, $A[i,.,.]$ can be computed exhaustively in $O(1)$ time.

For any other value of $i$, we compute $A[i,\mathcal L_i,P^L_i]$ by considering various subproblems 
$(i-1,\mathcal L_{i-1},P^L_{i-1})$ that are {\em consistent} (formally explained soon) with $(i,\mathcal L_{i},P^L_{i})$.
Consider event point $v_{i-1}$; it is either a segment or a portal that is between $\Gamma_{i-1}$ and $\Gamma_i$; which  means it does not belong to $S^L_{i-1}$, but belongs to $S^L_i$. Consider the solution for $(i,\mathcal L_i,P^L_i)$, and the legs (in that solution) that visit $v_i$. In case $v_i$ is a start/end
portal in $P$, there is one leg incident to $v_i$; if $v_i\in Q$ there are two legs incident to $v_i$, and
if $v_i$ is a segment, there are two legs that are incident to a point $v'_i$ on that segment. 
If there is one leg only ($v_i$ is a start/end portal), call that leg $\ell_i$, and if there are two legs, call them $\ell_{i-1},\ell_i$. Depending on whether these legs cross $\Gamma_{i-1}$ or $\Gamma_i$, we have the following situations, which are the {\em consistent} outcomes:
\begin{enumerate}
    \item \textbf{$v_i$ is a start/end portal}, we consider 2 different subcases:
    \begin{enumerate}
        \item \textbf{$\ell_i$ crosses $\Gamma_{i-1}$ but not $\Gamma_i$}:
        Say $\ell_i=v_iu$, where $u$ is a point in $S^L_{i-1}$. 
        In this case, there is a large leg $L\in\mathcal L_{i-1}$ with one end-point $v_i$. Then if $L$ crosses $\Gamma_i$, it means $L$ is a large leg in $\mathcal L_i$. If $L$ does not cross $\Gamma_i$, then $\mathcal L_i=\mathcal L_{i-1} \setminus L$. We consider both possibilities and in each case, consider $P^L_{i-1}$'s that are consistent %\todo{needs explanation} 
        with $P^L_i$ and 
        set $A[i,\mathcal L_i,P^L_i]=\min_{P^L_{i-1},\mathcal L_{i-1}} \{A[i-1,\mathcal L_{i-1},P^L_{i-1}]\}+||\ell_i||$.
        
        \item \textbf{$\ell_i$ crosses $\Gamma_i$ but not $\Gamma_{i-1}$:}
        In this case, there is a large leg $L\in\mathcal L_i$ that starts with $\ell_i$ and does not cross $\Gamma_{i-1}$, so does not belong to $\mathcal L_{i-1}$. All the other large legs in $\mathcal L_{i-1}$ and $\mathcal L_i$ are the same (as there is no other event point between $\Gamma_{i-1}$ and $\Gamma_i$), and $P^L_i$ and $P^L_{i-1}$ are consistent. Then 
        $A[i,\mathcal L_i,P^L_i]=\min_{P^L_{i-1},\mathcal L_{i-1}} \{A[i-1,\mathcal L_{i-1},P^L_{i-1}]\}$.
    \end{enumerate}
    \item \textbf{$v_i\in Q$}, we consider 3 different subcases:
    \begin{enumerate}
        \item \textbf{$\ell_{i-1},\ell_i$ both cross $\Gamma_{i-1}$ but not $\Gamma_i$:} In this case,
        there are two large legs $L,L'\in\mathcal L_{i-1}$ that both end at $v_i$, say $L$ contains $\ell_i$ and $L'$ contains $\ell_{i-1}$. If $L$ crosses $\Gamma_i$, then $L$ is a large leg in $\mathcal L_i$ as well, similarly for $L'.$ The other large legs of $\mathcal L_i$ and $\mathcal L_{i-1}$ are the same, and $P^L_{i-1}$ is consistent with $P^L_i$. We set $A[i,\mathcal L_i,P^L_i]= \min_{P^L_{i-1},\mathcal L_{i-1}}\{A[i-1,\mathcal L_{i-1},P^L_{i-1}]\}+||\ell_i||+||\ell_{i-1}||$.

        \item \textbf{$\ell_{i-1},\ell_i$ both cross $\Gamma_i$ but not $\Gamma_{i-1}$:} This similar to the previous case. There are two legs $L,L'\in\mathcal L_i$ that both start at $v_i$, say $L$ contains $\ell_i$ and $L'$ contains $\ell_{i-1}$. If $L$ crosses $\Gamma_{i-1}$, then $L$ is a large leg in $\mathcal L_{i-1}$ as well, similarly for $L'$. The other large legs of $\mathcal L_i$ and $\mathcal L_{i-1}$ are the same and $P^L_{i-1}$ is consistent with $P^L_i$.  In this case, $A[i,\mathcal L_i,P^L_i]=\min_{P^L_{i-1},\mathcal L_{i-1}} \{A[i-1,\mathcal L_{i-1},P^L_{i-1}]\}$.
        
        \item \textbf{Exactly one of $\ell_{i-1},\ell_i$ crosses $\Gamma_{i-1}$ and one crosses $\Gamma_i$:}
        Say $\ell_{i-1}$ crosses $\Gamma_{i-1}$, and $\ell_i$ crosses $\Gamma_i$. So $\ell_{i-1}$ will be the last leg of a large leg $L\in\mathcal L_{i-1}$, and $\ell_i$ will be the first leg of a large leg $L'\in\mathcal L_i$. If $L$ does not cross $\Gamma_i$, then $L$ is not in $\mathcal L_i$ at all. Similarly, if $L'$ doesn't cross $\Gamma_{i-1}$, then $L'$ isn't a large leg in $\mathcal L_{i-1}$. We consider both possiblities (i.e. consider sets $\mathcal L_{i-1}$ that are consistent with one of these cases).
        $A[i,\mathcal L_i,P^L_i]=\min_{P^L_{i-1},\mathcal L_{i-1}} \{A[i-1,\mathcal L_{i-1},P^L_{i-1}]\}+||\ell_{i-1}||$.
    \end{enumerate}
    \item \textbf{$v_i$ is a segment}: Subcases are similar to the previous case; let $ v_i' $ be the intersection point of $ \OPT_S $ with $ v_i $:
    \begin{enumerate}
        \item \textbf{$\ell_{i-1},\ell_i$ both cross $\Gamma_{i-1}$ but not $\Gamma_i$:}
        If $v'_i$ is a tip, then $\ell_{i-1}$ is the last leg of a large leg $L\in\mathcal L_{i-1}$, and $\ell_i$ 
        is the last leg of another large leg $L'\in\mathcal L_{i-1}$. Depending on whether $L$ ($L'$) crosses $\Gamma_i$, it can be a large leg in $\mathcal L_i$ or not. We consider both possibilities.
        If $v'_i$ is not a tip, then it must be a pure reflection, so there must be a large leg $L\in\mathcal L_{i-1}$ that contains this as a pure reflection. That large leg may or may not belong to $\mathcal L_i$. We consider all these possibilities (i.e. those $\mathcal L_{i-1}$ consistent with these), and also for each case consider a $P^L_{i-1}$ consistent with $P^L_i$. Then set         
        $A[i,\mathcal L_i,P^L_i]=\min_{P^L_{i-1},\mathcal L_{i-1}} \{A[i-1,\mathcal L_{i-1},P^L_{i-1}]\}+||\ell_{i-1}||
        +||\ell_i||$.
        \item \textbf{$\ell_{i-1},\ell_i$ both cross $\Gamma_i$ but not $\Gamma_{i-1}$:}
        If $v'_i$ is a tip, then $\ell_{i-1}$ is the first leg of a large leg $L\in\mathcal L_{i}$, and $ \ell_i$ 
        is the first leg of another large leg $L'\in\mathcal L_{i}$. Depending on whether $L$ ($L'$) crosses $\Gamma_{i-1}$, it can be a large leg in $\mathcal L_{i-1}$ or not. We consider both possibilities.
        If $v'_i$ is not a tip, then it must be a pure reflection, so there must be a large leg $L\in\mathcal L_{i}$ that contains this as a pure reflection. That large leg may or may not belong to $\mathcal L_{i-1}$ depending on whether it crosses $\Gamma_{i-1}$ or not. We consider all these possibilities, and also for each case consider a $P^L_{i-1}$ consistent with $P^L_i$. Then set         
        $A[i,\mathcal L_i,P^L_i]=\min_{P^L_{i-1},\mathcal L_{i-1}} \{A[i-1,\mathcal L_{i-1},P^L_{i-1}]\}$.
        
        \item \textbf{Exactly one of $\ell_{i-1},\ell_i$ crosses $\Gamma_{i-1}$ and one crosses $\Gamma_i$:}
        In this case, $v'_i$ must be a tip or a straight point. Say $\ell_{i-1}$ crosses $\Gamma_{i-1}$,
        and $\ell_i$ crosses $\Gamma_i$. If $v'_i$ is a tip, then $\ell_{i-1}$ is the last leg of a large leg
        $L\in\mathcal L_{i-1}$, and $\ell_i$ is the first leg of a large leg $L'\in\mathcal L_i$.
        $L$ may cross $\Gamma_i$ (in which case it also belongs to $\mathcal L_i$), also $L$ may cross $\Gamma_{i-1}$ in which case belongs to $\mathcal L_{i-1}$. We consider these possibilities. If $v'_i$ is a straight point, then both $\ell_{i-1},\ell_i$ are part of a large leg $L\in\mathcal L_{i-1}$, and $L$ belongs to $\mathcal L_{i}$ as well. We consider all these cases and consistent $P^L_{i-1},P^L_i$ and set
        $A[i,\mathcal L_i,P^L_i]=\min_{P^L_{i-1},\mathcal L_{i-1}} \{A[i-1,\mathcal L_{i-1},P^L_{i-1}]\}+||\ell_{i-1}||$.
    \end{enumerate} 
\end{enumerate}
\subsubsection*{Consistent Subproblems}\label{def:consistent}
	The consistency of a subproblem by configuration $ (i, \mathcal L_i, P^L_i) $, with a previous subproblem by configuration $ (i-1, \mathcal L_{i - 1}, P^L_{i - 1}) $, comes down to one of the cases mentioned in the previous section. In each subcase, we only need to define what we mean by consistent between $ P^L_i $ and $ P^L_{i - 1} $.
	
	We say $ P^L_i $ as a part of the configuration $ (i, \mathcal L_i, P^L_i) $, and $ P^L_{i - 1} $ as a part of the configuration $ (i-1, \mathcal L_{i - 1}, P^L_{i - 1}) $ are consistent if for any pair $ (a, b)\in P^L_i $:
	\begin{itemize}
	%	\item If there is a point $ p $ in some pair $ (a, b)\in P^L_i $, and if $ p $ is an end point of some large leg $L_j \in \mathcal L_i \cap \mathcal L_{i-1} $, then $ (a, b)\in P^L_{i - 1} $.
		\item If both $ a, b $ are in $ S_{i - 1}^L $, then either:
			\begin{itemize}
				\item $ (a, b) \in P^L_{i-1} $, or
				\item (When $ v_i \in Q $ or when $ v_i $ is a segment containing a pure reflection) There is a large leg $ L_j\in \mathcal L_{i - 1} \cup \mathcal L_{i} $ with end points $ p_1, p_2 $ corresponding to (i.e. having an intersection with) the event point $ v_i $,
					such that $ (a, p_1), (b, p_2)\in P^L_{i - 1} $, or
				\item (When $ v_i\in P $ or $ v_i $ is a segment containing a non-pure reflection or a break point)
				There are two large legs (in $ \mathcal L_{i - 1} \cup \mathcal L_{i} $) that have $ v_i $ as an end point, and have another end point, say respectively $ p_1 $ and $ p_2 $, such that  $ (a, p_1), (b, p_2)\in P^L_{i-1} $.
				
			\end{itemize}

		\item If both $ a, b $ are not in $ S_{i - 1}^L $, then it means that either $ a $ or $ b $, say $ a $, corresponds to the event point $ v_i $. This means either $ a $ is a portal ($ \in P\cup Q $) between $ \Gamma_{i-1} $ and $ \Gamma_i $, or $ a $ is a tip of the segment corresponding to $ v_i $. In either case, there is at least a large leg $ L_j \in \mathcal L_{i - 1} \cup \mathcal L_{i} $ that has $ a $ as one of its end points. There can be at most two such large legs; say $ p_1 $ and possibly $ p_2 $ are the other ends of these at most two large legs.
		Then 
%		if there are no large legs with end points $ \{p_1, b\} $ (or $ \{p_2, b\} $), 
		it must be the case that
		(either) $ (b, p_1) \in P^L_{i - 1} $ (or $ (b, p_2) \in P^L_{i - 1} $).
	\end{itemize}

\subsection{Wrap-up: Proof of Theorem \ref{thm:main}}\label{sec:main-theo-proof}

We finalize the proof of our algorithm for the case of unit length segments, and then extend the proof to the case of similar-length segments.

\begin{theorem}\label{thm:main2}
    There is a $ (1+\eps) $-approximation algorithm for TSPN over $ n $ vertical unit-length line segments that runs in time $n^{O(1/\eps^3)}$.
\end{theorem}

\begin{proof}
	Take any instance of the problem.
	As described at the beginning of this section,  we first scale the instance (at a loss of $(1+\eps)$)
	so that all segments have integer coordinates. We employ the hierarchical decomposition of Arora using dissecting lines
	as described in Subsection \ref{sec:dissection}, and drop the line segments crossing horizontal dissecting lines as described in Subsection \ref{subsec:dropping-horizontal}. We require a subset of
	portals around each square $ S $ of the dissectioning to be covered in the subproblems as described in the outer DP in Subsection \ref{sec:outer-dp}. Lemma \ref{lemma:horizontal-cover-lines} shows that we lose at most another $(1+\eps)$ factor in doing so. At the leaf level of our decomposition, we need to solve instances where
	each square has sides of length $\rho\cdot h$. Note that as discussed in the first paragraph of Subsection \ref{sec:inner-dp}, for any base square of the dissection, using Theorem \ref{thm:onestrip-shadow}, Lemma \ref{lemma:horizontal-cover-lines}, Lemma \ref{lemma:p-tour-r-light}, and Lemma \ref{lemma:reflections}, there is a near-optimum solution such that it is portal respecting, $r$-light for $ r = O(\frac1\eps) $, covers all the required portals, has a shadow bounded by $O(1/\eps^2)$, and the length
	of any pure reflection sequence in it is bounded by $O(1/\eps)$. The inner DP describes how to find such a solution. Note that the size of the inner DP table is $n^{O(1/\eps^3)}$. To compute each entry, we may consider (at worst) all other entries, and so the time complexity of computing the table for each
	square $S$ is at most $n^{O(1/\eps^3)}$. Given that the number of squares at the leaf nodes of the decomposition is $O(n\log^{O(r)} n)$, the total time for the inner and outer DP is $n^{O(1/\eps^3)}$.
\end{proof}

%\subsubsection{Similar-Length Line Segments (Proof of Theorem \ref{thm:main})}\label{sec:main-theo-proof}
% We finally prove the main theorem in this paper, Theorem \ref{thm:main}:
%We discuss how the result presented for unit length line segments can be extended to when line segments have
%length ratio $\lambda=O(1)$ and obtain a PTAS that runs in time $n^{O(\lambda/\eps^3)}$, and hence proving
%Theorem \ref{thm:main}.

% \begin{proof}
	We discuss how the result presented for unit-length line segments in Theorem \ref{thm:onestrip-shadow} can be extended to the case that line segments have
	length ratio $\lambda=O(1)$, and obtain a PTAS for it.
	In the case of segments with lengths in $ [1, \lambda] $, for every strip of height 1, we still have some top and bottom segments and we might have some line segments
%	(those with length larger than 1) 
	that completely span the height of the strip. Let's call these segments {\em full segments} of a strip.
	We claim that whenever we change the solution in the proof of Theorem \ref{thm:onestrip-shadow} to one that has a bounded shadow, the full segments of the strip remain covered. These changes are done in Lemmas \ref{lemma:reflections} and \ref{lem:strip-bounded-shadow}. For each of these cases, any new subpath (with smaller shadow) that replaces a subpath of larger shadow, will travel the same interval in the $x$-coordinate, and hence any full segment covered by the original path, remains covered by the new path.
	
	Next, when we scale the instance, we get line segments with length between $[\rho,\lambda\rho]$.
	Now we do our hierarchical decomposition until base squares have side length of $\lambda\rho h$, so the space between
	two cover-lines in the same group is $\lambda \rho h$ instead of $\rho h$. 
	Lemma \ref{lemma:box-size-opt-ineq} holds with  bound $\opt\geq\frac{\rho\cdot|\mathcal B|}{6\lambda}$.
	This implies Lemma \ref{lem:cover-line-group} holds if $j$ is chosen from $[1\dots h\lambda]$.
	It is straight-forward to check that Lemma \ref{lemma:horizontal-cover-lines} holds with the same ratio.
	For the inner DP, noting that the instance we start from has height $\rho\lambda h$, the shadow is bounded
	by $O(\lambda/\eps^2)$. The same DP works but the runtime will be $n^{O(\lambda/\eps^3)}$. This implies
	we get a PTAS with the same run time which completes the proof of Theorem \ref{thm:main}. \qed
% \end{proof}

%%%%%%%%%%%%%%%%%%%%%%%%%%%%%%%%%%%%%%%%%%%%%%%%%%%%%%%%%%
\section{Proofs of the Major Lemmas}
\label{sec:proofs}

In this section, we provide the proofs of the four major lemmas we had, namely Lemmas \ref{lemma:zig-zag}, \ref{lem:strip-bounded-shadow}, \ref{lemma:reflections}, and \ref{lemma:overlapping-loops-ladders}; as well
as the proof of Theorem \ref{thm:exact-answer}.

\subsection{Properties of an Optimum Solution}\label{sec:zig-zag-lem}
The goal of this subsection is to prove Lemma \ref{lemma:zig-zag}. Before that,
we need to state some further lemmas and definitions. We combine these in Subsection \ref{sec:lemma-zigzag} to build a complete proof.

\begin{definition}\label{def:exclusively-cover}
    Let $\OPT_\tau$ be the restriction of $\OPT$ to any strip $S_\tau$.
    We say a segment $s \in S_\tau$ is \textbf{exclusively covered} by some path $P\in \OPT_\tau$ if $P$ covers $s$ but no other subpath of $\OPT_\tau$ intersects with $s$, i.e. $\OPT_\tau/P$ doesn't intersect with it.
\end{definition}

\begin{lemma}\label{lem:cover-line-loop}
    Each loop with entry points on $C_{\tau+1}$ in $\OPT_\tau$ (i.e. bottom cover-line of $S_\tau$) must exclusively cover a top segment, or else it must be a cover-line loop. Analogous argument holds for loops that have entry points on $C_\tau$.
\end{lemma}
\begin{proof}
    Suppose $P$ is a loop with entry points $e_1,o_1$ on $C_{\tau+1}$ that does not exclusively cover a point on a top segment. This implies if we change it to cover only bottom segments in $S_\tau$, then the solution remains feasible. Let $s_\ell$ and $s_r$ be the left-most and right-most bottom segments that $P$ covers, let $q_\ell,q_r$ be intersections of $s_\ell$ and $s_r$ with $C_{\tau+1}$, respectively. Replace $P$ with $e_1,q_\ell,q_r,o_1$ and then short-cut $e_1,o_1$ like the way we argued for cover-line loops after Definition \ref{def:loops/ladders}.
    So we obtain a path that is shorter than the original, but is a cover-line loop and covers all the (bottom) segments $P$ was covering.
\end{proof}
\begin{corollary}\label{cor:non-cover-line-loop-top-segment}
    If $P$ is a non-cover-line loop with entry points on the bottom cover-line of some strip $S_\tau$, then $P$ has to exclusively cover some top segment in $S_\tau$. Similar argument holds for bottom segments and non-cover-line loops with entry points on the top cover-line.
\end{corollary}

\begin{lemma}\label{lem:nested-vertical-loops}
    Suppose that $\OPT_\tau$ is crossing a vertical line $\Gamma$ at least two times. Let $p_1,p_2$ be two such crossings
    and, $L_1$ be a subpath of $\OPT_\tau$ from $p_1$ to $p_2$ with no other crossings with $\Gamma$. Then there cannot be any other crossings of 
    $\OPT_\tau$ with $\Gamma$ on the section $p_1p_2$ of $\Gamma$.
\end{lemma}

\begin{proof}
    Without loss of generality, since $ L_1 $ doesn't intersect with $ \Gamma $ other than at points $ p_1 $ and $ p_2 $, assume that $L_1$ is on the left of $\Gamma$.
    By way of contradiction, suppose $q_1$ is another crossing of $\OPT_\tau$ with $\Gamma$ such that $y(p_1)<y(q_1)<y(p_2)$. This implies that there is a subpath of $\OPT_\tau$ inside the region $A=L_1\cup p_1p_2$ with one end-point being $q_1$. So there must be another crossing of $\OPT_\tau$ with the region $A=L_1\cup p_1p_2$; and since $\OPT_\tau$ is not self-crossing, that other crossing point with $A$ must be on $p_1p_2$, call it $q_2$. Let us denote the subpath of $\OPT_\tau$ inside $A$ with end-points $q_1,q_2$ by $L_2$.
    Let $r_1$ be the left-most point on $L_1$. Since $L_1$ is a path from a point on $\Gamma$ to the left of $\Gamma$ and back to a point on $\Gamma$, using Lemma \ref{obs:further-point-reflection}, $r_1$ must be a right reflection point.
    Similarly, if $r_2$ is the left-most point on $L_2$ then $r_2$ must be a right reflection point, say on segment $s_{r_2}$ (see Figure  \ref{fig:fig6}). But since $r_2$ is inside $A$, then regardless of whether $s_{r_2}$ is a top segment or a bottom segment it will intersect with $L_1$, contradicting Lemma \ref{obs:reflection}.
\end{proof}

\begin{figure}[h]
    \centering
    \includegraphics[width=0.24\tw]{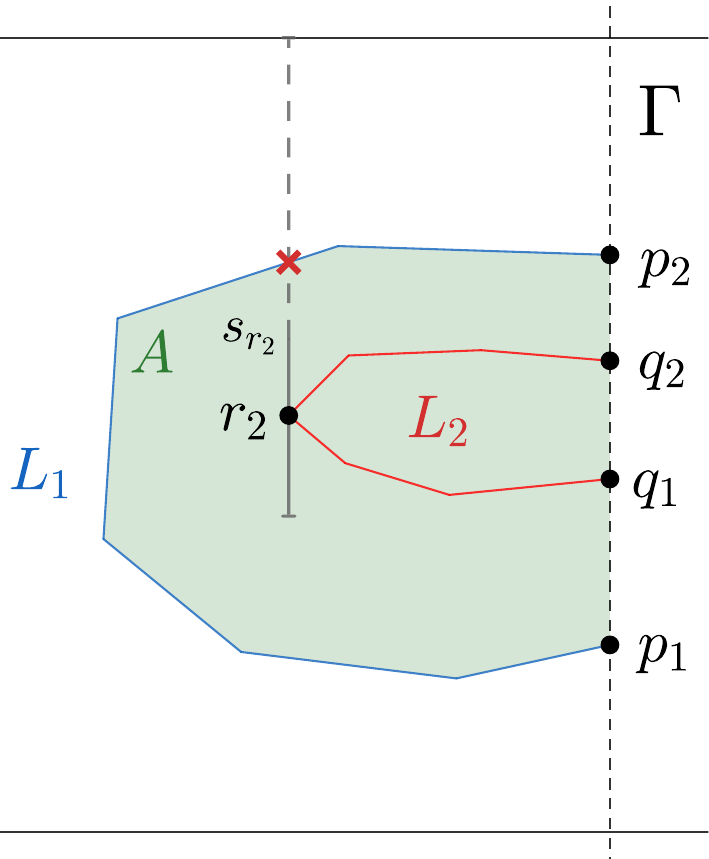}
    \caption{Configuration for Lemma \ref{lem:nested-vertical-loops}}\label{fig:fig6}
\end{figure}

The following lemma is a special case of Lemma \ref{lem:nested-vertical-loops}, but since it is used frequently, we state it as a separate lemma.

\begin{lemma}\label{claim:reflection}
	Consider a strip $S_\tau$ and $\OPT_\tau$ (the restriction of $\OPT$ within this strip). Let $s$ be any segment in this strip which has a reflection point $p_j$ on it. Without loss of generality, assume $s$ is a top segment and $p_j$ is a left reflection point. Let $\ell_u$ and $\ell_l$ be the upper and lower legs of $\OPT_\tau$ incident with $p_j$.
	Then the subpath of $\OPT_\tau$ starting at $p_j$ and travelling on $\ell_u$, will not reach to the right side of $s$.
\end{lemma}

\begin{proof}
	Suppose the subpath of $\OPT_\tau$ starting at $p_j$ and travelling along $\ell_u$, call it $P_u$, reaches
	the right side of $s$ while entirely within strip $S_\tau$. So $ P_u $ crosses the vertical line $x=x(s)$ at a point $p$ inside $S_\tau$ (different from $p_j$). This path will be $L_1$ in the setting of Lemma \ref{lem:nested-vertical-loops} and $p_j,p$ will be $p_1,p_2$ of the lemma. Consider the subpath $P_l$ of $\OPT_\tau$ starting at $p_j$ and following $\ell_l$. This subpath is in the region defined by $P_u$ and the vertical line at $x=x(s)$. Since $\OPT$ is non-self-crossing, $P_l$ has to exit this area between the lower tip of $s$ and point $p$. But this will violate Lemma \ref{lem:nested-vertical-loops}. This contradiction results in the statement of the lemma.
\end{proof}

\begin{figure}[h]
	\centering
    \includegraphics[width=.32\textwidth]{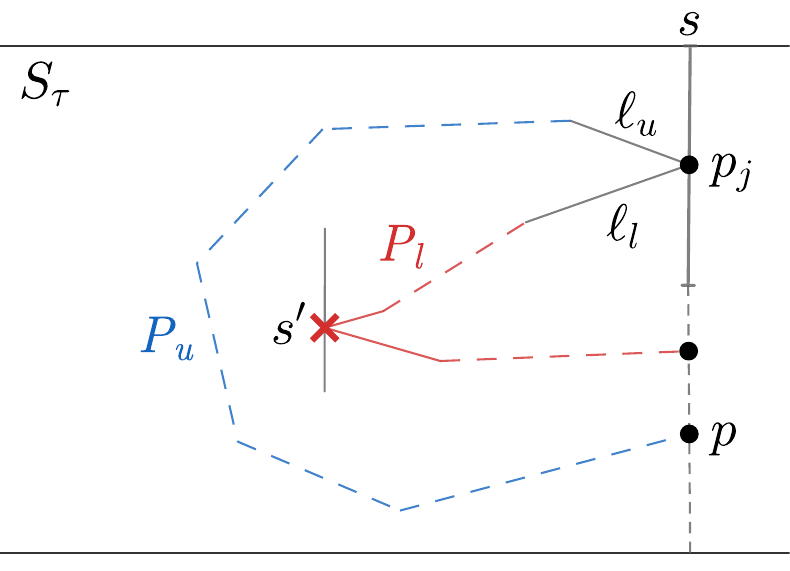}
    \caption{In a strip $S_\tau$, the path from the upper leg of a left reflection on a top segment, can't reach to the right of that segment}
    \label{fig:l_u-l_l}
\end{figure}

\begin{lemma}\label{lemma:left-path-is-above}
    Suppose $P_1$ and $P_2$ are two ladders of $\OPT_\tau$ in $S_\tau$ with entry points $e_1$ and $e_2$ on the bottom cover-line and entry points
    $o_1,o_2$ on the top cover-line, respectively, such that $x(e_1) < x(e_2)$, $x(o_1)<x(o_2)$ and both intersect a vertical line $\Gamma$ to the right of $e_1,e_2$. Then $P_1$ is above $P_2$ to the left of $\Gamma$. (symmetric arguments apply to the top cover-line, as well as entry points to the right of $\Gamma$)
\end{lemma}

\begin{proof}
	By way of contradiction, suppose $P_1$ is not above $P_2$ on the left of $\Gamma$, so there is a vertical line $\Gamma'$ to the left of $\Gamma$ whose top-most intersection is with $P_2$, say point $p$ on $\Gamma'$.
	Consider the (vertical) segment of $\Gamma'$ from $p$ to the top cover-line, call it $\Gamma''$ and let
	the subpath of $P_2$ from $e_2$ to $p$ be called $P'_2$. If we cut the strip $S_{\tau}$ along $P'_2\cup \Gamma''$, then $e_1$ is on one side, and $o_1$ on the other, which implies $L_1$ must be crossing $P'_2\cup\Gamma''$, which would be a contradiction (as it would have an intersection point on $\Gamma'$ higher than $p$ or has to cross $P'_2$).
\end{proof}

\begin{lemma}\label{lemma:in-between-reflections}
	Let $P$ be any ladder or loop of $\OPT_\tau$ in strip $S_\tau$. Let $r_{i_1}$ (on segment $s_{m_1}$) and $r_{i_2}$ (on segment $s_{m_2}$) and $r_{i_3}$ (on segment $s_{m_3}$) be any three consecutive reflections in the orientation of $\OPT_\tau$ in that order. If $x(r_{i_2}) < x(r_{i_1}) < x(r_{i_3})$ and $r_{i_2}$ is an ascending reflection, then $s_{m_1}$ is a bottom segment and $r_{i_1}$ is an ascending reflection. Symmetric argument applies for $r_{i_2}$ being a descending reflection (for which case $s_{m_1}$ will be a top segment and $r_{i_1}$ will be descending). 
\end{lemma}

\begin{proof}
	See Figure \ref{fig:lemma-3reflection}. According to Lemma \ref{obs:consecutive-reflections}, since $r_{i_1}$ and $r_{i_2}$ are consecutive reflections with $x(r_{i_2}) < x(r_{i_1})$, then $r_{i_1}$ is a left reflection and $r_{i_2}$ is a right reflection. Let $P_{1,2}$ be the subpath of $P$ from $r_{i_1}$ to $r_{i_2}$, and $P_{2,3}$ be the subpath of $P$ from $r_{i_2}$ to $r_{i_3}$. Since $r_{i_2}$ is an ascending reflection, then $P_{1,2}$ contains the lower leg of $r_{i_2}$, and $P_{2,3}$ contains the upper leg of $r_{i_2}$.\\
	Since $r_{i_1}$ and $r_{i_2}$ are two consecutive reflections with $x(r_{i_1}) > x(r_{i_2})$, this means that $P_{1,2}$ cannot reach to the left of $r_{i_2}$ or to the right of $r_{i_1}$; because otherwise, due to the difference in the $x$-coordinates, $P_{1,2}$ would require an additional reflection between $r_{i_1}$ and $r_{i_2}$, which isn't possible.\\
	This implies that the entirety of $P_{1,2}$, and specifically $r_{i_1}$, are in the region defined by $x=x(r_{i_2})$, $x=x(r_{i_1})$, and the path $P_{2,3}$. So $P_{1,2}$ is below $P_{2,3}$ in $I=[x(r_{i_2}),x(r_{i_1})]$.

\begin{figure}[h]
	\centering
	\includegraphics[width=.34\textwidth]{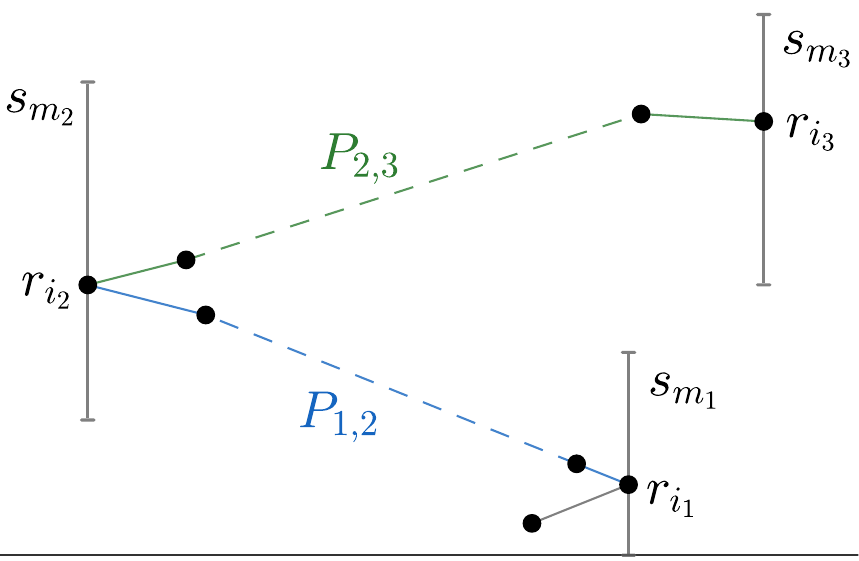}
	\caption[Valid arrangement of three consecutive reflections given a specific configuration.]{Valid arrangement of three consecutive reflections provided the $x$-coordinate of $r_{i_1}$ is between the $x$-coordinates of $r_{i_2}$ and $r_{i_3}$, and $r_{i_2}$ is an ascending reflection. Segments $s_{m_2}$ and $s_{m_3}$ could either be top or bottom segments in this strip, but $s_{m_1}$ must be a bottom segment}\label{fig:lemma-3reflection}
\end{figure}
	
	Since $x(r_{i_2}) < x(r_{i_3})$ and $P_{2,3}$ is a path between these two reflections, we get that for any $x_0 \in [x(r_{i_2}), x(r_{i_3})]$, there is an intersection between $x = x_0$ and $P_{2,3}$. Now, for the sake of contradiction, assume $s_{m_1}$ is a top segment. Since $r_{i_1}$ is below $P_{2,3}$, this would imply that $s_{m_1}$ is intersecting with $P_{2,3}$. But this is in violation with Lemma \ref{obs:reflection}. Thus, $s_{m_1}$ must be a bottom segment. According to Lemma \ref{claim:reflection}, $r_{i_1}$ cannot be a descending reflection, because otherwise, $P_{1,2}$ would contain the lower leg of $r_{i_1}$; therefore, the path $P_{1,2}\cup P_{2,3}$ is a path that contains the lower leg of $r_{i_1}$ and reaches to the right of segment $s_{m_1}$, which isn't possible.
	So we conclude that $s_{m_1}$ is a bottom segment and furthermore, $r_{i_1}$ is an ascending reflection.
\end{proof}
%\end{proof}
\begin{lemma}\label{lem:reflection-M}
    Suppose $P$ is a loop or ladder of $\OPT_\tau$ for a strip $S_\tau$ and $r_{i_1},r_{i_2},r_{i_3}$ are three reflection points
    visited in this order, but not necessarily consecutively (following orientation of $\OPT$), all are ascending (or all are descending) and are on segments $s_{m_1},s_{m_2},s_{m_3}$, respectively. Assume that $r_{i_1},r_{i_3}$ are left reflections and $r_{i_2}$ is a right reflection
    and $r_{i_2}$ is to the left of both $r_{i_1}$ and $r_{i_3}$, i.e. $x(s_{m_2})<x(s_{m_1})$ and $x(s_{m_2})< x(s_{m_3})$.\\
    Let $P_{0,1}$ be the subpath of $P$ up to $r_{i_1}$, $P_{1,2}$ be the subpath of $P$ from $r_{i_1}$ to $r_{i_2}$, $P_{2,3}$ be the subpath of $P$
    from $r_{i_2}$ to $r_{i_3}$, and $P_{3,4}$ be the subpath of $P$ from $r_{i_3}$ to the end of $P$.
    Then we cannot have both $P_{0,1}$ and $P_{3,4}$ reach to the left of $x(s_{m_2})$.
\end{lemma}

\begin{proof}
    Each of $P_{1,2}$ and $P_{2,3}$ include a leg of $r_{i_2}$; Without loss of generality, assume that the lower leg of $r_{i_2}$ is in $P_{1,2}$, and its upper leg is in $P_{2,3}$ (i.e. assume that $r_{i_2}$ is an ascending reflection). We take two cases based on whether $s_{m_2}$ is a top segment or a bottom segment:
    \begin{itemize}
        \item \textbf{$s_{m_2}$ is a top segment:}
            Path $P^u_2 = P_{2,3} \cup P_{3,4}$ includes the upper leg of $r_{i_2}$ (a right reflection) on a top segment $s_{m_2}$, so we can use the result of Lemma \ref{claim:reflection} to conclude that $P^u_2$ and particularly $P_{3,4}$ can't reach to the left of $s_{m_2}$.
            
        \item \textbf{$s_{m_2}$ is a bottom segment:}
            Path $P^l_2 = P_{0,1}\cup P_{1,2}$ includes the lower leg of $r_{i_2}$ (a right reflection) on a bottom segment $s_{m_2}$. Again, using Lemma \ref{claim:reflection}, we get the same result that $P^l_2$ and consequently $P_{0,1}$ can't reach to the left of $s_{m_2}$.
    \end{itemize}
\end{proof}

\begin{figure}[h]
    \centering
    \includegraphics[width=.24\tw]{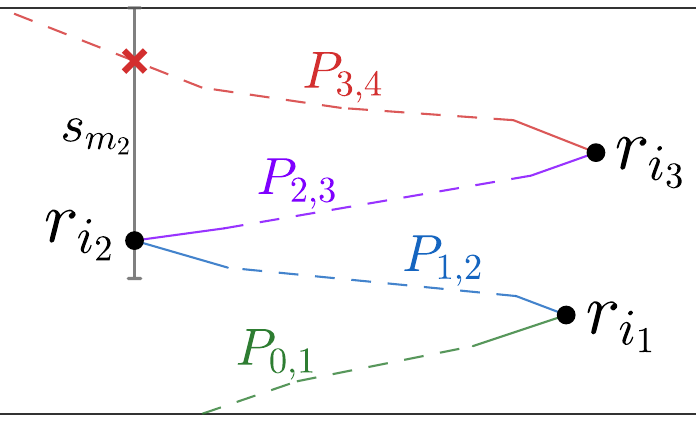}
    \caption{Configuration of Lemma \ref{lem:reflection-M} when $s_{m_2}$ is a top segment}
\end{figure}

\subsubsection{Proof of Lemma \ref{lemma:zig-zag}: Vertically Partitioning the Solution in Each Strip} \label{sec:lemma-zigzag}
This subsection is dedicated to the proof of Lemma \ref{lemma:zig-zag}. 
As mentioned before, we show that the vertical line at which the largest shadow for $P$ happens, can intersect with at most two sinks and a zig-zag. So the shadow of $P$ is within $O(1)$ of the maximum shadow of the zig-zags and sinks along it.

If $q\le 2$, then the correctness of lemma follows easily; so let's assume $q > 2$.
Starting from $i=1$, find the largest $j$ such that the sequence
$r_i,\ldots,r_j$ are all monotone, i.e. are all ascending or all are descending reflections. This will be the first part. We set $i=j+1$ and, again, find the largest $j$ such that $r_i,\ldots,r_j$ are monotone; this becomes the 2nd part. We repeat this procedure.

So we find a partition into maximal sub-sequences of consecutive reflection points $r_a$'s such that each sub-sequence contains only ascending or only descending reflections; each part might have just one point and the subpath path between the last reflection point of a part and the first reflection point of the next part has shadow 1 (since change in shadow can only happen if there is a reflection point by Lemma \ref{obs:shadow-increase}).

The proof has two parts, we first show that each part can only have up to two sinks and possibly a zig-zag in between them, and then we show that for any vertical line, there is at most one part intersecting with it. Since the subpath from one part to another is a path between two consecutive reflections (and has shadow 1) the statement of the lemma will follow.
We will use the following claim throughout this proof:
\begin{claim}\label{claim:bottom-bottom}
    For any subpath $P$ of $\OPT_\tau$ that is either a loop or a ladder, let $r_j$ (on segment $s_m$) and $r_{j'}$ (on segment $s_{m'}$) be any two consecutive reflections along $P$. Without loss of generality assume that $r_j$ is an ascending left reflection and in the orientation of $\OPT_\tau$, $r_j$ comes before $r_{j'}$.
    If both $s_m$ and $s_{m'}$ are bottom segments,
    then the subpath of $P$ up to $r_j$ (which  includes the lower leg $\ell_{j}$) can only contain reflection points lying on bottom segments. Analogous statement holds when both $s_m,s_{m'}$ are top segments.
\end{claim}

\begin{proof}[Proof of claim]
    According to Lemma \ref{obs:consecutive-reflections}, $r_{j'}$ is a right reflection and $s_{m'}$ is to the left of $s_m$. Let $P_j$ be the subpath of $P$ from $r_j$ to $r_{j'}$. Refer to the area of $S_\tau$ enclosed by $s_m$ and $s_{m'}$ and below $P_j$ by $A_j$; then $\ell_j$ lies inside $A_j$ (see Figure \ref{fig:bottom-bottom-claim}). 
	Let the subpath of $P$ ending with leg $\ell_j$ be called $P'_j$. So $P'_j$ is entirely within $A_j$
	as it cannot intersect with either of $s_m,s_{m'}$ (due to Lemma \ref{obs:reflection}, since they both have reflection points) and $P'_j$ cannot
	intersect $P_j$ other than at $r_j$ (since the solution is not self-crossing).
	So $P'_j$ is below $P_j$ within $A_j$. This implies any top segment that intersects $P'_j$ must also intersect $P_j$ due to Observation \ref{cor:path-above-top-segment}. So $P'_j$ cannot have a reflection on a top segment by Lemma \ref{obs:reflection}. So $P'_j$ can have
	reflection points only on bottom segments.
\end{proof}

\begin{figure}[h]
    \centering
    \includegraphics[width=0.25\textwidth]{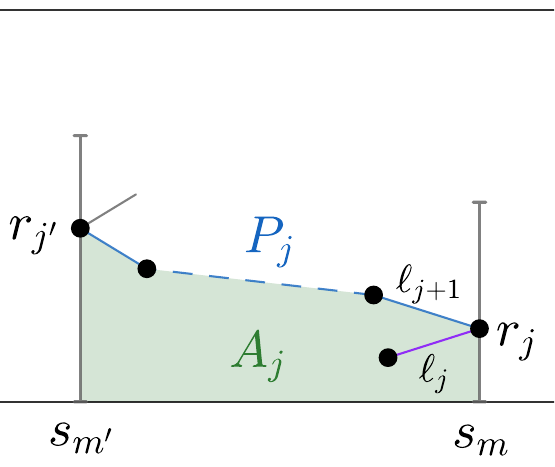}
    \caption[Area $ A_j $ used in the proof of Claim \ref{claim:bottom-bottom}.]{(Claim \ref{claim:bottom-bottom}) $A_j$ is the area in $S_\tau$ surrounded by segments $s_m, s_{m'}$ and the subpath $P_j$}
    \label{fig:bottom-bottom-claim}
\end{figure}

Now back to the proof of the lemma; we prove the following two parts:
     
\subsubsection*{Each partition can have up to two sinks and a zig-zag}
        Consider any part in the partition we defined before, which is a maximal sub-sequence of consecutive reflections that are all ascending or all descending. Our goal is to show this part is concatenation of a sink (possibly empty), followed by a zig-zag (possibly empty), followed by a sink (possibly empty), where the last
%        reflection
        point of the first sink is common with the first
%        reflection
        point of the zig-zag, and the last 
%        reflection 
		point of the zig-zag is common with the first 
%		reflection
		point of the last sink.
        For simplicity, suppose the sequence of this part is ${\mathcal R} = r_1,\dots, r_k$. Without loss of generality, assume $\mathcal R$ contains only ascending reflections and that the first one, $r_1$, is on a bottom segment. If all $r_i$'s belong to bottom segments, then $\mathcal R$ is a sink and we are done. Otherwise, let $j$ be the first index such  that $r_{j}$ is on a top segment (i.e. $r_1,\ldots,r_{j-1}$ are all on bottom segments). If $j=2$, i.e. $r_1$ was a bottom and $r_2$ is a top segment, then the first sink is empty and this part starts with a zig-zag. If $j > 2$, then $r_1, \dots, r_{j - 1}$ is a sink. We argue that starting at 
        $j' = \max\{1, j - 1\}$, we can form a zig-zag. Let $m$ be the largest index such that $r_{j'}, r_{j' + 1},\dots, r_m$ is a zig-zag, i.e. the reflection points alternate between top and bottom segments. If no such $m$ exists, it means $r_{j'}, \dots, r_k$ all belong to top segments, giving us a sink; so this together with the first possible sink gives us two sinks at most, concluding the lemma.
        If $m = k$, then the partition has (up to) a single sink followed by a zig-zag, and we're done. Otherwise, $m<k$, meaning $r_{m+1}\in \mathcal R$. Since any zig-zag has at least 2 reflections, we have $m \ge 2$, meaning $r_{m - 1} \in \mathcal R$ and since alternation between top and bottom ends at $r_m$, it means $r_m$ and $r_{m + 1}$ are both either on top segments or both on bottom segments. 
        We show that they can't be both on bottom segments. For the sake of contradiction, assume otherwise, i.e. both $r_m$ and $r_{m+1}$ are on bottom segments (and we assumed they are ascending). Also, we know $r_{m - 1}$ is on a top segment (as it must be different from $r_m$). This violates Claim \ref{claim:bottom-bottom}; because $r_{m - 1}$ is on a top segment and is on the subpath of $\OPT_\tau$ reaching $r_{m}$.
        Thus, both $r_m,r_{m+1}$ are on top segments.\\
	   Without loss of generality, assume that $r_m$ is a left reflection; Lemma \ref{obs:reflection} implies $r_{m + 1}$ is right reflection with $x(r_{m + 1}) < x(r_m)$. Let the path from $r_m$ to $r_{m+1}$ be $P_m$.
        Let $A_m$ denote the area (of $S_\tau$) bounded by $P_m$ and between the segments containing $r_m$ and $r_{m + 1}$.
        If $\ell_{m''},\ell_{m'' + 1}$ are the legs incident to $r_{m + 1}$ in the orientation of $\OPT$, then $\ell_{m'' + 1}$ lies inside $A_m$ (see Figure \ref{fig:top-top1}).
        Once again, using Claim \ref{claim:bottom-bottom}, we get that there can't be any reflections in the subpath in $P$ starting at $r_{m + 1}$ through $\ell_{m'' + 1}$ that lie on a bottom segment. This implies all of $r_m, r_{m + 1},\dots, r_k$ lie on top segments.
        Since all the remaining reflections are on top segments and all are ascending, this by definition means they form a sink. Thus, in total, we have up to a (bottom) sink, a zig-zag, and a (top) sink in this partition, concluding the first part of the proof.
\begin{figure}[H]
    \centering
    \includegraphics[width=.3\textwidth]{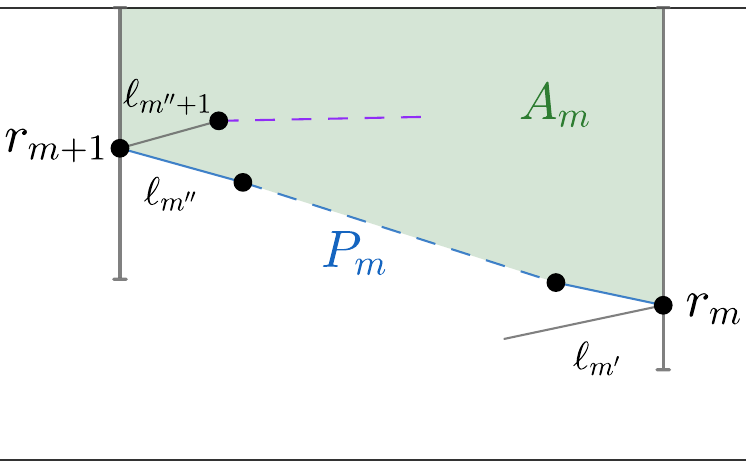}
    \caption{The upper leg of $r_{m + 1}$ lies inside $A_m$, and therefore, so does the rest of the path of $\OPT_\tau$ until $r_k$}
    \label{fig:top-top1}
\end{figure}

\subsubsection*{Any vertical line can intersect at most one part}
	Recall that $r_1, r_2, \dots, r_q$ denotes the sequence of {\em all} the reflection points on $P$ (in strip $S_\tau$). Let's call this $\mathcal R$.
	If $\mathcal R$ is made of only ascending or only descending reflections, we are done as it will have only one part in the partition. Otherwise, there must be two consecutive reflections $ r_i, r_{i + 1} \in \mathcal R$ such that one is an ascending reflection, but the other is descending. Suppose $i$ is the first index that this happens. So the subpath from $r_1$ to $r_i$ is one part, and $r_{i+1}$ is the start point of another part. Note that the path from $r_i$ to $r_{i+1}$ has no reflection points; hence has shadow 1 because of Lemma \ref{obs:shadow-increase}. 
	Without loss of generality, assume $r_i$ is a right reflection and is an ascending reflection.
	Then $r_{i+1}$ is descending and according to Lemma \ref{obs:consecutive-reflections}, it must be a left reflection as well with $x(r_i) < x(r_{i + 1})$. Let $P_i$ be the subpath of $P$ from $r_i$ to $r_{i + 1}$. 
	Since $q > 2$, we either have $ i > 1 $ or (if $i=1$ then) $ i + 1 < q $; meaning $r_{i - 1}\in \mathcal R$ or $r_{i + 2}\in \mathcal R$. In other words, there either is a reflection in $\mathcal R$ before $r_i$, or there is a reflection after $r_{i + 1}$. Assume the first case holds, similar argument applies to the second one. Since $r_i$ is a right reflection, using Lemma 
	\ref{obs:consecutive-reflections}, we get that $r_{i - 1}$ is a left reflection with $x(r_i) < x(r_{i - 1})$. We claim that we must have $x(r_{i - 1}) < x(r_{i + 1})$. For the sake of contradiction, assume otherwise. This means we have $x(r_i) < x(r_{i + 1}) < x(r_{i - 1})$. So we can use the result of Lemma \ref{lemma:in-between-reflections} with parameters being $i_1 = i + 1,\; i_2 = i,\; i_3 = i - 1$ and following the points in reverse order of orientation, i.e. $r_{i + 1}\to r_i \to r_{i - 1}$ (this is the mirrored setting of Lemma \ref{lemma:in-between-reflections}). This implies $r_{i + 1}$ must be a descending reflection in the reverse orientation, which means it must be ascending in the original orientation (that ravels $r_i$ to $r_{i+1}$). But we assumed $r_{i+1}$ is descending.  This contradicts  Lemma \ref{lemma:in-between-reflections} and proves our initial claim that $x(r_{i - 1}) < x(r_{i + 1})$. (see Figure \ref{fig:overlapping-partition1}).
	\begin{figure}[h]
		\centering
		\includegraphics[width=.25\textwidth]{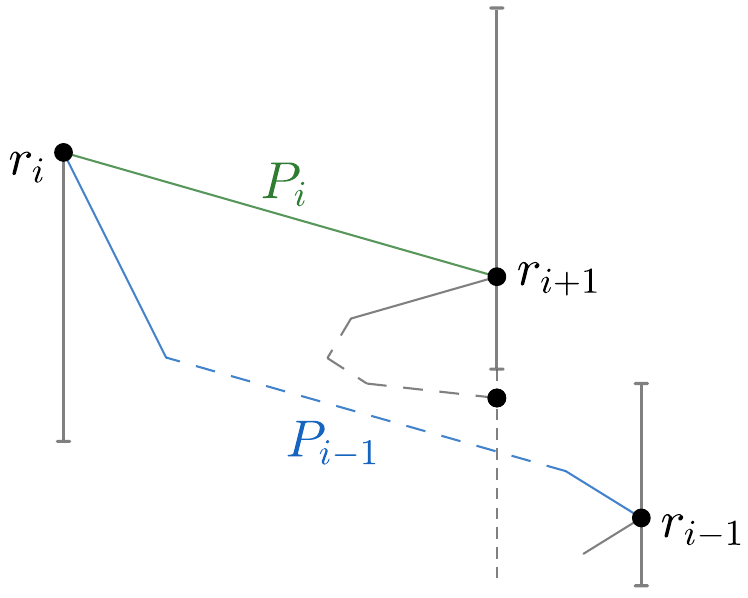}
		\caption{If between $r_i$ and $r_{i + 1}$ one is ascending and the other is descending, then $r_{i - 1}$ (or $r_{i + 2}$) must have an $x$-coordinate between $x(r_i)$ and $x(r_{i + 1})$}
		\label{fig:overlapping-partition1}
	\end{figure}
	Thus, $x(r_i) < x(r_{i - 1}) < x(r_{i + 1})$. Let $s_j$ be the segment of the instance that $r_{i - 1}$ lies on. Once again, using Lemma \ref{lemma:in-between-reflections}, we get that $r_{i - 1}$ is an ascending reflection and $s_j$ is a bottom segment (see Figure \ref{fig:overlapping-partition2}).
	
	% Since $r_i$ is an ascending reflection, we get that $P_{i - 1}$ contains the lower leg of $r_i$. Now since $P_i$ contains the upper leg of $r_i$, we get that $P_i$ is above $P_{i - 1}$ (see figure \ref{fig:overlapping-partition2}). 
	% This specifically implies $s_j$ is a bottom segment, because otherwise it would intersect with $P_i$, which is a contradiction with Observation \ref{obs:reflection}.
	According to Lemma \ref{claim:reflection}, we get that the subpath of $P$ from $r_{1}$ to $r_{i - 1}$ (since it contains the lower leg of $r_{i - 1}$ due to it being an ascending reflection) can't reach to the right of $s_{j}$.      
	
	We will show that the subpath of $P$ from $r_{i + 1}$ to $r_k$ will not reach to the left of $s_j$ either. This implies no vertical line can at the same time cross the first part that ends at $r_i$ and the other parts starting at $r_{i+1}$ onward. Repeating this argument implies no vertical line can intersect two parts as wanted.
        Consider the area surrounded by the line $x = x(s_j)$ and $P_{i - 1}\cup P_{i}$, and refer to it by $A_{j}$. 
        Now consider the subpath of $P$ from $r_{i + 1}$ to $r_{k}$ and refer to it as $\mathcal P_{i + 1}$. Similar to the proof of Lemma \ref{claim:reflection}, $\mathcal P_{i+1}$ can't enter $A_j$, because in order to exit from $A_j$, it has to reflect at some point inside $A_j$. 
        But for such a reflection point to exists, there has to be a segment containing it,
        and that segment will intersect with $P_{i-1}$ or $P_{i}$, 
        which contradicts Lemma \ref{obs:reflection}. 
        So we conclude that if $\mathcal P_{i + 1}$ were to go to the left of $s_{j}$, it has to do so from outside of $A_{j}$, i.e. from above $P_i$ (since $P_i$ is the upper hull of $A_j$).

    Take two cases based on whether the segment $s_{j'}$ that contains $r_{i+1}$ is a top segment or a bottom segment:
\begin{itemize}
    \item \textbf{$s_{j'}$ is a top segment}:\\
        The area of $S_\tau$ is cut into two parts by $P_{i-1} \cup P_i \cup s_j\cup s_{j'}$. Since $r_{i+1}$ is a descending reflection, then the lower leg of $r_{i + 1}$ is in the same part as the bottom tip of $s_{j'}$; refer to this part by $A_1$ and let $A_2$ be the other area. Since $\mathcal P_{i + 1}$ includes this leg, it means that if $\mathcal P_{i + 1}$ is going to reach to the left of $s_j$, it has to reach from $A_1$ to $A_2$. This would require it to either intersect with $P_i$ or with $s_{j'}$. The former isn't possible because it would make $\OPT$ self-crossing, and the latter isn't possible because of Lemma \ref{obs:reflection}.

    \item \textbf{$s_{j'}$ is a bottom segment}:\\
        The lower leg of $r_{i + 1}$ is in the area $A_{j'}$ surrounded by $P_{i - 1}\cup P_i\cup s_j\cup s_{j'}$. Since we mentioned $\mathcal P_{i + 1}$ can't reach inside of $A_j$, then it needs to exit $A_{j'}$ and go over $P_i$. This means $\mathcal P_{i + 1}$ has to either intersect with $P_i$ or $s_{j'}$, which gives us the same contradictions as above.
\end{itemize}

	\begin{figure}[H]
	    \centering
	    \includegraphics[width=.25\textwidth]{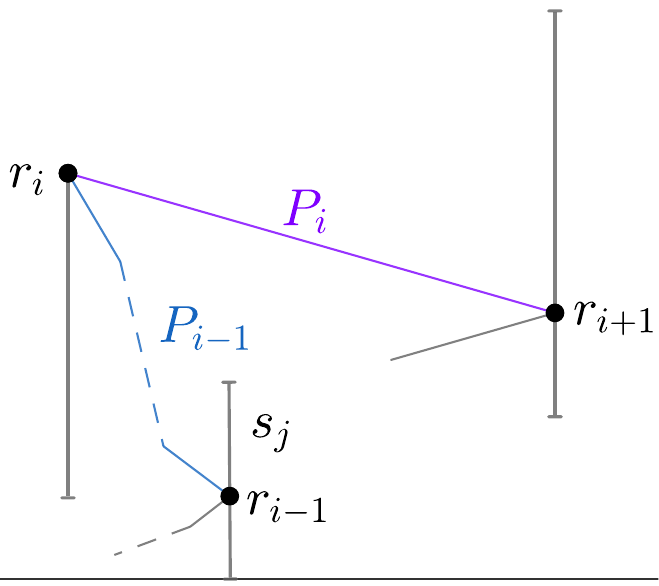}
	    \caption{If $r_i$ is an ascending reflection and $r_{i + 1}$ is a descending reflection, then $r_{i - 1}$ must be an ascending reflection lying on a bottom segment. $r_i$ and $r_{i+1}$ can be either on top or bottom segments}
	    \label{fig:overlapping-partition2}
	\end{figure}
        So we conclude that there is no vertical line $\Gamma$ that intersects both subpath $\mathcal P^-_i = \bigcup_{u = 1}^i P_u$ and subpath $\mathcal P_{i + 1}=\bigcup_{u=i+1}^k P_u$. Thus, $r_1,\dots, r_i$ gives us a partition as desired. By continuing this process for the rest of the reflections, we get that no vertical line can intersect two parts. Since the path between two conseutive parts (last reflection of one part and the first reflection of the next part) has shadow 1, this completes the proof of the last part of Lemma \ref{lemma:zig-zag}. \qed

%%%%%%%%%%%%%%%%%%%%%%%%%%%%%%%%%%%%%%%%%%%%%%%%%%
\subsection{Properties of a Near-Optimum Solution}

In this section we start with some lemmas that are needed to prove major Lemmas \ref{lem:strip-bounded-shadow}, \ref{lemma:reflections}, and \ref{lemma:overlapping-loops-ladders}.
 
\begin{lemma}\label{obs:x-order}
	Let $ \mathcal R = r_1, r_2, \dots, r_k $ denote the reflection points for any zig-zag or sink in a strip $S_\tau$ where $k\geq 3$. Without loss of generality, assume $r_1$ is on a bottom segment and is an ascending left reflection point. Then:
    \begin{itemize}
        \item If $\mathcal R$ is a zig-zag, then
            $x(r_1) < x(r_3) < \dots < x(r_{2i-1})<\ldots$ and $x(r_2) < x(r_4) < \dots < x(r_{2i})\ldots$.\\
            All the inequalities hold in the other direction if $r_1$ is a right reflection point.

        \item If $\mathcal R$ is a sink then
            $x(r_1) < x(r_3) < \dots < x(r_{2i-1})<\ldots$ and $x(r_2) > x(r_4) > \dots > x(r_{2i})\ldots$.\\ 
            Again, all the inequalities hold in the other direction if $r_1$ is a right reflection point.  
    \end{itemize}
\end{lemma}
\begin{proof}
	Assume $r_1,\ldots,r_k$ lie on segments $s_{i_1},s_{i_2},
	\ldots,s_{i_k}$, respectively. Also, let's denote the path (following the orientation of $\OPT$) from $r_m$ to $r_{m+1}$ by $P_m$. By definition, all $P_m$'s are monotone in the $x$-coordinates (see Lemma \ref{obs:shadow-increase}).
	
	First, consider the case that $\mathcal R$ is a zig-zag.
	Since $r_1$ is a left reflection and $s_{i_1}$ is a bottom segment,
	and all reflection points are ascending, it means $r_2$ is a right reflection to the left of $r_1$ (because of Lemma \ref{obs:consecutive-reflections}), and $s_{i_2}$ is a top segment. This implies $P_1$ is a decreasing path in the $x$-coordinate. 
	Once again using Lemma \ref{obs:consecutive-reflections}, since $ r_2 $ is a right reflection, we have $x(r_3) > x(r_2)$.
	We claim that $x(r_3) > x(r_1)$. If this is not the case, then we have $x(r_2) < x(r_3) < x(r_1)$. 
	Using Lemma \ref{lemma:in-between-reflections} for parameters $r_{i_1} = r_3, r_{i_2} = r_2, $ and $r_{i_3} = r_1$ in the order $r_3 \to r_2 \to r_1$ (which makes these reflections descending), implies that $s_{i_3}$ must be a top segment, which is a contradiction. So we get $x(r_3) > x(r_1)$. 
	Analogous argument shows that we must have $x(r_2)<x(r_4)$. Iteratively applying this argument establishes the inequalities.
	        
	Now consider the case that $\mathcal R$ is a sink. The argument is very similar to the case of zig-zag. Note that in this case, all the segments $s_{i_1},\ldots,s_{i_k}$ are now bottom segments, all the reflection points are ascending, and they must alternate between left and right reflection points.
	Since $r_1$ is a left reflection, $r_2$ is a right reflection with $x(r_2)<x(r_1)$ (due to Lemma \ref{obs:consecutive-reflections}). 
	We again have $x(r_3) > x(r_2)$
	because of Lemma \ref{obs:consecutive-reflections}.
	Again, if we have $x(r_3) < x(r_1)$, then we have $x(r_2) < x(r_3) < x(r_2)$.
	Using Lemma \ref{lemma:in-between-reflections} for parameters $r_{i_1} = r_3, r_{i_2} = r_2,$ and $r_{i_3} = r_1$ in the order $r_3 \to r_2 \to r_1$ (which means the reflections are descending) implies $s_{i_3}$ is a top segment, a contradiction. Thus, we get $x(r_3) > x(r_1)$.
	A similar argument shows that $x(r_2)>x(r_4)$, otherwise by an application of Lemma \ref{obs:consecutive-reflections},  $s_{i_4}$ must be a top segment, which contradicts the assumption of a sink.
	By iteratively applying the same argument, we obtain the inequalities stated.
\end{proof}

\begin{lemma}\label{lem:cons-ref}
    If $p_j,p_{j'}$ are consecutive reflection points in $\OPT$, and both are pure reflections  and all the other points of $\OPT$ in between them (if any) are straight points, then either both $p_j,p_{j'}$ are ascending or both are descending.
\end{lemma}
\begin{proof}
	By  way  of contradiction, suppose 
	$p_j$ is ascending and $p_{j'}$ is descending. Note that one is a left reflection and the other is a right reflection (as reflection points must alternate). Suppose $p_j$ is a point on segment $s_i$, and $p_{j'}$ is on segment $s_{i'}$. From the assumption, the path  from $p_j$ to $p_{j'}$ is a straight line.
	Let $\ell_j,\ell_{j+1}$ be the two legs incident to $p_j$ and $\ell_{j'},\ell_{j'+1}$ be the two legs incident to $p_{j'}$. From the definition of pure reflection, we need to have the angle between $\ell_j$ and $s_i$ and the angle between $\ell_{j+1}$ and $s_i$ be the same, and the angle between $\ell_{j'}$ and $s_{i'}$ and the angle between $\ell_{j'+1}$ and $s_{i'}$ be the same. The only way this is possible is when 
	$\ell_j,\ell_{j+1},\ell_{j'+1}$ are all horizontal but this means $\OPT$
	is self-crossing. This contradiction yeilds the result of the lemma.
\end{proof}

\begin{lemma}\label{lemma:consecutive-reflection-shadow}
    Let $\mathcal R = r_0, r_1,\dots, r_k$ be any sequence of reflections that form a sink or zig-zag in a strip $S_\tau$. For $1\le j \le k$ let $P_j$ be the subpath of $\mathcal R$ between $r_{j - 1}$ to $r_{j}$. Let $ \mathcal P = \{P_1, P_2,\dots, P_k\} $. Take any vertical line $\Gamma$ and let $P_\Gamma = \{P_{j_1}, P_{j_2},\dots, P_{j_m}\}\ (j_1< j_2< \dots < j_m)$ be the maximal subset of $\mathcal P$ that each $ P_j\in P_\Gamma $ intersect with $\Gamma$. Then $P_\Gamma$ must be a consecutive subset of $\mathcal P$. In other words, $P_\Gamma = \{P_{j_1}, P_{j_1 + 1}, P_{j_1 + 2},\dots, P_{j_1 + m - 1}\}$.
\end{lemma}
\begin{proof} 
    We say a reflection $r_j$ is {\em included} in $P_\Gamma$ if $P_j\in P_\Gamma$ or $P_{j + 1}\in P_\Gamma$. 
    We prove the following claim to use throughout this proof:
    \begin{claim}\label{claim:included-reflections}
    	There are no included left reflections to the left of $\Gamma$ (and similarly no included right reflections to the right of it). 
    \end{claim}
    \begin{proof}[Proof of Claim]
    	Assume the contrary, that there is some left reflection $r_{j_i}$ included in $P_\Gamma$ that is to the left of $\Gamma$. Without loss of generality, assume that $P_{j_i}\in P_\Gamma$. Similar to the proof of Lemma \ref{obs:change-in-x}, the points on $P_{j_i}$ are monotone in the $x$-coordinate. This means the path from $r_{j_i}$ on $P_{j_i}$, is decreasing in the $x$-coordinate (because $r_{j_i}$ has its legs facing left), implying $P_{j_i}$ is completely to the left of $r_{j_i}$. Since $\Gamma$ is to the right of $r_{j_i}$, this means that $P_{j_i}$ can't intersect with $\Gamma$, contradicting the assumption that $P_{j_i}\in P_\Gamma$. This proves the claim.
	\end{proof}
	
    Now back to the statement of the lemma;
    Without loss of generality, assume that reflections in $\mathcal R$ are all ascending. 
    For the sake of contradiction, assume that there is an index $1\le a < m$ for which $P_{j_a}$ and $P_{j_{a+1}}$ aren't consecutive. This means $j_a < j_{a + 1} - 1$, and so we conclude that the subpath $P' = \bigcup_{j' = j_a + 1}^{j_{a + 1} - 1}P_{j'}$ of $\mathcal R$ from $r_{j_a}$ to $r_{j_{a+1}} - 1$
    is on one side of $\Gamma$ (or else there will be another $P_{j'}\in P_\Gamma$ with $j_a < j' < j_{a+1}$). So there is at least one reflection point from $\mathcal R$ that is in $P'$. Let $r_{i}$ be the first reflection on $P'$ after $r_{j_a}$.
    % Note that $P_i \in P_{j_a \to j_{a+1}}$, and $P_i\not \in P_\Gamma$ because it must not intersect with $\Gamma$.
    Without loss of generality, assume that $r_{j_a}$ (and therefore the entirety of $P'$) is on the right side of $\Gamma$. So $r_{j_a}$ is a left reflection because of Claim \ref{claim:included-reflections}. 
    By Lemma \ref{obs:consecutive-reflections}, both $r_{i}$ and $r_{j_a - 1}$ (the reflection in $\mathcal R$ before $r_{j_a}$) are right reflections. 
    
    Let $r_{q}$ be the end-point of $P_{j_{a+1}}$ that is to the left of $\Gamma$ (either $r_q = r_{j_{a+1}}$ or $r_q = r_{j_{a+1} - 1}$). Once again, using Claim \ref{claim:included-reflections}, we get that $r_{q}$ is a right reflection.
    So we have three right reflections $r_{j_a - 1}, r_{i},$ and $r_{q}$ such that $x(r_i) \ge x(\Gamma) \ge \{x(r_{j_a - 1}), x(r_q)\}$ and the order they're visited in $\mathcal R$ is $r_{j_a - 1}$, then $ r_i$, and then $r_q$. According to Lemma \ref{obs:x-order}, based on whether $\mathcal R$ is a sink or a zig-zag, we either must have $x(r_{j_a - 1}) < x(r_i) < x(r_q)$ or the reversed inequality; which neither are the case here. This contradiction implies the statement of the lemma.
\end{proof}

\subsubsection{Proof of Lemma \ref{lem:strip-bounded-shadow}: Bounding the Shadow of each Sink/Zig-zag}\label{sec:lemma-shadow-of-sink/zigzag}

This subsection is dedicated to the proof of Lemma \ref{lem:strip-bounded-shadow}. 

Let $\sigma=\lceil 1/\eps\rceil + 1$  and
consider any loop or ladder $P\in \OPT_\tau$ and let $\mathcal R$ be an arbitrary zig-zag/sink along $P$ with shadow larger than $\sigma$ at some vertical line $x=x_0$. Without loss of generality, assume that following the orientation of $\OPT$ along $P$, reflection points on $\mathcal R$ are ascending. Suppose that the subpath of $P$ following reflection points $r_j,r_{j+1},\ldots,r_k$ of $\mathcal R$ is crossing $x=x_0$ (note that, using Lemma \ref{lemma:consecutive-reflection-shadow}, the reflection points must be consecutive).
Let this subpath of $P$ starting at $r_j$ and ending at $r_k$ be $\mathcal R'$
and let $s_{a_j}, s_{a_{j+1}},\dots, s_{a_k}$ denote the segments that contain reflections $r_{j}, r_{j + 1},\dots, r_{j + k}$, respectively.	Also let $P_i$ (for $1\leq i\leq k-j$) be the subpath of $\mathcal R'$ from $r_{j+i-1}$ to $r_{j+i}$. 
Note that there might be several straight points or break points between $r_{j'},r_{j'+1}$ on $P$ (for each $j'$); 
the segments of these points are all covered by the shadow one path (due to Lemma \ref{obs:shadow-increase}) from $r_{j'}$ to $r_{j'+1}$.
According to Lemma \ref{obs:reflection-paths-above-below}, since all reflections are ascending, if $m_1 < m_2$, then $P_{m_1}$ is below $P_{m_2}$ (in the range that $P_{m_1}$ is defined on the $x$-axis).
So this specifically implies $P_1$ is below any other $P_m$ (in the range that $P_1$ is defined), and similarly, $P_k$ is above any other $P_m$ (in the range that $P_k$ is defined).
Note that $\mathcal R'$ is part of a zig-zag/sink itself (the only difference with the definition of zig-zag/sink is that $\mathcal R'$ is no longer necessarily maximal in $ \OPT_\tau $). 
Also, note that for each path $P_i$, the $x$ value of the points it visits
between the two reflection points $r_{j+i-1}$ to $r_{j+i}$ are monotone
increasing or decreasing (see Lemma \ref{obs:shadow-increase}).
Let $\Psi$ denote the cost of legs of $\mathcal R'$.
It follows that $r_j,r_{j+2},r_{j+4},\ldots$ are on one side of $x_0$
(say to the right) and $r_{j+1},r_{j+3},\ldots$ are on the other side (say left of $x_0$). Since the number of reflections to the right of $x = x_0$ differs from the number of reflections to the left of $x = x_0$ by at most $1$, then on each side of $x = x_0$ we have at least $(\sigma-1)/2 = \lceil1/2\eps\rceil$ reflections. Let $\sigma'=\lceil 1/2\eps\rceil$.
The idea of the proof is to show that aside from the  $2\sigma'$ reflections at the end of $\mathcal R'$ (i.e. the last $2\sigma'$ paths $P_{j}$), we can replace the paths between the rest of the reflection points so that it reduces the shadow of the entire $\mathcal R'$ to $O(1/\eps)$ while increasing the cost of the path by at most $O(\eps\cdot\Psi)$. 

Note that using Lemma \ref{obs:shadow-increase}, each subpath $P_j$ of $\mathcal R'$ is between two consecutive reflection points and so has a shadow of $1$. This implies that $P_{k-2\sigma'} \cup \ldots\cup P_{k-1}$ has a shadow of $O(\frac1\eps)$  as it has $2\sigma'$ consecutive reflections. We replace the rest of $\mathcal R$ (as we describe below) with a new path of a shadow of $O(1)$; this will yield the result of the lemma.

\subsubsection*{When $\mathcal R'$ is a part of a zig-zag}
    \begin{figure}[h]
        \centering
        \includegraphics[width=0.7\textwidth]{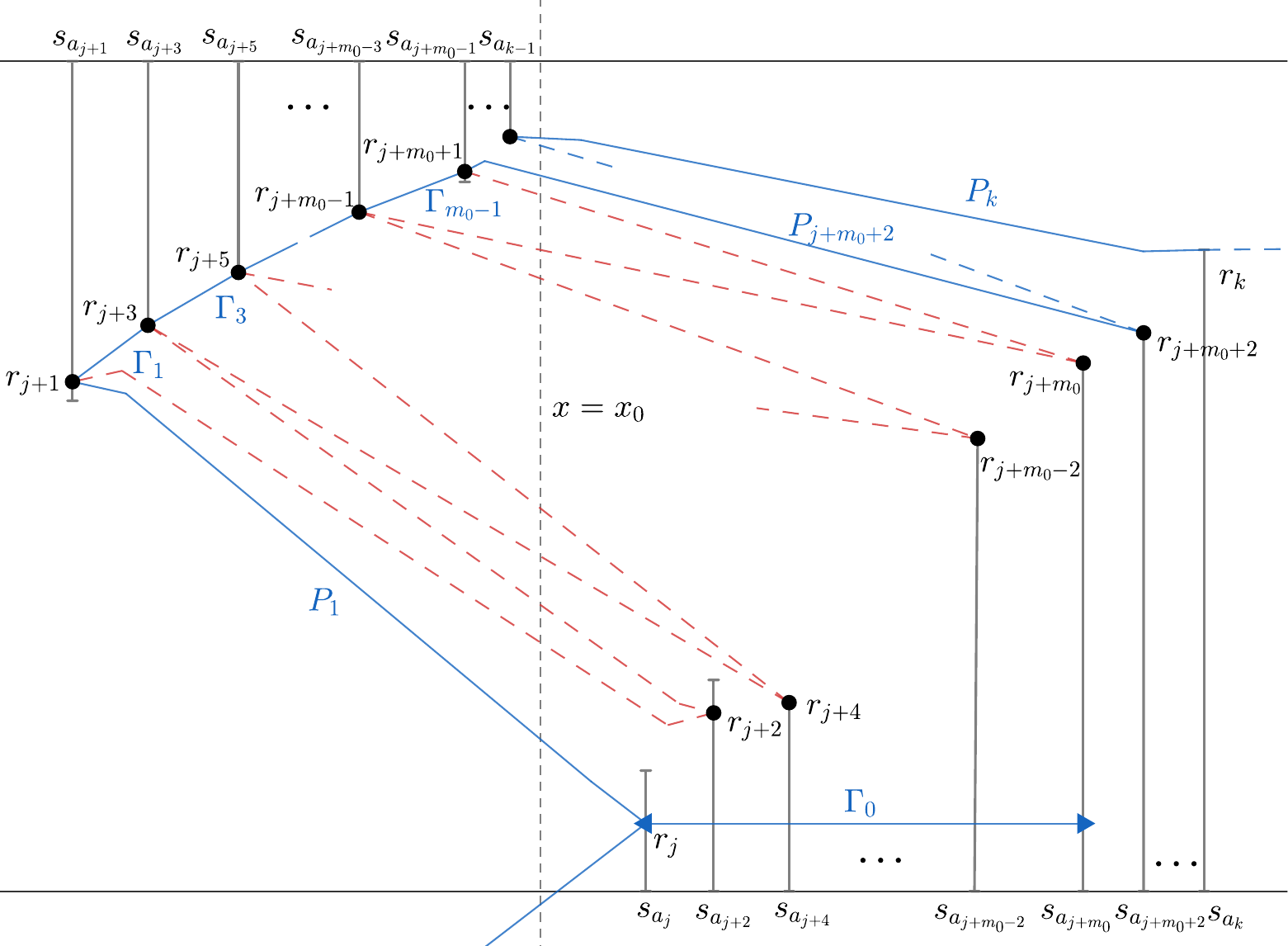}
        \caption[The charging scheme in Lemma \ref{lem:strip-bounded-shadow} for zig-zags.]{Alternative path for a zig-zag; The red dashed parts are discarded. There are further details about $\Gamma_m$'s that are explained throughout the proof of Lemma \ref{lem:strip-bounded-shadow}}
        \label{fig:zig-zag-charging}
    \end{figure}

    Without loss of generality, assume that $s_{a_j}, s_{a_{j + 2}}, \dots$ are bottom segments and to the right of $x=x_0$, and consequently, $s_{a_{j + 1}}, s_{a_{j + 3}},\dots$ are  top segments and to the left of $x=x_0$.
    Let $d_{m}$ for $m = 0, 1,\dots, k-j$ denote $ \abs{x(r_{j + m}) - x_0} $. Using Lemma \ref{obs:x-order}, we have $ d_0 < d_2 < \cdots $ and $d_1 > d_3 > \cdots$.
    Let's focus on the right side of $x = x_0$ (where the bottom segments are), so
    $r_{j}, r_{j + 2},\dots, r_{j + 2q}$ are all the reflections of $\mathcal R'$ on this side where $ q\ge \sigma'-1$. 
    We claim that except for at most the $\sigma'$ largest values in $d_2,d_4,d_6,\ldots$, all other values of $d_{2m}$'s are at most $\eps\cdot\Psi$ and this is done by an averaging argument. 
    More specifically, we show that
    the largest integer $m_0 \in \{0,2,4,\dots, 2q\}$ for which we have $d_{j+m_0} \le \eps\cdot\Psi$, has value $m_0 \ge 2(q - (\sigma'-1))$. To see why this is the case, assume otherwise, that for all even integers $m \ge 2(q - (\sigma'-1))$, we have $d_{j+m} > \eps\cdot\Psi$. Adding these inequalities for $m = 2(q-(\sigma'-1)),\; 2(q - (\sigma'-2)),\; \dots,\; 2q$ gives us
    
\begin{eqnarray*}
        d_{j+2(q - (\sigma'-1))}  + d_{j+2(q - (\sigma'-2))} + \dots + d_{j+2(q)}  &>& \sigma'\cdot (\eps\cdot\Psi)\\
        &=& \lceil1/2\eps\rceil\cdot \eps\cdot\Psi \\
        &\ge& \Psi/2,
\end{eqnarray*}

    which clearly isn't possible, due to $2\sum_{m\in\{2(q-(\sigma'-1)),\ldots,2q\}} d_{j+m}\leq\Psi$; this inequality holds because paths $P_{j+m},P_{j+m+1}$ ($m\in\{2(q-(\sigma'-1)),\ldots,2q\}$) have to travel the $x$-distance from $x_0$ to $s_{a_{j+m}}$ to the reflection points $r_{j+m}$, and all these paths are part of $\mathcal R'$. This contradiction shows our initial claim, that for some $m_0 \ge 2(q - (\sigma'-1))$, we have all of $d_{j},d_{j+2},\ldots,d_{j+m_0} \le \eps\cdot\Psi$.
        
        We are going to change $\mathcal R'$ from $r_{j}$ up to $r_{j + m_0}$, but keep $P_{j + m_0+1}$ and after; this change will result in another feasible solution with an $O(1)$ shadow up to $r_{j+m_0}$, and cost increase will be at most $O(\eps\cdot\Psi)$. 
        Our modification of $\mathcal R'$ is informally as follows (skipping some details to be explained soon). Starting at $r_j$ instead of following $P_1$ to $r_{j+1}$, we first travel horizontally to the right until we hit $s_{a_{j+m_0}}$ (the bottom segment which $r_{j+m_0}$ is located on), and travel back to $r_j$. Let's call this horizontal back and forth subpath $\Gamma$. 
	    This subpath $ \Gamma $ will ensure that all the bottom segments that $\mathcal R'$ covers between $x=x_0$ and $s_{a_{j+m_0}}$ are covered (we may need to deviate from $\Gamma$ further down if $\mathcal R'$ goes further below $\Gamma$ at some point; will formalize this soon). The shadow of $\Gamma$ will easily be shown to be 2.
	    Then from $r_j$, we follow $P_1$ and go to $r_{j+1}$ which is the left-most reflection on a top segment (to the left of $x=x_0$).
	    Now instead of following $P_2$ to go to $r_{j+2}$ and then $P_3$ to go to $r_{j+3}$, we go straight from $r_{j+1}$ to $r_{j+3}$ (with some little details skipped here), then to $r_{j+5}$ and so on until
	    we get to $r_{j+m_0+1}$, and from there we follow $\mathcal R'$. One observation is that the shadow of the new path from $r_{j+1}$ to $r_{j+m_0+1}$ is also 1 since it won't have any reflection points. 
	    The rest of the path from $r_{j+m_0+1}$ to $r_k$ that follows $\mathcal R'$ has at most $O(\sigma')$ reflection points and hence the shadow is $O(1/\eps)$.
	    We show that the new path hits all the segments $\mathcal R'$ were hitting; and so we still have a feasible solution where the overall increase in the cost is at most $O(d_{j+m_0})$, which is bounded by $O(\eps\cdot\Psi)$. Hence, we
	    find a modification of the path $\mathcal R'$ with shadow bounded by $O(1/\eps)$, and cost increase is at most $O(\eps\cdot\Psi)$.
		Note that any bottom segment (if any) to the left of $x=x_0$ that was covered by $\mathcal R'$, must intersect $P_1$;
        as $P_1$ is below the rest of $\mathcal R'$ to the left of $x=x_0$. 
        Thus, any bottom segment to the left of $x_0$ that is covered by any of the $P_{b>1}$, is also covered by $P_1$. There are some details missing in this informal description that are explained below.
	    
        We will introduce a new subpath $\Gamma_0$, responsible for covering all bottom segments in $\mathcal R'$ to the right of $x = x_0$ until $s_{j+m_0}$; and we introduce a collection of subpaths $\Gamma_m$ for odd $m$ in $\{1,\dots, m_0\}$ for covering the top segments to the left of $x = x_0$. All of $\Gamma_m$'s, will have a shadow of $1$. We ensure that any bottom segment hit by $\mathcal R'$ between $s_{j_0}$ and $s_{j+m_0}$, is also hit by $\Gamma_0$ between $x=x(s_{j})$ and $x=x(s_{j+m_0})$; and also
        any top segments that $\mathcal R'$ was hitting in the range that each 
        $\Gamma_m$ is defined, is hit by $\Gamma_m$, for $m\geq 1$.

        Consider the horizontal line $y = y(r_j)$ from $s_{a_j}$ to $s_{a_{j+m_0}}$. Refer to this horizontal portion as $\Gamma$. For reflection points $r_j,r_{j+2},\ldots,r_{j+m_0}$ (on bottom segments $s_{a_{j}},s_{a_{j+2}},\ldots,s_{a_{j+m_0}}$),
        the two paths that contain a leg incident to $r_{j+m}$ are $P_{j+m}$ and $P_{j+m + 1}$ for each $0\leq m\leq m_0$. Recall that using Lemma \ref{obs:reflection-paths-above-below}, $P_{j+m}$ is below $P_{j+m+1}$ between $s_{a_j+m-1}$ and $s_{a_j+m}$. 
		Consider the area $A_{\Gamma}$ of the strip bounded by $\Gamma\cup s_{a_j} \cup s_{a_{j+m_0}}$. Then $\mathcal R'\cap A_{\Gamma}$ are (possibly empty) subpaths that start and end at $\Gamma$. These subpaths form the lower-envelope of $\mathcal R'\cup\Gamma$ in $A_{\Gamma}$ (for e.g. in Figure \ref{fig:gamma-0}, paths $P_4,P_5$ that reach $r_{j+4}$, cross $\Gamma$ at points $q^4_1,q^4_2$.).
          \begin{figure}[H]
            \centering
            \includegraphics[width=.55\textwidth]{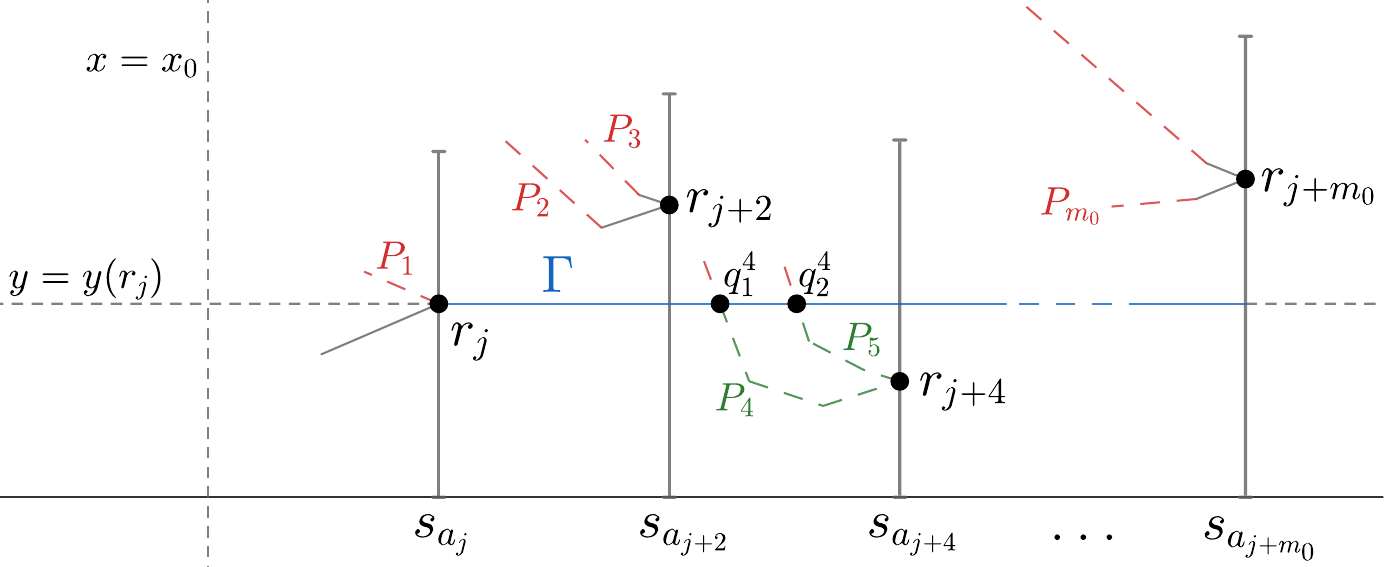}
            \caption[An illustration of subpath $ \Gamma_0 $.]{$\Gamma_0$ is the lower envelope of the blue line (the segment $\Gamma$) together with the green parts (portions of $\OPT_\tau$ that go below $\Gamma$)}
            \label{fig:gamma-0}
        \end{figure}
        
        We define $\Gamma_0$ to be a path starting at $r_j$ that travels right along $\Gamma$ and the lower envelope of $\mathcal R'$ in this area, i.e. whenever travelling right on $\Gamma$, if we arrive at an intersection of $\mathcal R'$ with $\Gamma$ (say a path $P_{j+m}$) then we travel along $P_{j+m}$ inside $A_{\Gamma}$ until we hit back at $\Gamma$, and then continue travelling right. For instance, in
        Figure \ref{fig:gamma-0}, when travelling on $\Gamma$ from $r_j$ to right, once we arrive at $q^4_1$, we follow $P_4$ to $r_{j+4}$, then follow $P_5$ to $q^4_2$, and then continue right on $\Gamma$. Once we arrive at $s_{a_{j+m_0}}$, we travel $\Gamma$ horizontally back to $r_{j}$. The length of $\Gamma_0$ can be bounded by the length of $\mathcal R'\cap A_{\Gamma}$ plus $2d_{j+m_0}\leq 2\eps\Psi$.
        Also, it can be seen that any bottom segment that was covered by $\mathcal R'$ in between $x(r_j)$ and $x(r_{j+m_0})$, is covered by $\Gamma_0$ (since we travel the lower envelope of $\mathcal R'\cup A_\Gamma$ in the range we're defining $ \Gamma_0 $). Any top segment that is covered by $\mathcal R'$ within $[x(r_j),x(r_{j+m_0})]$, must be also covered by $P_{m_0+1}$; as that path is above all other $P_m$'s in the range of $[x_0, x(r_{j+m_0})]$.   After travelling $\Gamma_0$, we travel along $P_1$ to $r_{j+1}$.
        Now we're going to define $\Gamma_m$ for odd $ 1\le m \le m_0 $. Each $\Gamma_m$ goes from $r_{j + m}$ to $r_{j + m + 2}$ 
        % is going from $r_{j+1}$ to $r_{j+3}$, $\Gamma_2$ is going from $r_{j+3}$ to $r_{j+5}$ and so on
        until we arrive at $r_{j+m_0+1}$; after which we follow along $\mathcal R'$ (i.e. $P_{j+m_0+2}$, then $P_{j+m_0+3}$ and so on). Path $\Gamma_1$ will replace $P_2+P_3$, $\Gamma_3$ will replace
        $P_4+P_5$, and so on. Note that $\Gamma_m$'s are all to the left of $x = x_0$.
	   For any two reflections $r_{j + m}$ and $r_{j + m + 2}$ that lie on top segments $s_{a_{j+m}}$ and $s_{a_{j+m + 2}}$, let $\gamma_m$ be the subpaths of $\mathcal R'$ restricted to the area of the strip cut by segment $r_{j + m}r_{j + m + 2}$ and $s_{a_{j+m}}$ and $s_{a_{j+m + 2}}$ (i.e. the area between $s_{a_{j+m}}$ and $s_{a_{j+m + 2}}$ and above $r_{j+m}r_{j + m + 2}$). Path $\Gamma_m$ is obtained by starting at $r_{j+m}$ and following line $r_{j + m}r_{j + m + 2}$ and whenever we hit $\mathcal R'$, i.e. a subpath of $\gamma_m$ (see Figure \ref{fig:gamma-m}) we follow that subpath until we arrive back to $r_{j+m}r_{j+m+2}$ again; we continue until we reach $r_{j+m+2}$. In other words, we follow the upper envelope
    of $r_{j+m}r_{j+m+2}\cup \mathcal R'$ between $r_{j+m}$ and $r_{j+m+2}$.
    \begin{figure}[h]
    	\centering
    	\includegraphics[width=.35\textwidth]{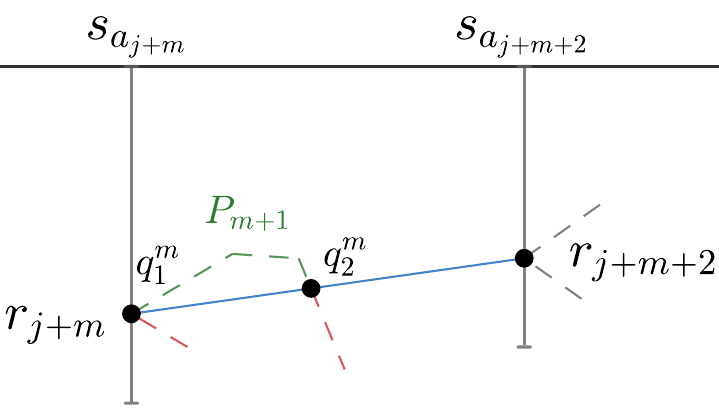}
    	\caption[An illustration of subpaths $ \Gamma_m $.]{$\Gamma_m$ is travelling the upper envelope of the blue line (the segment $r_{j + m}\left.r_{j + m + 2}\right.$) and the green parts (portions of $\OPT_\tau$ that go above $r_{j + m}\left.r_{j + m + 2}\right.$)}
        \label{fig:gamma-m}
    \end{figure}
	   If we define $q^m_1, \dots, q^m_{2i_m}$ to be the intersections of $\mathcal R'$ with $\Gamma_m$, ordered in the direction of $r_{j + m} \to r_{j + m + 2}$, then if $\mathcal R'$ intersects with $r_{j+m}r_{j+m+2}$ and goes above it, it has to be at a point $q^m_{u}$ where $u$ is odd, and otherwise $u$ has to be even. 
	   
	Overall, we have changed the subpaths of $\mathcal R'$ from
	$r_j$ to $r_{j + m_0}$ as follows (see Figure \ref{fig:zig-zag-charging}):
		
    \label{alg:zig-zag-path}
    \begin{itemize}
        \item Follow $\Gamma$ from $r_{j}$ towards $s_{a_{j+m_0}}$, such that every time an intersection point with $\mathcal R'$ (say point $q^0_{2u - 1}$) is reached, then follow along $\mathcal R'$ until the next intersection of $\mathcal R'$ with $\Gamma$ (say point $q^0_{2u}$) is reached; then continue along $\Gamma$. Repeat this process until we reach $s_{a_{j+m_0}}$, then follow along $\Gamma$ from right to left directly back to $r_{j}$; this is subpath $\Gamma_0$

	    \item From $r_{j}$, follow $P_1$ to reach $r_{j + 1}$.
	    
	    \item From $r_{j + m}$ (initially $m = 1$), follow $\Gamma_m$ similar to the first step; meaning follow the segment $r_{j + m}\left. r_{j + m + 2}\right.$, and when an intersection point $q^m_{2u - 1}$ with $\mathcal R'$ is reached, follow $\mathcal R'$ instead, until you reach the next intersection point $q^m_{2u}$ on $r_{j+m}r_{j+m+2}$. Repeat this process (for $m = 1, 3, \dots$) until $r_{j + m_0 + 1}$ is reached. 
	    \item From $r_{j + m_0 + 1}$, follow $P_{j+m_0+2}$ and the rest of $\mathcal R'$ to the end.
	\end{itemize}
        First, we show that we still have a feasible solution, i.e.
        every segment that $\mathcal R'$ used to cover, will have an intersection with the new solution. To see why this is the case, first note that $\Gamma_0$ by definition is always on or below $\mathcal R'$ in the ranges it's defined; this means that (using Observation \ref{cor:path-above-top-segment}) $\Gamma_0$ covers any bottom segment that $\mathcal R'$ covers between $s_{a_{j}}$ to $s_{a_{j+m_0}}$. Also, all the top segments in this range  will be intersecting the path $P_{m_0+1}$, since $P_{m_0+1}$ is above all $P_{\leq m_0}$ in this range.
        Similarly, each $\Gamma_m$ is on or above $\mathcal R'$ in the range $[x(r_{j+m}),x(r_{j+m+2})]$, meaning they cover all the top segments that $\mathcal R'$ used to cover between $s_{a_{j+1}}$ to $s_{a_{j+m_0+1}}$. Also, any bottom segment that was covered in this range is covered by $P_1$. 

		Next, note that from $r_{j}$ to $r_{j+m_0+1}$, the shadow is at most 3. That is because shadow of $\Gamma_0$ is 2, shadow of each $\Gamma_m$ ($1\leq m$) is 1, and shadow of $P_1$ is 1.

        Now we show the new solution has an additional cost of at most $3\eps\Psi$. All parts of $\Gamma_0$ and the rest of $\Gamma_m$'s that used portions of $\mathcal R'$ can be charged onto $\mathcal R'$ itself. So we only have to properly charge the line segment $\Gamma$ along with its duplicate (part of $\Gamma_0$) and line segments $r_{j + m}r_{j + m + 2}$ (part of $\Gamma_m$). We know that $\norm \Gamma = x(r_{j + m_0}) - x(r_{j}) < x(r_{j + m_0}) - x_0 = d_{m_0} \le \eps\cdot\Psi$. So we pay at most $2\eps\Psi$ extra (compared to $\OPT_\tau$) for travelling $\Gamma_0$. We consider one additional copy of $\Gamma$ for the extra cost we pay elsewhere in $\Gamma_m$ ($m\geq 1$), and we are going to use this for our charging scheme. So at the end, the total extra cost is going to be bounded by $3\eps\Psi$.
        
        For each two reflections $r_{j + m}$ and $r_{j + m + 2}$ that lie on top segments, note that $\mathcal R'$ had
        two subpaths paths $P_{j+m + 1},P_{j+m + 2}$ whose concatenation makes a path from $r_{j+m}$ to $r_{j+m+2}$; but now, it is possible that some portions of $P_{j+m+1},P_{j+m + 2}$ are used in $\Gamma_0$ during our alternate solution (those that belonged to $A_\Gamma$). But having that additional copy of $\Gamma$ that we accounted for, we can use it to short-cut the missing parts of $P_{j+m + 1}\cup P_{j+m + 2}$ to again make a path from $r_{j + m}$ to $r_{j + m + 2}$. Overall, the total length of $P_1+\sum_{m\geq 0} \Gamma_m$ that is replacing $P_1+P_2+\ldots+P_{j+m_0+1}$ is at most $3\eps\Psi$ larger than length of $P_1+P_2+\ldots+P_{j+m_0+1}$.
        Thus, we conclude the lemma for the case of zig-zags.

\subsubsection*{When $\mathcal R'$ is a part of a sink}
    The proof is analogous to the case of zig-zags. Without loss of generality, assume all reflections in $\mathcal R'$ are on bottom segments. Define $d_{m} = \abs{x(r_{j + m}) - x_0}$ like before. If $r_{j}, r_{j+2},\dots, r_{j + 2q}$ are all the reflections to the right of $ x = x_0 $, then with the same arguments as the case of zig-zags, we will find an integer $m_0 \ge 2(q - (\sigma' -1))$ $(\sigma' = \lceil \frac{1}{2\eps}\rceil)$ for which we have $d_{m_0} \le \eps\cdot\Psi$. 
    
    We will replace the subpath of $\mathcal R'$ from $r_j$ to $r_{j + m_0}$ in the same fashion as before. Let $\Gamma$ be the segment on the line $y = y(r_j)$ in the range $[x(r_j),\left. x(r_{j + m_0}) \right.]$. Define $\Gamma_0$ to be the union of $\Gamma$ with the portions of $\mathcal R'$ that go below it. Define each $\Gamma_m$ for a reflection $r_{j + m}$ to the left of $x = x_0$ to be the union of segment $r_{j+m}\left. r_{j + m + 2}\right.$ with the portions of $\mathcal R'$ that go below it.

    The same arguments as before hold, that each of $\Gamma$ or $r_{j+m}\left. r_{j + m + 2}\right.$ that we defined above, will have an even number of intersections with $\mathcal R'$. Define the new path between these reflections in
    %\ref{alg:zig-zag-path}.  
    the same way as we did for zig-zags.
	
    The cost arguments still hold, implying that the new path has an additional cost of $O(\eps\cdot\Psi)$. Also, the new path is always on or below $\mathcal R'$, so it covers all the bottom segments that $\mathcal R'$ used to cover previously. But one can see that $P_{m_0}$ is above all of $\mathcal R'$ in the path between $r_{j}$ to $r_{j + m_0}$. Thus, $P_{m_0}$ alone will cover all top segments that $\mathcal R'$ used to cover. Once again, the new solution has a shadow of at most $3$ in the subpath between $r_j$ and $r_{j+m_0}$ and shadow $O(1/\eps)$ afterwards.
    This concludes the proof for the case of sinks and the proof of 
    Lemma \ref{lem:strip-bounded-shadow}.
    \qed

%%%%%%%%%%%%%%%%%%%%%%%%%%%%%%%%%%%%%%%%%%%%%%%%
\subsubsection{Proof of Lemma \ref{lemma:reflections}: Bounding the Size of Pure Reflection Sequences}\label{sec:lemma-pure-reflection-size}
In this subsection, we will prove Lemma \ref{lemma:reflections}. 

We prove the lemma by showing how to change each ladder or loop (i.e. any path of $\OPT_\tau$ that starts and ends on one of the cover-lines) so that the size of each pure reflection sub-sequence is bounded without increasing the cost by more than $(1+\eps)$ factor. Consider any loop or ladder $P\in \OPT_\tau$ and any maximal pure reflection sequence  $P'=r_0, r_2, \dots, r_k $ in $P$ where $k> \frac1\epsilon $. Let $\Psi$ be the length of subpath of $\OPT_\tau$ from $r_0$ to $r_k$.
We show how we can modify this subpath to another one whose length is at most $(1+O(\eps))\Psi$ such that the length  of each
pure reflection sub-sequence is bounded by $O(1/\eps)$.
Note that using Lemma \ref{lem:cons-ref},  all $r_i $'s are ascending or all are descending. 
This also implies that the $ y $-coordinates of $r_i$'s are monotone. i.e. either $ y(r_0)\le y(r_1)\le \dots \le y(r_k) $ or the other way around. Without loss of generality, assume it is the former case and so all are ascending reflection points.
Proof of Lemma \ref{lemma:zig-zag}, shows if we have a maximal monotone (i.e. all ascending or all descending) sequence of reflection points, then it consists of at most a sink followed by a zig-zag, followed by a sink. Therefore, it suffices to bound the size of pure reflection sequence in a single sink or a zig-zag alone as a function of $1/\eps$. So let's assume all $ r_i $'s form a single sink or all form a single zig-zag.

Recall from the definition of pure reflection sequence that there might be straight points in $P$ between two consecutive reflection points.
For any reflection point $r_i$ on a segment $s$, let $ d_i^+ $ and $ d_i^- $ be the distances of $r_i$ to the top and bottom tips of $s$, respectively. 
By the definition of a pure reflection sequence, $d^-_i > 0 $ and $ d_i^+ > 0 $ for all $0\leq i\leq k $ (because the reflections are not at the tips). 
With  $\sigma=\lceil\frac 1\epsilon\rceil$, we break $P'$ into $m$ subpaths $G_1,\ldots,G_m$ where $G_j$ is the subpath of $P'$ from $r_{{j-1}\sigma}$ to $r_{j\sigma}$,
except that the last group ends at $r_k$. Note that the concatenation of these paths is $P'$, and each subpath has at most $\sigma+1$ reflections.
Consider any group $G_j$ and let $\mathcal G_j$  be the cost of the legs of $P$ between the reflection points of $G_j$ and let $D_j$ be the smallest value among minimum of $d^+_a,d^-_a$ among all reflection points $r_a\in G_j$, i.e. $D_j=\min_{(j-1)\sigma\leq a\leq j\sigma}\{d^+_a,d^-_a\}$.

\begin{claim}
	For each $1\leq j\leq m$:	$D_j \leq \frac{2\eps}{1 - \eps}\cdot \mathcal G_j$.
\end{claim}
\begin{proof}[Proof of Claim]
For	simplicity of notation of indices, we prove this for $j=1$, i.e. $G_1=r_0,\ldots,r_\sigma$.
	As mentioned above, it suffices to show the claim for $G_1$ being part of a sink or a zig-zag.

    \vspace{.5cm}
    
	\noindent
	$\bullet$ \textbf{$G_1$ is part of a sink}:\vspace{.25cm}\\
	Without loss of generality, assume all reflections in $G_1$ are on bottom segments.
	Consider the three consecutive reflection points $r_0, r_1$, and $r_2$.
	Using Lemma \ref{obs:reflection}, we get that subpath $r_1\to r_2$ (which, according to the definition of pure reflection sequence, is a straight line) does not intersect with the segment containing $r_0$. Since we assumed the $y$-coordinates of $r_i$'s are increasing, this implies that $r_1r_2$ and consequently $r_2$ lie above the segment containing $r_0$. This along with triangle inequality yields us
            $d^+_0 \leq y(r_2)-y(r_0)\leq ||r_2r_0|| \le ||r_0r_1||+||r_1r_2||$. With 
 the same argument, we get $d_i^+ \le y(r_{i + 2}) - y(r_i) \le ||r_ir_{i + 1}|| + ||r_{i + 1}r_{i + 2}||$ for all $0\le i \le \sigma - 2$ (see Figure \ref{fig:sink-d_1+}).
  Considering these inequalities for different $r_i$'s ($1\leq i\leq\sigma -2 $) and summing them up
  for all $r_i$'s in group $G_1$, using the fact that $D_1 \le d_i^+$, we obtain
  \begin{align*}
      (\sigma - 1)\cdot D_1 &\le \sum_{i = 1}^{\sigma - 2}d_i^+
      \le \sum_{i = 1}^{\sigma - 2}(\norm{r_ir_{i + 1}} + \norm{r_{i + 1}r_{i + 2}}) \le 2\cdot \mathcal G_1\\
      \implies D_1 &\le \frac2{\sigma - 1}\cdot \mathcal G_1.
  \end{align*}
  
This implies in a sink, $D_1 \le \frac2{1/\eps - 1}\cdot \mathcal G_1 = \frac{2\eps}{1 - \eps}\cdot \mathcal G_1$ and in general, $D_j \le \frac{2\eps}{1-\eps}\cdot \mathcal G_j$ for all $1\le j \le m$.
\begin{figure}[H]
	\centering
    \includegraphics[width=.3\textwidth]{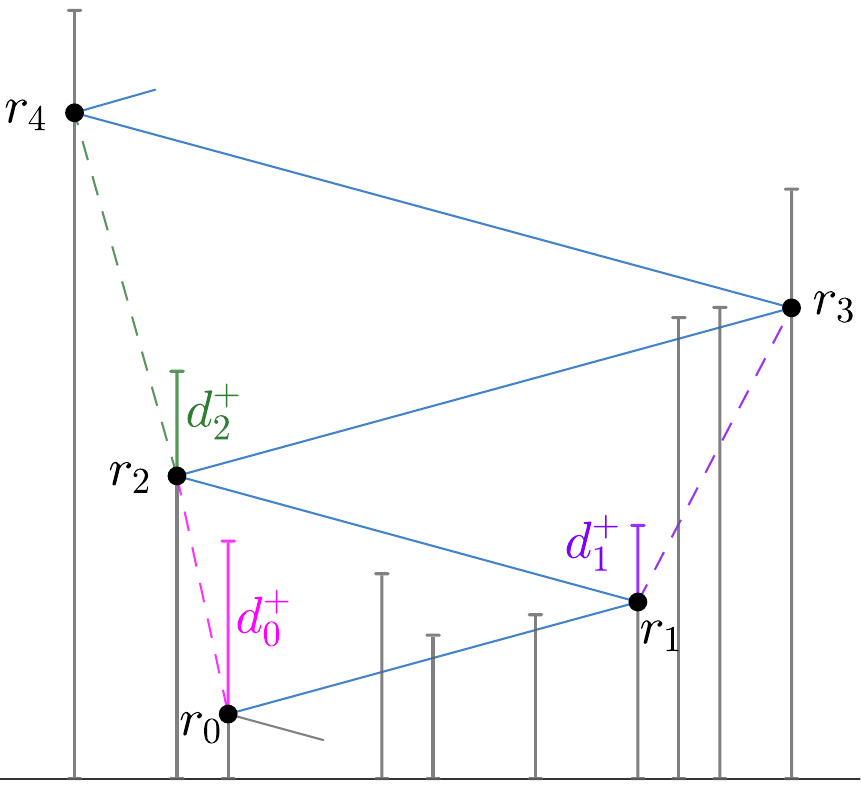}
    \caption{Each $d_i^+\ (i\le \sigma - 2)$ can be charged into the lines reaching the next two reflections}
    \label{fig:sink-d_1+}
\end{figure}

%    \item 
\noindent $\bullet$ \textbf{$G_1 $ is part of a zig-zag}:\vspace{.25cm}\\
        The inequalities are almost analogous, but there are two of them. Without loss of generality, assume that $r_1, r_3,\dots$ are on bottom segments, and therefore $r_0,r_2, r_4,\dots$ are on top segments. We give inequalities for $d_i^+$ on bottom segments, and for $d_{i}^-$ on top segments.
        For $i = 1, 3,\dots$ with $i\le \sigma - 2$, similar to the case of $G_1$ being a sink, we have
        % \begin{equation}\label{eq:zigzag-pure-reflection1}
            $d_i^{+} \le y(r_{i + 2}) - y(r_i) \le ||r_ir_{i + 1}|| + ||r_{i + 1}r_{i + 2}||$.
        % \end{equation}
        For $i = 2,4,6,\dots$, we have 
        % \begin{equation}\label{eq:zigzag-pure-reflection2}
            $d_i^- \le y(r_i) - y(r_{i - 2}) \le \norm{r_ir_{i - 1}} + \norm{r_{i - 1}r_{i - 2}}$ (see Figure \ref{fig:zigzag-d_1+}).
        % \end{equation}

    Now if we add these inequalities (with proper selection between $d_i^+$ and $d_{i'}^-$) we get
    \begin{align*}
        (\sigma - 1)\cdot D_1 &\le (d_1^+ + d_3^+ +\cdots) + (d_2^- + d_4^- +\cdots)\\
        &\le \hspace{-.25cm}\sum_{i{\rm \; is\; odd,\;} i\geq 1} (\norm {r_ir_{i + 1}} +  \norm{r_{i + 1}r_{i + 2}}) + \hspace{-.35cm}\sum_{i {\rm\; is\; even,\; } i \ge 2}(\norm{r_ir_{i - 1}} + \norm{r_{i - 1}r_{i - 2}})\\
        &\le 2\cdot \mathcal \mathcal G_1\\
        \implies D_1 &\le \frac2{\sigma - 1}\cdot \mathcal G_1
    \end{align*}
    
    And like before, this implies that in a zig-zag, $D_1 \le \frac{2\eps}{1 - \eps}\cdot \mathcal G_1$, and in general, $D_j \le \frac{2\eps}{1 - \eps} \cdot \mathcal G_j$.

\begin{figure}[h]
    \centering        
	\includegraphics[width=.3\textwidth]{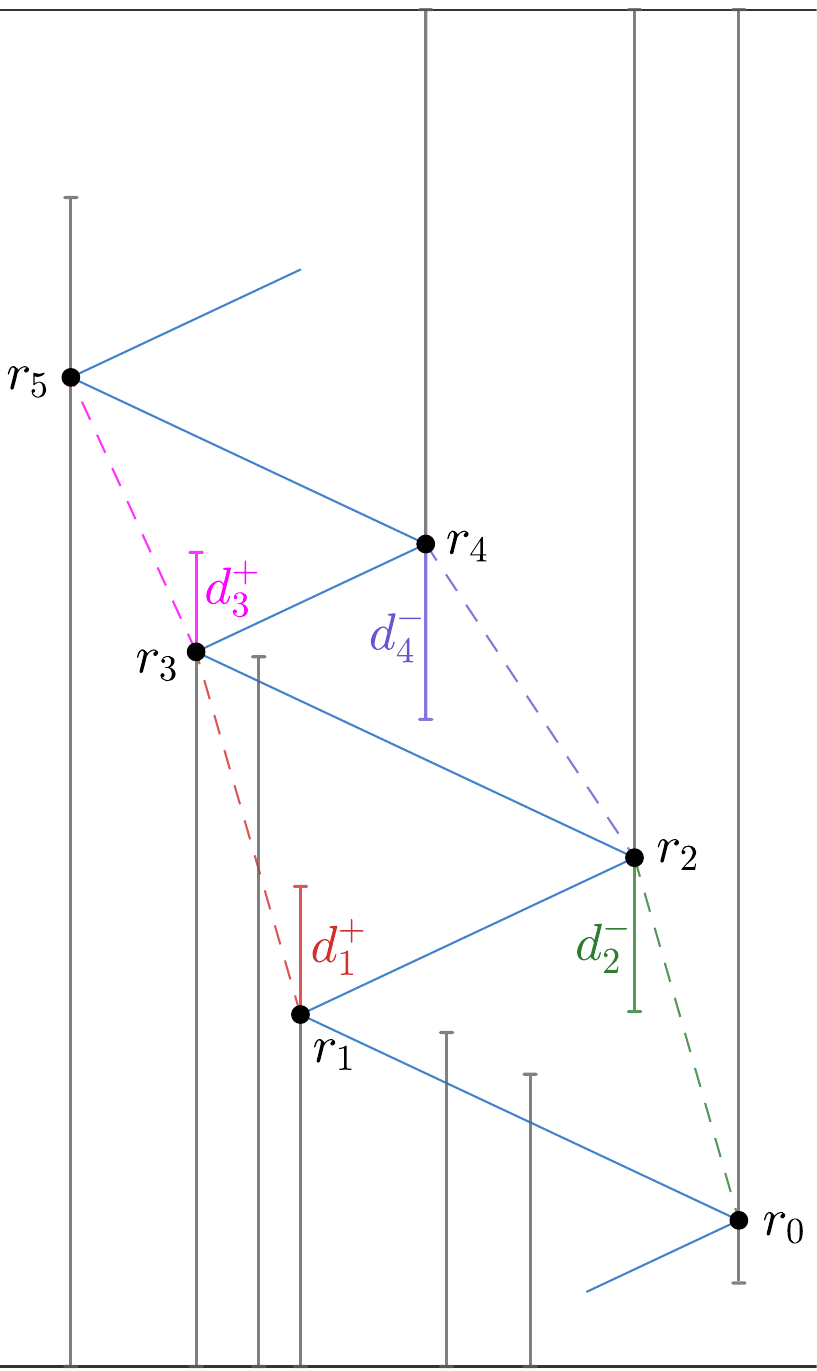}
    \caption{$d_i^+\ (i\le \sigma - 2)$ on bottom segments and $d_{i'}^-\ (i'\ge 2)$ on top segments can be charged into the lines reaching the next/previous two reflections}
    \label{fig:zigzag-d_1+}
\end{figure}

So we see that the claim holds for loops and ladders.
\end{proof}
Also note that if we have any three consecutive points $p,r_j,q$ on $\OPT_\tau$ where $r_j$ is a reflection on segment $s$ (with $s^t$ being its top tip), then $||p r_j||+||r_j q||+ 2d^+_j\geq ||p s^t||+||s^t q||$ using triangle inequality. Now consider any group $G_j$ ($1\leq j\leq m$) and assume $r_{j^*}$
is a reflection point in $G_j$ that lies on the segment $s$ for which $d^+_{j^*}=D_j$ (if $d^-_{j^*} = D_j$, then consider the bottom tip, $s^b$ instead). Consider the two legs of $\OPT_\tau$ incident to $r_{j^*}$,  namely
$\ell_{j^*-1}$ and $\ell_{j^*}$. Let $\ell_{j^* - 1} = p_{a}r_{j^*}$ and $\ell_{j^*} = r_{j^*}p_{b}$ where $p_a$ and $p_b$ are points on $\OPT_\tau$. Suppose we move $r_{j^*}$ from its current location to $s^t$, i.e. replace the two legs with $p_as^t$ and $s^t p_b$. Note that this will remain a feasible solution, as $\ell_{j^*-1},\ell_{j^*}$ have no other intersections with any other segment (as the definition of legs). The new cost is upper bounded by $||p_a r_{j^*}||+||r_{j^*} p_b||+2d^+_j$,
which means the increase is bounded by $2d^+_j=2D_i \le \frac{4\eps}{1 - \eps}\cdot\mathcal G_i$.

Note that in this new solution in each group $G_i$, one of the points is moved to be a tip of the segment it lies on. This implies the maximum size of a pure reflection sequence is now bounded by $2\sigma = 2\lfloor 1/\eps \rfloor$ and the 
total increase in the cost (over all groups) is bounded by $\sum_{j}\frac{4\eps}{1-\eps}\cdot \mathcal G_j = O(\eps\cdot \sum_j\mathcal G_j)=O(\eps\Psi)$.
%\end{pproof}
This completes the proof of Lemma \ref{lemma:reflections}. \qed

%%%%%%%%%%%%%%%%%%%%%%%%%%%%%%%%%%%%%%%%%%%%%%%%%%%%%%%%%%%
\subsubsection{Proof of Lemma \ref{lemma:overlapping-loops-ladders}: Bounding the Number of Overlapping Loops/Ladders}\label{sec:lemma-overlapping-loops/ladders}

Our goal in this subsection is prove Lemma \ref{lemma:overlapping-loops-ladders}, that shows that any given vertical line can intersect with at most $O(1)$ many loops or ladders of
$\OPT_\tau$ in a strip $S_\tau$. This together with Corollary \ref{cor:loop-bounded-shadow} implies there is a $(1+\eps)$-approximate solution where the shadow in each strip $S_\tau$ is $O(1/\eps)$.

We will show that there are at most $12$ overlapping loops, and at most $7$ overlapping ladders in $ \OPT_\tau $. 
    %Let $\Gamma$ be the vertical line that intersects with some subpaths of $\OPT_\tau$.
    %, i.e. the maximum number of overlapping loops/ladders happens at the vertical line $\Gamma$. 
    Suppose there is a vertical line $\Gamma$ and a number of loops and ladders are all crossing $\Gamma$.
    We bound the number of loops separately from the number of ladders.
    
\subsubsection*{Overlapping Loops}
    
    For each of the cover-lines of $S_\tau$, we will show that there are at most $6$ overlapping loops that have both their entry points on that cover-line. This will imply that there are at most $12$ overlapping loops in total. So from this point onward, let's focus on all overlapping loops on the bottom cover-line. This holds for all the claims and proofs that we introduce in this subsection, unless stated otherwise.
    
    Recall that $\OPT$ is not self-crossing, so it cannot have two overlapping cover-line loops.
    We say a loop $L_1$ with entry points $e_1,o_1$ is {\em nested} over loop $L_2$ with entry points $e_2,o_2$
    if both $e_2,o_2$ are between $e_1,o_1$. 

\begin{lemma}\label{claim:nested-loop-above}
    Let $L_1$ and $L_2$ be any two loops such that $L_1$ is nested over $L_2$. Let $p^2_r$ and $p^2_l$ be the right-most and left-most points on $L_2$, respectively. Then, in the range $I = [x(p^2_l), x(p^2_r)]$, $L_1$ is above $L_2$. For simplicity, in this case, we say $L_1$ is above $L_2$.
\end{lemma}
\begin{proof}[Proof of Lemma \ref{claim:nested-loop-above}]
    For $L_j, j=1,2$, let $e_j$ and $o_j$ be its entry points and without loss of generality, assume that $x(e_1) \le x(e_2) \le x(o_2) \le x(o_1)$. So $L_1$ is a path  from $e_1$ to $o_1$; meaning it crosses the vertical lines $x = x(e_2)$ and $x = x(o_2)$ at some point. This implies if $\mathcal L_1$ is the area of strip $S_\tau$ bounded by $L_1$ and the bottom cover-line, then $L_2$ is entirely inside $\mathcal L_1$. This means if the left-most and right-most points on $L_1$ are $p^1_l$ and $p^1_r$, then $x(p^1_l) \le x(p^2_l)$ and $x(p^1_r) \ge x(p^2_r)$. So we conclude that $L_1$ is defined in the range $I' = [x(p^1_l), x(p^1_r)]$ and that 
    $I \subseteq I'$. Therefore, in particular, $L_1$ is defined in the range $I$ and is above $L_2$.
\end{proof}

\begin{lemma}\label{claim:loop-above-loop}
    Suppose $L_1$ with entry points $e_1,o_1$ and $L_2$ with  entry points $e_2,o_2$
    are overlapping such that $x(e_1)<x(e_2)<x(o_1)$. Then $L_1$ must be nested over $L_2$ and
    $L_2$ is a cover-line loop. 
\end{lemma}
\begin{proof}[Proof of Lemma \ref{claim:loop-above-loop}]
    If $L_1,L_2$ are not nested (i.e. $x(e_1)<x(e_2)<x(o_1)<x(o_2)$) and none is a cover-line loop, then they are intersecting inside $S_\tau$, a contradiction.
    If they are not nested and one (say $L_2$) is a cover-line loop, then again they are intersecting at one of the entry  points. So they must be nested, 
    say $x(e_1)<x(e_2)<x(o_2)<x(o_1)$. Thus, using Lemma \ref{claim:nested-loop-above}, $L_1$ is above $L_2$; and if $L_2$ intersects with any top segment, $L_1$ would already be intersecting with it because of Observation \ref{cor:path-above-top-segment}. So $L_2$ should only cover bottom segments, which means it
    must be a cover-line loop by Lemma \ref{lem:cover-line-loop}.
\end{proof}

	Using these lemmas, it follows that there are at most $2$ overlapping loops with entry points on opposite sides of $\Gamma$. Furthermore, if there are two such loops, then one of them is a cover-line loop.    
    
    We will finally show that there are at most $2$ overlapping loops that have both their entry points on the same side, say left of $\Gamma$. This will imply the result of the lemma for loops, because on each of the cover-lines, there are at most $2$ loops with entry points on the left of $\Gamma$, $2$ with entry points on the right, and $2$ with entry points on the opposite sides.
    Between the loops with both entry points to the left of $\Gamma$, none can be a cover-line loop because such a loop cannot intersect $\Gamma$ ($\Gamma$ needs to be between the two entry points of a cover-line loop). We will show that there will be at most $2$ (non-cover-line) overlapping loops with entry points to the left of $\Gamma$.
    
    For the sake of contradiction, assume that there are at least $3$ loops with entry points on the left of $\Gamma$ that none are cover-line loops and all cross $\Gamma$. Let $L_1, L_2$, and $L_3$ be any $3$ consecutive loops with this property. 
    Without loss of generality let $x(e_1) \le x(o_1)$, $x(e_2)\leq x(o_2)$, and $x(e_3)\leq x(o_3)$, and assume an order for the entry points of $L_m$'s, say $x(e_1) \le x(e_2) \le x(e_3)$. We must have $x(e_2) \ge x(o_1)$; or else $L_1,L_2$ must be nested by Lemma \ref{claim:loop-above-loop}, implying $L_2$ should be a cover-line loop which contradicts the assumption.
    So we get that $x(e_2) \ge x(o_1)$. Similarly, we have $x(e_3) \ge x(o_2)$. 
    These imply that $e_1,o_1,e_2,o_2,e_3,o_3$ appear in this order on the bottom cover-line.
    Corollary \ref{cor:non-cover-line-loop-top-segment} implies each of $L_1, L_2,$ and $L_3$ must exclusively cover some top segment.
    % Note that with the same argument as in Claim \ref{claim:loop-above-loop}, each $L_i$ has to cover some top segment, because otherwise it has to be a cover-line loop by Lemma \ref{lem:cover-line-loop}. 
    Let $r_1,r_2$ be the right-most point on $L_1,L_2$, respectively. 
    Since each $L_1,L_2$ starts and ends on the left of $\Gamma$ and travels to the right of $\Gamma$, by Lemma \ref{obs:further-point-reflection}, the right-most point on each is a reflection point, which implies it must be exclusively covered by using Lemma \ref{obs:reflection}. Let $s_{i_1}$ be the segment that
    reflection point $r_1$ lies on, and similarly $s_{i_2}$ the segment for $r_2$ (see Figure \ref{fig:fig4}).

\begin{lemma}\label{claim:top-segments-on-right}
    $s_{i_1},s_{i_2}$ are top segments and $x(s_{i_1})<x(s_{i_2})$.
\end{lemma}

\begin{proof}[Proof of Lemma \ref{claim:top-segments-on-right}]
	By way of contradiction, assume $s_{i_1}$ is a bottom segment. 
	Consider the two subpaths of $L_1$ between the entry points $e_1,o_1$ and $r_1$, let us denote them by $P^1_r: e_1 \to r_1$ and $P^1_l:r_1 \to o_2$. 
	$L_2$ (starting at $e_2$) is in the region bounded by $P^1_r\cup s^1_r$ and the bottom cover-line, which means $L_1$ will intersect any top segment
	$L_2$ intersects with (i.e. $L_2$ cannot exclusively cover any top segment), which implies $L_2$ is a cover-line loop, a contradiction.
	This implies that $s_{i_1}$ is a top segment.
	A similar argument (for $L_2,L_3$) implies $s_{i_2}$ is a top segment.
	
	We show that $x(s_{i_1}) \le x(s_{i_2})$. Similar to before, define the subpath $P^1_r$ of $L_1$ that goes from $e_1$ to $r_1$ and $P^2_r$ from $e_2$ to $r_2$. Considering the two areas of strip $S_\tau$ separated by $P^1_r \cup s_{i_1}$, if segment $s_{i_2}$ is on one side and the entry points $o_1, e_2$ on the other side, then path  $P^2_r$ must either intersect $P^1_r$ or $s_{i_1}$, which is not possible (due to Lemma \ref{obs:reflection}). So, $s_{i_2}$ and $e_2,o_2$ are on the same part of $S_\tau$ cut by $P^1_r\cup s_{i_1}$. This implies $s_{i_2}$ is to the right of $s_{i_1}$ i.e $x(s_{i_1}) \le x(s_{i_2})$. 
\end{proof}
    
    We can reuse the same arguments in the second part of the proof to conclude the following lemma:
    \begin{lemma}\label{claim:no-bottom-cover}
        Neither of $L_1$ or $L_2$ exclusively cover a bottom segment on the right of $\Gamma$.
    \end{lemma}
    
    \begin{proof}[Proof of Lemma \ref{claim:no-bottom-cover}]
        Let $L_j$ be either one of $L_1$ or $L_2$. Assume the contrary, that there is some bottom segment $s_j$ to the right of $\Gamma$ that $L_j$ exclusively covers. So $L_3$ does not intersect with this segment. Let $p_j$ be the last intersection point of $L_j$ with $s_j$. Consider the subpath $P_j: p_j \to o_j$ on $L_j$. Similar to the proof of Lemma \ref{claim:top-segments-on-right}, we get that both entry points of $L_3$ are surrounded by $P_j \cup s_j$ from the right or above; which means any top segment that $L_3$ intersects with, is already intersecting with $P_j$. This requires $L_3$ to be a cover-line loop, giving us a contradiction. 
    \end{proof}
    
    Now we define an alternate path that replaces $L_1$ and $L_2$ with two new loops that no longer overlap at $\Gamma$, and overall the shadow does not increase but also costs less than the cost of the current solution. The idea of this change (which will be made precise soon) is to follow $L_1$ from $e_1$ to the right-most point on $L_1$ (which must be a reflection on $s_{i_1}$), then from that point follow a horizontal line until it hits $s_{i_2}$; if there are portions of $L_2$ that are above this horizontal line, we follow the  upper envlope of those portions of $L_2$ and the horizontal line (similar to how we reduced the shadow in the case of zig-zag or sink), and then from the intersection point on $s_{i_2}$, follow the horizontal line back to the right-most reflection on $L_1$ and continue to follow $L_1$ to $o_1$; $L_2$ is going to be simply replaced with a smaller subset of its projection on the bottom cover-line. We show we will have a cheaper feasible solution with smaller shadow at $\Gamma$, a contradiction. Now we describe this more precisely.

    Again, let $r_1,r_2$ be the right-most points on $L_1,L_2$, respectively. Lemmas \ref{claim:top-segments-on-right} and \ref{obs:further-point-reflection} imply that they are reflection points on segments $s_{i_1}, s_{i_2}$, respectively, which both are top segments. Consider the horizontal line $y = y(r_1)$, and let $q$ be the intersection point of this line with the vertical line $x = x(s_{i_2})$.
    Define the subpath $P^2_r$ on $L_2$ as $P^2_r: e_2\to r_2$ (assuming that $e_2$ is to the left of $o_2$), In other words, between the two paths from $r_2$ to the two entry points of $L_2$, $P^2_r$ is the one that is above the other. Let $U_2$ be the portion of 
    $P^2_r$ in the region bounded by lines $s_{i_1}\cup (y = y(r_1)) \cup s_{i_2}$, and the top cover-line. 
So these are the portions of $P^2_r$ that go above the line segment $r_1q$ (see Figure \ref{fig:fig4}). 
    Let $L''_1$ be the upper envelope of
    $U_2\cup r_1q$ plus the line $r_1q$. 

    So $L''_1$ consists of a path that goes on the upper envelope of
    $U_2\cup r_1q$ from $r_1$ to $q$ and then goes straight back to $r_1$.
    We now define the replacements for $L_1$ and $L_2$.

\begin{figure}[h]
	\centering
    \includegraphics[width=.5\textwidth]{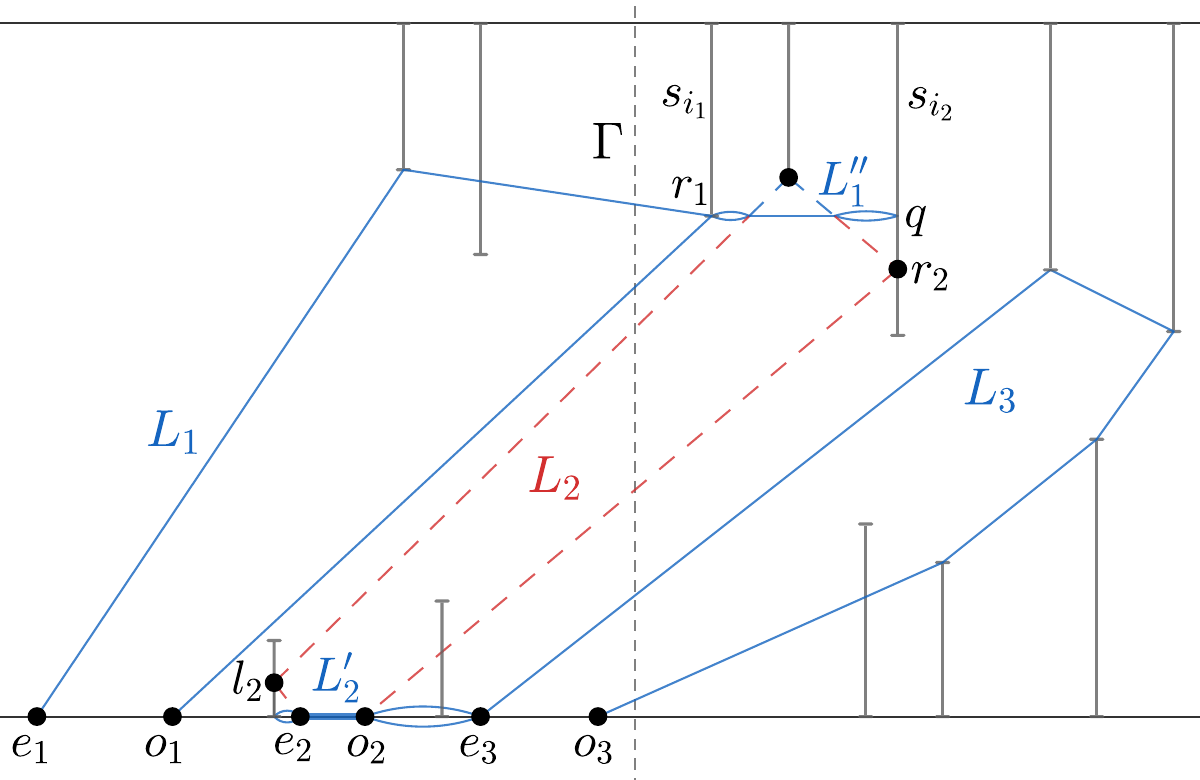}
    \caption[Alternative solution for 3 overlapping loops.]{Alternative solution for 3 overlapping (non-cover-line) loops. Pairs of arcs represent doubled segments}
    \label{fig:fig4}
\end{figure}
We replace $L_1$ with $L'_1$ as follows:
    \begin{itemize}
        \item Take the subpath $P^1_r: e_1\to r_1$ on $L_1$.
        %in the same fixed orientation on $L_1$.
        \item From $r_1$, follow $L''_1$ and thus, get back to $ r_1 $.
        \item From $r_1$, follow the rest of $L_1$ to $o_1$.
    \end{itemize}
 
    So $L'_1$ is obtained by adding $L''_1$ to $L_1$ at $r_1$.
    If $l_2$ is the left-most point that $L_2$ travels, then let
    $L'_2$ be a cover-line loop that travels from $e_2$ left to $x(l_2)$, then right to $e_3$ and then back to $o_2$ (this is essentially the projection of the portions of $L_2$ to the left of $e_3$ and hence to the left of $\Gamma$ on the bottom cover-line); recall that we can reduce $L'_2$ to remove the possible overlapping of its legs.     
     Now replace $L_2$ with $L'_2$. We will show that these two loops in total cost strictly less than $L_1$ and $L_2$, the shadow does not increase (and in fact shadow decreases at $\Gamma$) and we still have a feasible solution. It's clear to see that between the loops $L'_1, L'_2,$ and $L_3$, only $L'_1$ and $L_3$ overlap at $\Gamma$. So we decreased the number of overlapping loops at $\Gamma$ by at least one.
    
    To prove all segments are still covered, note that $L'_1$ includes the entirety of $L_1$, and thus covers all the segments that $L_1$ used to cover. In order to show that all the segments that $L_2$ covered, are still covered, we only need to show that the segments that $L_2$ exclusively covered, are still covered. That is because in the new configuration we still have all the parts of $L_1$ and $L_3$.
    According to Lemma \ref{claim:no-bottom-cover}, there are no bottom segments that $L_2$ exclusively covers to the right of $\Gamma$. Also, it is easy to see that any bottom segment that was exclusively covered by $L_2$ to the left of $\Gamma$ must have an $ x $-coordinate between $x(l_2)$ and $x(e_3)$. All of those bottom segments are now covered by $L'_2$.
    Finally, for the top segments that $L_2$ exclusively covers, with the same arguments as in the second part of the proof in Lemma \ref{claim:no-bottom-cover}, we get that there are no such segments to the left of $s_{i_1}$. So it suffices to show that only the top segments that $L_2$ covers to the right of $s_{i_1}$, are covered. This is easy to see, because $L''_1$ includes the entire $U_2$; and it is always on or above $L_2$ in the range between $s_{i_1}$ and $s_{i_2}$. Thus, $L''_1$ will cover all the top segments that $L_2$ exclusively covers in that range.
    
    Also, the shadow does not increase: the shadow of $L''_1$ from $r_1$ to $r_2$ can be charged to the sections
    of $L_2$ between $x=x(r_1)$ and $x=x(r_2)$ and hence is no more than that; note that this portion is entirely to the right of $\Gamma$. The cover-line loop $L'_2$ is entirely to the left of $\Gamma$ and its shadow can be charged to the shadow of $L_2$ to the left of $\Gamma$ in the range $[x(l_2), x(e_3)]$. 
    
    Now let's prove that the new cost is decreased compared to $L_1$ and $L_2$. $L'_1$ includes $L_1$, so we set aside those parts and charge them on $L_1$. So it suffices to show that $L''_1$ along with $L'_2$ can be charged into $L_2$. Note that $L'_2$ is part of the projection of $L_2$ on the bottom cover-line to the left of $e_3$. So the cost of $L'_2$ is strictly less than the cost of $L_2$ to the left of $\Gamma$, since $L'_2$ extends at most to $e_3$ which is to the left of $\Gamma$. As for $L''_1$, note that $L_2$ travels back and forth between $x(s_{i_1}),x(s_{i_2})$; so $L''_1$ can be charged to these two sections of $L_2$ between $x(s_{i_1})$ and $x(s_{i_2})$.

    So at the end, we found a new solution with $1$ fewer overlapping loops at $\Gamma$, no increase of shadow elsewhere, and with a strictly less cost than $\OPT_\tau$. Applying this argument implies that at most two overlapping non-cover-line loops can exist to the left of $\Gamma$. So in total on each of the cover-lines of $S_\tau$, there are at most $2$ non-cover-line loops to the left of $\Gamma$, similarly $2$ to the right, plus at most 2 with entry points to opposite sides of $\Gamma$. In total, there are at most $6\times 2 = 12$ overlapping loops at $\Gamma$.

    \subsubsection*{Overlapping Ladders}
    Recall that by Definition \ref{def:loops/ladders}, ladders are subpaths of $\OPT_\tau$ (in strip $S_\tau$) that have one entry point on the bottom cover-line of $S_\tau$, and one on the top cover-line. 
    Depending on their orientation compared to $\Gamma$, there are two types of ladders (see Figure \ref{fig:type1-type2-ladders}):  
    
    \begin{itemize}
	    \item  \textbf{Type 1 Ladder:} Has both its entry points on the same side of $\Gamma$.
	    \item  \textbf{Type 2 Ladder:} Has its entry points on  opposite sides of $\Gamma$.
    \end{itemize}
    \begin{figure}[h]
        \centering
        \includegraphics[width=0.38\tw]{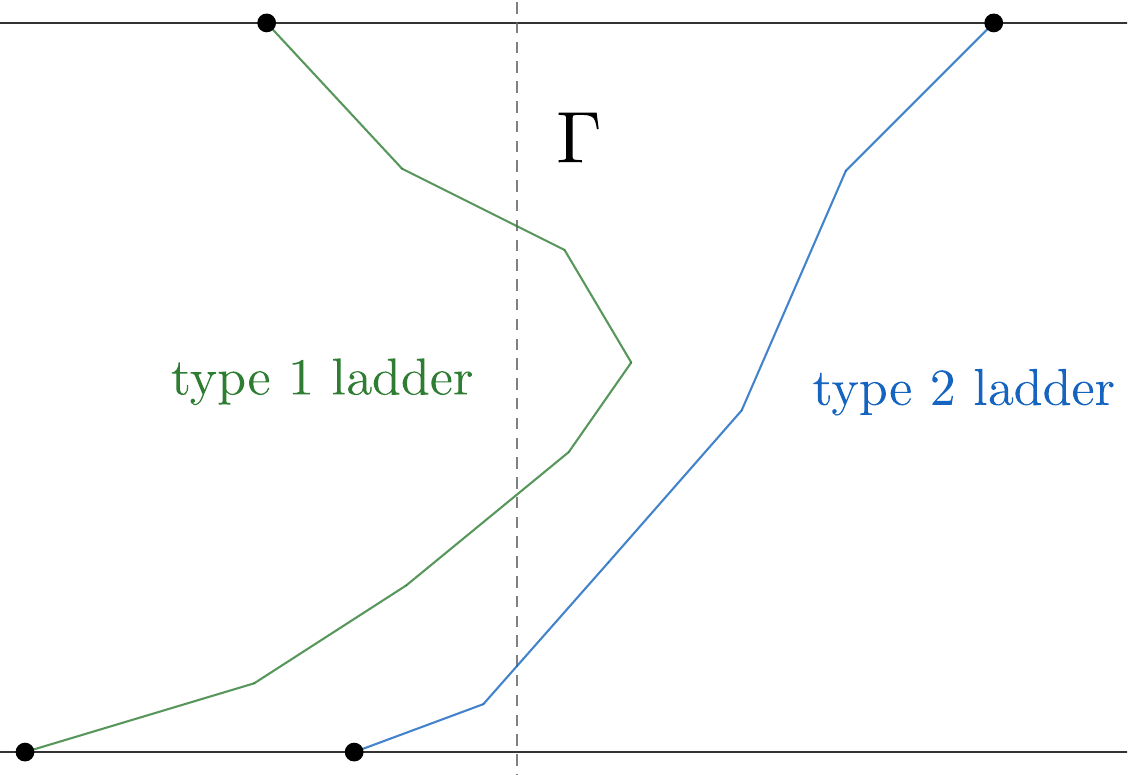}
        \caption{An example of type 1 and type 2 ladders}
        \label{fig:type1-type2-ladders}
    \end{figure}

    We will prove that there are at most $2$ overlapping Type 1 ladders, and at most $5$ overlapping type 2 ladders.

    \subsubsection*{Type 1 Ladders}
    In particular, we show that there is at most one Type 1 ladder with entry points to the left of $\Gamma$, and one with entry points to its right.
    To prove this, assume the contrary, that there are at least $2$ overlapping 
    Type 1 ladders with entry points to the same side, say right of $\Gamma$. Let $L_1$ and $L_2$ be two such ladders. 
   
    Let $(b_m, t_m),\; m = 1,2$ be the entry points of $L_m$ on the bottom cover-line and the top cover-line, respectively. Without loss of generality, assume that $t_1$ is to the left of $t_2$. 
   This implies $b_1$ is also to the left of $b_2$ (or else $L_1$ and $L_2$ intersect inside $S_\tau$). So if we consider cutting $S_\tau$ along $L_2$, $L_1$ is entirely in one of the two regions created, namely the one that contains $b_1,t_1$.     
    Since both $L_1$ and $L_2$ overlap at $\Gamma$ and are Type 1, and they both have their entry points on the same side of $\Gamma$, say  left, this means that both have to reach to the right of $\Gamma$. 
    First, we show that the top-most and bottom-most intersection point of $L_1,L_2$ with $\Gamma$ must be on $L_2$. By way of contradiction, suppose $p$ is a point on $L_1$ and is the bottom-most intersection of these two ladders on $\Gamma$. Consider the subpath of $L_1$ from $b_1$ to $p$, call it $L'_1$ and consider the region bounded by $L'_1\cup\Gamma$ and the bottom cover-line, call it $A$. Since $L_2$ starts at $b_2$ inside $A$ and $t_2$ is outside $A$, $L_2$ must either cross $\Gamma$ at a point lower than $p$, or cross $L'_1$, both of which are contradictions. A similar argument shows the top-most intersection point on $\Gamma$ is with $L_2$.
    
    Consider any two consecutive crossing of $L_1$ with $\Gamma$, say $p_1,p_2$, where the subpath of $L_1$ from $p_1$ to $p_2$ (denoted by $L'_1$) is to the right of $\Gamma$. Since $L_2$ crosses $\Gamma$ both above and below $p_1,p_2$ (the lowest and highst intersection points on $\Gamma$ are with $L_2$), there is a subpath of $L_2$ with end-points $q_1,q_2$ on $\Gamma$ with $q_1$ below $p_1,p_2$, and with $q_2$ above them, call it $L'_2$. We consider two cases based on whether $L'_2$ is on the left or right of $\Gamma$, and derive contradictions in each case. If $L'_2$ is on the right (like $L'_1$) then $L'_1$ is inside the
    region bounded by $L'_2\cup q_1q_2$ and this violates Lemma \ref{lem:nested-vertical-loops}.
    So let us assume $L'_2$ is on the left of $\Gamma$. Since $p_1,p_2$ are between $q_1,q_2$ there is subpath of $L_1$ starting from $p_1$ inside the region $L'_2\cup q_1q_2$ that crosses $q_1q_2$. This subpath with $L'_2$ violates Lemma \ref{lem:nested-vertical-loops} again.
    %Following $L_1$ from $e_1$, consider the last crossing of $L_1$ with $\Gamma$ below $q_1,q_2$, call it $p_1$ ($p_1$ must exist because the lowest crossing on $\Gamma$ is with $L_1$). So following $L_1$ from $p_1$, the subpath will cross $\Gamma$ again at a point we call $p_2$ which is above $q_1,q_2$; call this subpath $L'_1$. 
%    So $L'_2$ is inside the region bounded by $L'_1\cup p_1p_2$. This violates Lemma \ref{lem:nested-vertical-loops}. 
    Thus, we conclude that there can be at most $1$ Type 1 ladder with entry points to the right of $\Gamma$, and similarly, at most $1$ with entry points to the left of $\Gamma$. 
    
    \subsubsection*{Type 2 Ladders}
    For each Type 2 ladder $L_m$ with entry points $(b_m, t_m)$ on bottom and top cover-lines, there are two cases:
    \begin{itemize}
        \item $b_m$ is to the left of $\Gamma$, therefore $t_m$ is to the right of $\Gamma$. We say $L_m$ is a {\em top-right/bottom-left} ladder.
        
        \item $b_m$ is to the right of $\Gamma$, therefore $t_m$ is to the left of $\Gamma$. We say $L_m$ is a {\em top-left/bottom-right} ladder.
    \end{itemize}
    There can't be two overlapping ladders that one is a top-right/bottom-left ladder, and the other is a top-left/bottom-right ladder (or else they intersect). So if we have a collection of Type 2 overlapping ladders they are all either top-right/bottom-left or all top-left/bottom-right. We show we can have at most 5 Type 2 overlapping ladders.
    For the sake of contradiction, assume there is a maximal set $\mathcal L = \{L_1, L_2, \dots, L_k\}$ of Type 2 ladders that all overlap at some vertical line $\Gamma$ with $k\ge 6$ and all are top-right/bottom-left. 
    Let $(b_m, t_m)$, $1\leq m\leq k$ denote the bottom and top entry points of ladder $L_m$. Without loss of generality, assume that $x(b_1) \le x(b_2) \le \dots \le x(b_k)$, which also implies $x(t_1) \le x(t_2) \le \dots \le x(t_k)$ (or else the ladders will be intersecting each other). Let $L^l_m$ be the subpath of of $L_m$ from $b_m$ to the first intersection of $L_m$ with $\Gamma$ (so $L^l_m$ is to the left of $\Gamma$), and $L^r_m$ be the subpath of $L_m$ from its last intersection with $\Gamma$ to $t_m$ (so it is to the right of $\Gamma$).
    Note that if $m<m'$ then $L^l_m$ is above $L^l_{m'}$ (in the range that $L^l_m$ is defined) and $L^r_{m'}$ is below $L^r_{m}$ (in the range that $L^r_{m'}$ is defined) due to Lemma \ref{lemma:left-path-is-above}.
    Using Observation \ref{cor:path-above-top-segment}, this implies $L^l_1$ covers all the top segments that $L^l_2, L^l_3,\dots, L^l_k$ cover to the left of $\Gamma$ and, similarly, $L^r_k$ covers all the bottom segments that $L^r_1, L^r_2,\dots, L^r_{k-1}$ cover to the right of $\Gamma$ (we will use this fact shortly).

    We will introduce an alternate set of ladders (and loops) that cover all the segments the ladders in $\mathcal L$ cover without increasing the shadow anywhere, with a cost strictly smaller cost, and with a smaller shadow at $\Gamma$. The set of ladders we introduce differ based on the parity of $k$. For odd $k$ we keep $L_1,L_{k-1},L_k$, and for even $k$ we keep $L_1,L_3,L_{k-2},L_k$. We also add some cover-line loops (possibly two 
    copies) to make sure we still have a tour that visits all the points $b_j,t_j$ and 
    all the top and bottom segments that $L_1,\ldots,L_k$ covered remain covered in the new solution
    \footnote{This change is somewhat similar to the proof of patching lemma used in the PTAS for Euclidean TSP
    that reduces the number of crossings into a region.}.

    Imagine a graph $G(V,E)$ where $V$ consists of all $b_j,t_j$'s and there are edges between two vertices if there is a subpath in $\OPT$ between them without visiting any vertex (so we have direct edge between $b_j,t_j$ and also an edge between $b_j,t_i$ if there is a path in $\OPT$ between them outside the strip  $S_\tau$).
    Note that $G$ is simply a cycle. In the new alternative solution, we keep $L_1,L_k$ and either $L_{k-1}$ or both $L_3,L_{k-2}$ (depending on the parity of $k$) and add
    cover-line loops between some consecutive $b_j$'s and consecutive $t_j$'s such that the resulting graph $G'$ defined based on these new paths still forms an Eulearian (connected) graph on $V$, all the segments covered in $S_\tau$ by  $L_1,\ldots,L_k$ are covered.
 Let $b'_2$ be the projection of the left-most point on $L_2$ on the bottom cover-line, and let $t'_{k - 1}$ be the projection of the right-most point on $L_{k - 1}$ on the top cover-line. 
    We add doubled segment $b_2b'_2$ and $t_{k-1}t'_{k-1}$. These intend to cover any bottom segment exclusively covered by $L_2$ to the left of $b_2$, and any top segment exclusively covered by $L_{k-1}$ to the right of $t_{k-1}$. 
The doubled segments $b_2b'_2$ and $t_{k - 1}t'_{k - 1}$ fully appear in the projection of $L_2$ and $L_{k - 1}$ on those cover-lines; meaning that they can be charged onto $L_2$ and $L_{k - 1}$ that travel left (and right)  to those segments, respectively. Add each of these two segments twice to the solution. Since we're adding these segments twice, the parity of the degree of nodes in $G'$ won't change. 
    We keep $L_1,L_k$ from $\mathcal L$ and add the following segments and ladders as well to the alternative solution (see Figure \ref{fig:fig5}):

       \begin{itemize}
               \item If $k=2m$ for some integer $ m \ge 3 $, then include $L_3$ and $L_{k - 2}$. We also add the following cover-line loops:                
            \begin{itemize}
                \item $ b_2b_3, \ b_{k - 2} b_{k - 1},\ b_{2q - 1} b_{2q}\ (2\le q\le m - 1) $
                \item $ t_2t_3,\ t_{k - 2} t_{k - 1},\ t_{2q - 1} t_{2q}\ (2\le q\le m - 1)$
            \end{itemize}           
        We add double the following cover-line loops (i.e. a path back and forth on the same pair of points):
            \begin{itemize}
                \item $ b_{k - 1} b_k, b_{2q} b_{2q + 1}\ (2\le q\le m - 2)$
                \item $ t_1 t_2,\ t_{2q} t_{2q + 1}\ (2\le q\le m - 2) $
            \end{itemize}
        \item If $ k = 2m + 1$ for some integer $ m\ge 3 $, then we include $L_{k - 2}$ and also
            the following cover-line loops:
            \begin{itemize}
                \item $ b_{k - 2}b_{k - 1},\ b_{2q} b_{2q + 1}\ (1\le q \le m - 2)$                
                \item $ t_{2q}  t_{2q + 1}\ (1\le q \le m - 2)$
            \end{itemize}
        We add double the following segments:
            \begin{itemize}
                \item $ b_{k - 1} b_k,\ b_{2q - 1} b_{2q}\ (2\le q\le m - 1) $
                \item $ t_{2q - 1} t_{2q}\ (1\le q\le m) $
            \end{itemize}
    \end{itemize}
    
\begin{figure}[H]
    \centering
    \includegraphics[width=.65\textwidth]{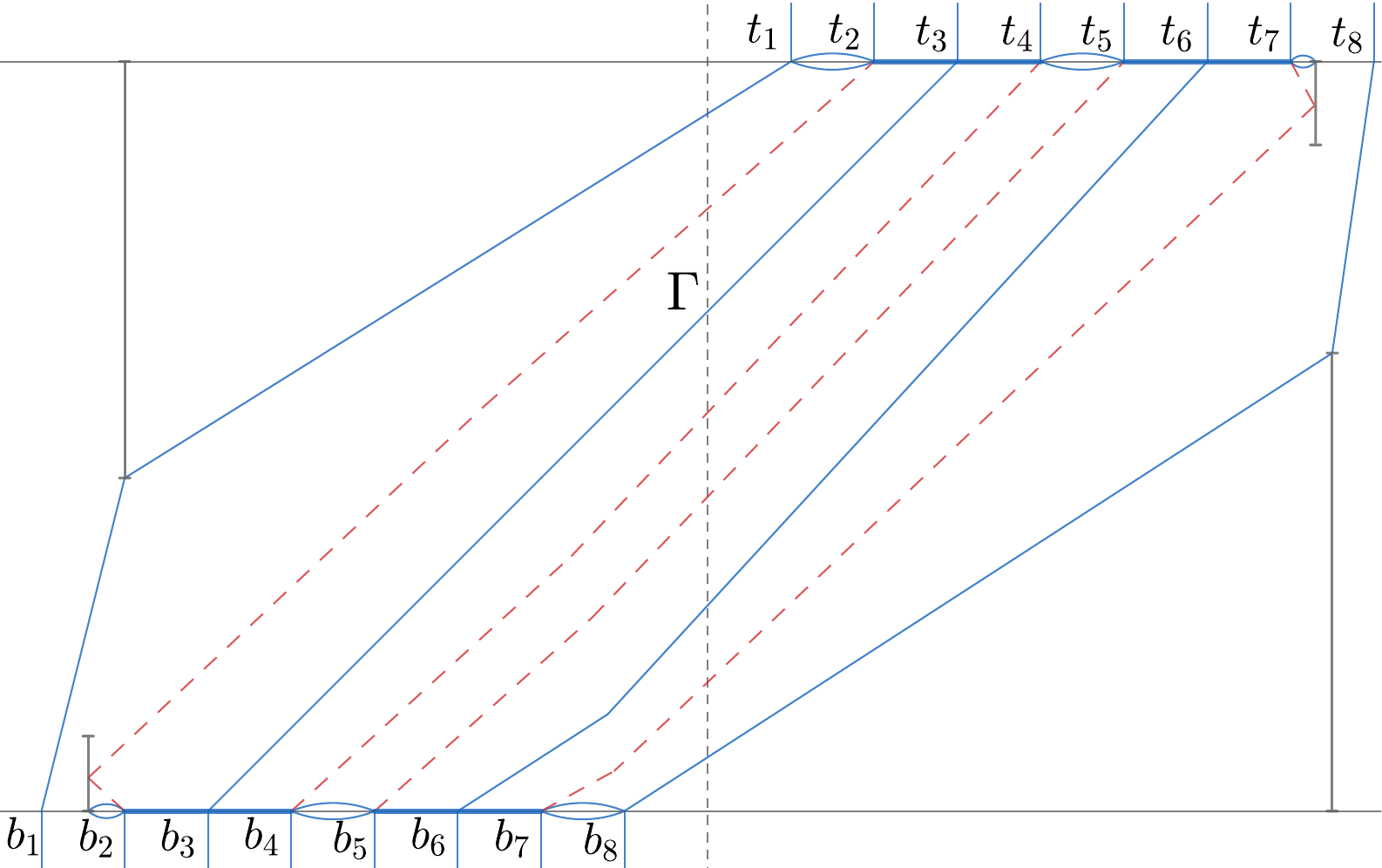}
    \caption[Alternative solution for $8$ overlapping ladders.]{Alternative solution for $8$ overlapping (bottom-left/top-right) ladders. Red dashed lines are discarded. The arcs represent the doubled segments}\label{fig:fig5}
\end{figure}

    It can be seen that with the above additions, if we build the graph $G'$ based on the new paths it is an Eulerian graph as each $b_j,t_j$ has even degree; also $G'$ remains connected since
    all the $t_1,\ldots,t_{k-1}$ are connected via cover-line loops added at the top and $b_2,\ldots,b_k$ are connected via cover-line loops at the bottom and we have $L_1,L_k$ and there is a path from $b_1$ to at least one of $b_2,\ldots,b_k$ in outside the strip, and similarly a path from $t_k$ to one of $t_1,\ldots,t_{k-1}$. Thus in the new solution we visit all the points $b_j,t_j$ and this tour can be short-cut over repeated points to obtain a new solution that visits all the $b_j,t_j$'s and covers all the segments outside the strip $S_\tau$. 

    Next we show all the segments that $L_1,\ldots,L_k$ were covering, remain covered. Recall that the portion of $L_1$ to the left of $\Gamma$ covers all the top segments that were covered by these paths to the left of $\Gamma$, and similarly $L_k$ covers all the bottom segments that were covered to the right of $\Gamma$. The bottom segments covered to the left of $\Gamma$ are covered by the 
    new cover-line loops added and similarly, the top segments covered to the right of $\Gamma$ are covered by the cover-line loops added. So the new solution remains feasible.
    
    Now we are going to bound the total cost of the new solution. We charge all the new parts that we added to some portion of the ladders that we have discarded. Note that in every case, $L_2, L_4, L_{k - 3},$ and $L_{k - 1}$ are discarded. We will use only these ladders to charge the new parts to.
    The doubled segments $b_2b'_2$ and $t_{k - 1}t'_{k - 1}$ are already charged to the portion of $L_2$ travelling in the interval $[x(b'_2), x(b_2)]$ and the portion of $L_{k - 1}$ travelling in $[x(t_{k - 1}), x(t'_{k - 1})]$. 
    Now consider the ranges $\beta_j = [x(b_{j - 1}), x(b_{j})]$, $3\leq j\leq k$ and $\theta_j = [x(t_{j - 1}), x(t_{j})]$, $2\leq j\leq k - 1$. These are disjoint, and all $\beta_j$'s lie under $L^l_2$ and $L^l_4$, while all $\theta_j$'s lie above $L^r_{k-3}$ and $L^r_{k-1}$. 
    Each of the new included segments (doubled or not) can be charged to one or two of the ladders $L_2, L_4, L_{k-3}, L_{k - 1}$.

    It can be seen that in the new configuration, there are at most 4 overlapping ladders and 2 overlapping loops (doubled segment loops that we added). This concludes the case for ladders.

    In general, when given a collection of loops and ladders, we first alter the ladders as described above (and might get some new cover-line loops in the process), then we apply the alteration on the loops. The statement of Lemma \ref{lemma:overlapping-loops-ladders} follows from this. \qed

%%%%%%%%%%%%%%%%%%%%%%%%%%%%%%%%%%%%%%%%%%%%%%%%%%%%%%%%%%%%%%%%%%5
\subsection{Proof of Theorem \ref{thm:exact-answer}}\label{sec:exact-answer}
    This section is dedicated to the proof of Theorem \ref{thm:exact-answer}.
    Let $ C_1,\dots, C_\sigma $ be the cover-lines (as defined in Subsection \ref{sec:bounded-shadow}) for an instance of the problem with $ H\le 3 $. It can be seen that $ \sigma \le 2 $; in other words, all the segments of the instance can be covered with only at most 2 cover-lines. If $ H \le 2 $, then the number of cover-lines is 1, and it can be seen that the portion on that cover-line itself (doubled from the left-most segment to the right-most segment) is an optimum solution. So let's assume $ 2 < H \le 3 $, therefore $ \sigma = 2 $, and that we have a single strip, $ S_1 $. Furthermore, there must be both top segments and bottom segments in $ S_1 $ (otherwise one of the cover-lines would intersect with all segments). We will essentially prove that the optimum solution must be a bitonic tour.
	
	Take any optimum solution $ \OPT $ for this instance of the problem, and let $ p^l $ and $ p^r $ be the left-most and right-most points on it, respectively. There is a path $ P_1 $ from $ p^l $ to $ p^r $, and there is a path $ P_2 $ in the other way. Since $ \OPT $ is not self-intersecting, and since both $ P_1 $ and $ P_2 $ cover the range $ I = [x(p^l), x(p^r)] $, then for any vertical line $ \Gamma $ with $ x(\Gamma) \in I $, they both will intersect with it at distinct points. 
	We can use Lemma \ref{lem:nested-vertical-loops} (for the concatenation of $ P_1 $ and $ P_2 $ restricted to the left of $ \Gamma $) to get that $ p^l $ is a right reflection. Similarly, $ p^r $ is a left reflection. 
	
	Without loss of generality, assume that $ P_1 $ includes the upper leg of $ p^l $, and thus $ P_2 $ includes its lower leg. Using Lemma \ref{obs:reflection-paths-above-below} for the reflection point $ p^l $ and the vertical line $ x = x(p^r) $, we get that $ P_1 $ is above $ P_2 $ in range $ I $, which is the entirety of $ \OPT $.
	Observation \ref{cor:path-above-top-segment} implies that all the top segments are covered by $ P_1 $, while all the bottom segments are covered by $ P_2 $.
	We claim that there are no reflection points other than $ p^r $ and $ p^l $ in $ \OPT $. To see why this is the case, assume the contrary, that there is some reflection point $ r $ on $ \OPT $ other than those two points.
	
	Without loss of generality, assume $ r \in P_1 $, and assume that $ r $ is the first such reflection point on $ P_1 $ after $ p^l $. According to Lemma \ref{obs:consecutive-reflections}, $ r $ is a right reflection. Let $ s $ be the segment of the instance that $ r $ lies on. If $ s $ is a bottom segment, then $ P_2 $ will be intersecting with it, and we get a violation of Lemma \ref{obs:reflection}. Thus, $ s $ is a top segment.
%	This means that following $ p^l $ to $ s $, the value of the $ x $-coordinate among the points on $ P_1 $ is increasing, and there is a decrease at $ r $. 
%	Since $ P_1 $ has to reach as far right as $ p^r $, this means there is must be another reflection $ r' $ in $ P_1 $ after $ r $. According to Lemma \ref{obs:consecutive-reflections}, $ r' $ is a left reflection and $ x(r') < x(r) $. 
%	Let $ s' $ be the segment that $ r' $ lies on. If $ s' $ is a bottom segment, then $ P_2 $ will already be intersecting with it; if $ s' $ is a top segment, then 
	
	Now, let $ \mathcal P_1 $ be the concatenation of $ P_1 $ (restricted to the subpath from $ p^l $ to $ r $) along with the entirety of $ P_2 $. $ \mathcal P_1 $ is a path that goes from $ r $ (a left reflection on a top segment $ s $) and reaches to the right of $ s $. The rest of the path of $ P_1 $ (from $ r $ to $ p^r $), refer to it as $ \mathcal P_2 $, is another path that goes from $ r $ and reaches to its right. Depending on whether the top leg of $ r $ belongs to $ \mathcal P_1 $ or $ \mathcal P_2 $, we get a violation of Lemma \ref{claim:top-segments-on-right}. This contradiction shows that such $ r $ cannot exist, and that both $ P_1 $ and $ P_2 $ are monotone paths with shadow 1, due to Lemma \ref{obs:shadow-increase}. So in total, $ \OPT $ has a shadow of 2, as was to be shown. 
    \qed

    \vspace{.5cm}
    \noindent \textbf{Note.} It can be shown that in these special cases, we can find an exact solution in poly-time. But since we made some assumptions about the $ x $-coordinates of the segments of the instance, we have to undo them to prove this claim. The resulting algorithm will be somewhat detailed for such a limited special case of the problem, because we have to cover cases such as vertical legs in an optimum solution. So we only settled on showing that an optimum solution has a constant shadow instead, as it's enough for the purposes of our main theorem in this paper.

%%%%%%%%%%%%%%%%%%%%%%%%%%%%%%%%%%%%%%%%%%%%%%%%%%%%%%%%%%%%%%%%%%%%%%%%%
\section{Extensions and Concluding Remarks}\label{sec:conclusion}

Using Theorem \ref{thm:main}, we can get a $(2+\eps)$-approximation for the setting where the input line segments are axis-parallel and have bounded aspect ratio. First we consider the unit-length axis-parallel segments:

\begin{theorem}
    Given $n$ unit length segments $s_1, s_2,\dots, s_n$ as an instance of TSPN that are each parallel to either the $x$-axis or the $y$-axis, there is a $(2+\epsilon)$-approximation with running time in $n^{O(1/\eps^3)}$.
\end{theorem}
\begin{proof}
    Split the segments into two groups based on them being horizontal or vertical. Let the minimal bounding box for the vertical segments have sides $ L_v\times H_v $, and the one for horizontal segments have sides $ L_h \times H_h $. Similar to what was mentioned at the start of Section \ref{sec:main-alg}, if $ \opt $ is the cost of an optimum solution $ \OPT $, then $ \opt/2 \ge \max\{ L_v, H_v - 2, L_h - 2, H_h \}$.
		
    Consider the boxes of sizes $ L_v\times (H_v - 2) $ and $ (L_h - 2)\times H_h $ contained in the aforementioned minimal bounding boxes. Let $ B $ be the smallest bounding box that contains these two new boxes. So we get that $ \opt $ is at least as large as any sides of $ B $; and also it is the case that $ \OPT $ lies completely inside of $ B $. 
    
    The left side of $ B $, refer to it as $ B_l $, either has a vertical segment on it (in the case when the left side of the $ L_v\times (H_v - 2) $ box overlaps with $ B_l $), or it has the right-most point of the left-most horizontal line (in the case when the left side of the $ (L_h - 2)\times H_h $ box overlaps with $ B_l $). The same argument holds for $ B_r $, the right side of $ B $. If neither of $ B_l $ and $ B_r $ are on a side of the $ (L_h - 2)\times H_h $ box, then it means that all the horizontal segments of the problem have an intersection with the interior of $ B $. In this case, take any horizontal segment $ s_h $ that has a length of $ l_h $ inside of $ B $; we get that $ \opt \ge l_h $.
    
    Assuming $ \opt \ge l_h $, take the portion of $ s_h $ lying inside $ B $, and break it into $ 8/\eps $ parts of size $ l_h\cdot \eps/8 $. For each of these parts, consider their left-most points, and let them be $ p_1, p_2,\dots, p_{8/\eps} $. The $ \OPT $ must intersect this segment at one of these parts. Assume that $ p_i $ is the left-most point of the part that $ \OPT $ intersects with (we can check all the $ 8/\eps $ cases). 
    Add $ p_i $ to the set of vertical segments, and apply the result of Theorem \ref{thm:main} for parameter $ \eps/4 $ to get a solution covering all the vertical segments along with point $ p_i $.
    This solution is a lower bound for the restriction of $ \OPT $ on $ s_h $ along with the vertical segments; with the exception that the intersection on $ s_h $ itself can add at most $ l_h\cdot \eps/4 $ to the cost. So in total, this new solution, along with a doubled copy of the part containing $ p_i $, cost at most $ (1 +\eps/4)\cdot \opt + l_h\cdot\eps/4 \le (1 + \eps/2)\cdot \opt $.
    Do the same thing for horizontal segments, meaning find a solution covering $ p_i $ and all the other horizontal segments with parameter $ \eps/4 $ in Theorem \ref{thm:main}. We get another solution with cost at most $ (1 + \eps/2)\cdot \opt$.
    The conjunction of these two solutions, make a feasible solution for the main problem and cost at most $ (2 + \eps)\cdot \opt $, proving our claim.
    
    If we don't have $ \opt \ge l_h $, it must be the case that there is a point on a horizontal segment on one of the vertical sides of $ B $. This implies that $ \OPT $ must specifically contain that point. Similar to above, we can add that point to the set of vertical segments and the set of horizontal segments separately; we then find solutions using Theorem \ref{thm:main} with parameter $ \eps/2 $. Combining those two solutions will yield the same result.
\end{proof}

It is straightforward to extend the result  of this theorem to a $(2+\eps)$-approximation for when we have axis parallel line segments of size in $[1,\lambda]$, that runs in time $O(n^{\lambda/\eps^3})$ just as in proof of Theorem \ref{thm:main}.

There are several interesting open questions left. First, is Insight 2 actually correct or not? In other words, are there instances for which the shadow of an optimum solution when restricted to a bounded strip is {\em not} bounded? We have not been able to construct an explicit example extending the picture of Figure \ref{fig:insight}.
Another question is whether one can get a PTAS for the case of axis-parallel line segments or, more generally, when the line segments have only a bounded number of possible directions.

\bibliography{references}

\end{document}